\documentclass[a4paper]{article}

\usepackage{a4wide}
\usepackage[UKenglish]{babel}
\usepackage{graphicx,subfigure,pgfplots}
	\pgfplotsset{compat=newest}
	\usepgfplotslibrary{fillbetween}
	\usepgfplotslibrary{groupplots}
	\usetikzlibrary{patterns}
	\usetikzlibrary{matrix}
\usepackage{epstopdf}
	\graphicspath{{./}{figure/}}
\usepackage[colorlinks=true,linkcolor=blue,citecolor=red]{hyperref}
\usepackage{xcolor}
\usepackage{amsfonts,amsmath,amssymb,amsthm,mathtools,bbold}
\usepackage{float}
\usepackage[inline]{enumitem}

\newtheorem{theorem}{Theorem}[section]
\newtheorem{proposition}[theorem]{Proposition}
\theoremstyle{remark}\newtheorem{remark}[theorem]{Remark}

\DeclareMathOperator*{\argmin}{arg\,min}
\newcommand{\abs}[1]{\left\lvert#1\right\rvert}
\newcommand{\ave}[1]{\left\langle#1\right\rangle}
\newcommand{\Beta}{\operatorname{B}}
\newcommand{\cI}{\mathcal{I}}
\newcommand{\ucontr}{\mathfrak{u}}
\newcommand{\E}{\mathbb{E}}
\newcommand{\F}{\mathcal{F}}
\newcommand{\cL}{\mathcal{L}}
\newcommand{\cP}{\mathcal{P}}
\newcommand{\pr}[1]{{}^\prime\!#1}
\newcommand{\R}{\mathbb{R}}
\newcommand{\rhomax}{\rho_\mathrm{max}}
\newcommand{\cU}{\mathcal{U}}
\newcommand{\Var}{\operatorname{Var}}
\newcommand{\vmax}{v_\mathrm{max}}

\allowdisplaybreaks

\begin{document}
%\title{Hydrodynamics for microscopically controlled traffic flow}
%\title{Kinetic derivation of hydrodynamic traffic models with driver-assist vehicles}
\title{Kinetic-controlled hydrodynamics for traffic models with driver-assist vehicles}

\author{Andrea Tosin, Mattia Zanella \\
		{\small Department of Mathematical Sciences ``G. L. Lagrange''} \\
		{\small Politecnico di Torino, Torino, Italy} \\
		{\small\tt andrea.tosin@polito.it, mattia.zanella@polito.it}} % \\[5mm]
%		Mattia Zanella \\
%		{\small	Department of Mathematical Sciences ``G. L. Lagrange''} \\
%		{\small Dipartimento di Eccellenza 2018-2022} \\
%		{\small Politecnico di Torino, Torino, Italy} \\
%		{\small\tt mattia.zanella@polito.it}}
\date{}

\maketitle

\begin{abstract}
We develop a hierarchical description of traffic flow control by means of driver-assist vehicles aimed at the mitigation of speed-dependent road risk factors. Microscopic feedback control strategies are designed at the level of vehicle-to-vehicle interactions and then upscaled to the global flow via a kinetic approach based on a Boltzmann-type equation. Then first and second order hydrodynamic traffic models, which naturally embed the microscopic control strategies, are consistently derived from the kinetic-controlled framework via suitable closure methods. Several numerical examples illustrate the effectiveness of such a hierarchical approach at the various scales.

\medskip

\noindent{\bf Keywords:} Kinetic modelling, binary control, hydrodynamic equations, road risk mitigation \\

\noindent{\bf Mathematics Subject Classification:} 35Q20, 35Q70, 35Q84, 35Q93, 49J20, 90B20
\end{abstract}

\section{Introduction}
In the last two decades the legacy of the classical kinetic theory has emerged as a sound mathematical paradigm for the description of collective phenomena involving a large number of agents such as socio-economic~\cite{naldi2010BOOK,pareschi2013BOOK} and traffic~\cite{agnelli2015M3AS,festa2018KRM,freguglia2017CMS,herty2005CMAM,herty2010KRM,illner2003CMS,klar1997JSP,puppo2016CMS,visconti2017MMS} dynamics. In particular, the mathematical modelling of vehicular traffic by methods of the kinetic theory has a quite long history dating back to the pioneering works~\cite{paveri1975TR,prigogine1971BOOK}. One of the main reasons for such a success is the multiscale flexibility of the kinetic equations, which bridge organically the gap between the microscopic, often unobservable, scale of the individual agents, where elementary fundamental dynamics take place, and the macroscopic scale of the observable collective manifestations. This confers on the kinetic approach a great explanatory power about the way in which multi-agent systems work. At the same time, the possibility to derive hydrodynamic descriptions of those systems consistent with microscopic interaction dynamics is of paramount importance for designing fast numerical methods which possibly help decision-making tasks.

Euler or Navier-Stokes-type equations are classically obtained from collisional kinetic equations by means of suitable closure procedures based on the relaxation of the system towards its equilibria. For instance, in the context of rarefied gas dynamics they are derived taking advantage of the microscopic conservations of mass, momentum and energy in the binary collisions between gas molecules, which allow one to identify the Maxwellian asymptotic distribution, see e.g.~\cite{benedetto2004CMS,cercignani1994BOOK}. The derivation of hydrodynamic equations in the non-classical setting of multi-agent systems is instead a currently underexplored topic due to the general lack of information about the asymptotic statistical behaviour of the system. Some quite recent works in this direction are~\cite{carrillo2017BOOKCH,during2007PHYSA,ha2008KRM}. In this work we focus on the case of vehicular traffic, in which only the total mass of vehicles is conserved by the microscopic interactions.

In recent times the challenge of vehicular automation has posed new and exciting questions about traffic management and governance, which in turn boosted broad developments in the technology for intelligent intersections, driver-assist and self-driving vehicles~\cite{santi2014PNAS,tachet2016PLOS1}. One of the main goals of such technologies is the enhancement of driver safety through the mitigation of road risk factors, which, as reported in~\cite{WHO2015report}, are largely linked to the heterogeneity of the individual driving behaviour. Among others, here we recall in particular those related to the variability of the speed in the traffic flow: large differences in the speeds of the vehicles within the traffic stream appear to be responsible for a sensible increase in the crash risk. The idea which is progressively gaining ground in this context is to exploit the possibility to \textit{control} a few automated vehicles in order to induce a regularisation of the whole traffic flow~\cite{stern2018TRC}, which should then mitigate the aforesaid risk factors.

The control of multi-agent systems has been initiated relatively recently as a natural follow-up of the description and modelling of their self-organisation features. Contributions are available for mean field and kinetic equations~\cite{albi2017AMO,bensoussan2013SPRINGERBRIEFS,fornasier2014PTRSA} as well as for macroscopic conservation laws~\cite{banda2013MCRF,colombo2011COCV,cristiani2015SIAP}. In this paper we tackle the control of driver-assist vehicles taking inspiration from a general method called \textit{Model Predictive Control} (MPC). MPC is based on determining the control by optimising a given cost functional of the agents (in our case the driver-assist vehicles) over a finite, rather than infinite, short time horizon, which recedes as time evolves. In particular, we will consider the cost functionals recently proposed in~\cite{tosin2018IFAC}, which provide measures of the driving risk in terms of the speed differences of the vehicles, and we will assume that the time horizon for their minimisation coincides with the duration of a single binary interaction between a driver-assist vehicle and its leading vehicle. Consequently the control is recomputed each time that the driver-assist vehicle interacts with a new leading vehicle. Interestingly, this leads mathematically to a \textit{binary} control problem which can be solved explicitly. The result is a feedback control given in terms of the microscopic states of the interacting vehicles, which can be embedded straightforwardly in a Boltzmann-type kinetic equation. In particular, considering that a randomly chosen vehicle may be equipped with driver-assist technologies with a certain probability $p$, the so-called \textit{penetration rate}, this setting allows us to consider naturally \textit{sparse} control problems. The MPC strategy was first introduced in the engineering literature~\cite{camacho2007BOOK,sontag1998BOOK}, where traditionally it has been used in connection with ordinary differential equations. Very few results are currently available for other types of differential models, see e.g.~\cite{albi2017IFAC,albi2015CMS,albi2014PTRSA}. Concerning optimality, it is well known that MPC leads typically to controls which are suboptimal compared to the theoretically optimal control computed on the (possibly infinite) global time domain of the problem. Nevertheless performance bounds have been established, which guarantee the consistency of the MPC approximation also in the kinetic framework~\cite{grune2009SICON,herty2017DCDS}. In addition to this, the computational cost of the proposed Boltzmann formulation of MPC scales linearly with the total number of vehicles of the system. This makes it competitive with respect to other techniques based e.g., on the control of mean field equations.

In more detail, the paper is organised as follows. In Section~\ref{sect:homog} we introduce a space homogeneous kinetic model of human-manned traffic able to explain how fundamental diagrams and statistical speed distributions consistent with the empirical observations are generated by simple microscopic interactions between pairs of vehicles. The analysis of this case provides insights into the normal flow of uncontrolled vehicles, thereby constituting the reference for all the subsequent developments. In Section~\ref{sect:micro.control} we design and solve the binary control problem for driver-assist vehicles, taking into account their penetration rate into the traffic stream, and we repeat the space homogeneous kinetic analysis so as to assess the effectiveness of the control strategies in reducing the road risk with respect to the previous case of fully human-manned vehicles. In Section~\ref{sect:hydro} we consider a space inhomogeneous kinetic description, whence we consistently derive first and second order hydrodynamic traffic equations with embedded microscopic control via the local equilibrium closure and the monokinetic assumption, respectively. In Section~\ref{sect:num} we extensively investigate the solutions produced by the kinetic and hydrodynamic models by means of numerical simulations specifically focused on the collective impact of the control strategies and on the influence of the penetration rate of the driver-assist technology. Finally, in Section~\ref{sect:conclusions} we summarise the highlights of the work and draw some conclusions.

\section{Homogeneous kinetic modelling of traffic flow}
\label{sect:homog}
In this section we study interactive vehicle dynamics which explain how typical speed distributions and traffic diagrams arise in stationary flow conditions. This will be the basis for constructing later control strategies for driver-assist vehicles aimed at making traffic flow more uniform by essentially reducing the speed variance among the vehicles.

The literature accounts nowadays for a number of empirical and theoretical investigations of the so-called \textit{speed} and \textit{fundamental diagrams} of traffic. These are relationships linking the mean speed $V$ and the macroscopic flux $\rho V$ of the vehicles to the traffic density $\rho$ in stationary homogeneous conditions along the main longitudinal direction of the flow. From the empirical point of view, the common observation to all measured traffic diagrams is that the mean speed is nearly constant and close to the maximum speed in the free flow regime, i.e. for $\rho$ sufficiently small; conversely, it decreases steeply to zero in the congested flow regime, i.e. for $\rho$ approaching the maximum possible density $\rhomax>0$. In turn, the macroscopic flux grows almost linearly with the density in the free flow regime, then it decreases non-linearly to zero in the congested flow regime. The two traffic regimes are separated by a critical value of $\rho$, called the density at capacity, where the macroscopic flux is maximum, see e.g.~\cite{leutzbach1988BOOK}. A further intermediate regime might exist, called the synchronised traffic regime, in which vehicles tend to travel all at the same speed, cf.~\cite{kerner2004BOOK}. From the theoretical point of view, a few mathematical models have been able to explain the emergence of such large scale characteristics of traffic from a microscopic description of vehicle interactions~\cite{fermo2014DCDSS,gunther2003SIAP}. In some cases, models have also successfully investigated the origin of the data scattering typically seen in measured traffic diagrams~\cite{puppo2016CMS,visconti2017MMS}. Finally, very recently diagrams in the transversal direction of the flow produced by the lateral displacements of lane-changing vehicles have started to be measured and their origin investigated by means of tools of statistical physics~\cite{herty2017SIAP_preprint}.

By far less studied is instead the statistical distribution of the microscopic speeds of the vehicles in similar stationary homogeneous conditions. This information is nonetheless fundamental for assessing traffic features correlated with the road safety, such as e.g. the dispersion of the speeds of the vehicles, which has been reported as one of the major causes of the increase in the crash risk~\cite{peden2004WHO,WHO2015report}. Moreover, the speed distribution at equilibrium may play a role analogous to the Maxwellian distribution in the kinetic theory of gases for the theoretical derivation of macroscopic hydrodynamic traffic models consistent with microscopic models of vehicle interactions. A typical claim in the literature is that the speed in highway-like traffic can be approximated by a Gaussian distribution~\cite{berry1951PROCEEDINGS}. Some studies suggest in particular that this may be valid in free and congested traffic regimes, while in the intermediate regime the approximation by means of a log-normal distribution performs better~\cite{abuelenin2014IEEE}. Other studies have found that bimodal distributions may be more appropriate to describe the speeds in a mixture of different categories of vehicles~\cite{dey2006JTE}. A drawback of the Gaussian distribution is that it is not compactly supported, whereas vehicle speeds normally vary in an interval of the form $[0,\,\vmax]$, where $\vmax>0$ is some maximum speed. Consequently if, on one hand, the Gaussian curve can be a healthy empirical approximation, on the other hand more accurate distributions need to be sought from the theoretical point of view. In~\cite{maurya2016TRP,ni2018AMM} the authors find a good agreement between empirical traffic speed curves and beta distributions. In the following we demonstrate that beta distributions can indeed be obtained as equilibrium distributions of a kinetic traffic model, starting from simple and very reasonable assumptions on the interactions between pairs of vehicles. Interestingly, our approach allows us to recover the relevant parameters of such distributions in terms of the traffic density $\rho$. This provides organically average statistical quantities parametrised by $\rho$, such as the mean speed and the macroscopic flux, which turn out to compare qualitatively well with the experimental diagrams of traffic described above.

\subsection{Boltzmann-type model and traffic diagrams}
Inspired by classical methods of kinetic theory, we introduce the distribution function $f=f(t,\,v)$ such that $f(t,\,v)\,dv$ is the fraction of vehicles which, at time $t>0$, are travelling with a speed comprised between $v$ and $v+dv$. We understand all variables as dimensionless and in particular we set $v,\,\rho\in [0,\,1]$, meaning that we have normalised the speed and the density by their maximum values $\vmax$, $\rhomax$, respectively.

In the Boltzmann-type kinetic approach one assumes that the time evolution of $f$ is determined by microscopic stochastic processes consisting in binary (i.e., pairwise) interactions responsible for speed changes. If $v,\,w$ are the pre-interaction speeds of any two representative vehicles and $v'$, $w'$ their post-interaction speeds, a binary interaction takes the form of a rule expressing $v'$, $w'$ as functions of $v$, $w$:
\begin{align}
	\begin{aligned}[c]
		v' &= v+\gamma I(v,\,w;\,\rho)+D(v;\,\rho)\eta \\
		w' &= w.
	\end{aligned}
	\label{eq:binary}
\end{align}
Consistently with a microscopic follow-the-leader approach~\cite{gazis1961OR}, we assume that a vehicle of the pair, in this case the one with speed $w$, plays the role of the leading vehicle. Since in vehicular traffic interactions are mainly anisotropic, and particularly frontal, the leading vehicle is unaffected by the rear vehicle, which here is the one with speed $v$. In contrast, the rear vehicle may change speed according to the interaction function $I(v,\,w;\,\rho)$, which may depend on the traffic density $\rho$ because traffic congestion may affect acceleration and deceleration. In addition to that, in the equation for $v'$ the constant $\gamma>0$ is a proportionality parameter and $\eta\in\R$ is a centred random variable, i.e. one with mean $\ave{\eta}=0$, modelling stochastic fluctuations of the post interaction speed $v'$. The variance of $\eta$ is set to $\ave{\eta^2}=:\sigma^2>0$ while the intensity of the stochastic fluctuation is tuned by the function $D(v;\,\rho)\geq 0$ representing the local relevance of the diffusion.

In order to write a Boltzmann-type equation ruling the evolution of $f$ we argue like in~\cite{pareschi2013BOOK}.
%\textcolor{blue}{The evolution in time of the density function $f$ is then described by the following integro-differential equation 
%\begin{equation}\label{eq:Boltzmann_strong}
%\partial_t f(t,v) = Q(f,f)(t,v),
%\end{equation}
%where $Q(\cdot,\cdot)$ is the binary interaction operator defined as follows
%\begin{equation}\label{eq:Q_Boltzmann_strong}
%Q(f,f)(t,v) = \ave{ \int_0^1 \left({}^\prime B\dfrac{1}{J} f(t,{}^\prime v)f(t,{}^\prime v_*)- B\, f(t,v)f(t,v_*) \right)dv_* },
%\end{equation}
%where $({}^\prime v,{}^\prime v_*)$ are the pre-interaction speeds of the couple $(v,v_*)$, the term $J$ is the Jacobian of the transformation $(v,v_*) \rightarrow (v^\prime,v_*^\prime)$ and the kernels ${}^\prime B,B$ define the binary interactions. Moreover, $\ave{\cdot}$ denotes the average with respect to the probability distribution of $\eta$. As usual, we can express \eqref{eq:Boltzmann_strong} in weak form by introducing the test function $\varphi$.
%}
If $\varphi=\varphi(v)$ is any observable quantity which can be expressed as a function of the speed $v$ then the time variation of the expectation of $\varphi(v)$ is due, on average, to the variation of $\varphi$ in a binary interaction. In formulas this writes:
\begin{equation}
	\frac{d}{dt}\int_0^1\varphi(v)f(t,\,v)\,dv=\frac{1}{2}\ave{\int_0^1\int_0^1(\varphi(v')-\varphi(v)+\varphi(w')-\varphi(w))f(t,\,v)f(t,\,w)\,dv\,dw},
	\label{eq:Boltzmann.spathom}
\end{equation}
where the coefficient $\frac{1}{2}$ at the right-hand side is used to average the variations $\varphi(v')-\varphi(v)$ and $\varphi(w')-\varphi(w)$ in a binary interaction (notice from~\eqref{eq:binary} that actually $\varphi(w')-\varphi(w)=0$) and $\ave{\cdot}$ denotes the expectation with respect to the probability distribution of $\eta$.

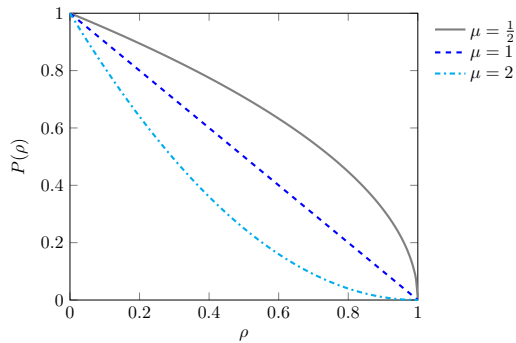
\begin{figure}[!t]
\centering
\resizebox{7cm}{!}{
\begin{tikzpicture}[
declare function={
	P(\x,\mu)=(1-\x)^\mu;
} 
]
\begin{axis}[
	xlabel={$\rho$},
	ylabel={$P(\rho)$},
	xmin=0,
	xmax=1,
	ymin=0,
	ymax=1,
	legend style={draw=none},
	legend cell align=left,
	legend pos=outer north east,
]
	\addplot[gray,domain=0:1,smooth,very thick,style={solid},samples=200]{P(x,0.5)}; \addlegendentry{$\mu=\frac{1}{2}$};
	\addplot[blue,domain=0:1,smooth,very thick,style={dashed}]{P(x,1)}; \addlegendentry{$\mu=1$}
	\addplot[cyan,domain=0:1,smooth,very thick,style={dash dot}]{P(x,2)}; \addlegendentry{$\mu=2$}
\end{axis}
\end{tikzpicture}
}
\caption{The probability of accelerating $P(\rho)$ given in~\eqref{eq:P} plotted for various $\mu$.}
\label{fig:P}
\end{figure}

We discuss now in more detail the structure of the functions $I$, $D$ in~\eqref{eq:binary}. The interaction function $I$ has often been modelled in the literature by considering separately the cases $v\leq w$, which would induce an acceleration of the $v$-vehicle, and $v>w$, which instead would induce a deceleration, see e.g.~\cite{delitala2007M3AS,herty2010KRM,puppo2016CMS}. Here we propose a simpler form, which has the merit of making the whole model more tractable analytically while giving rise to many interesting physical consequences. Specifically, we set:
\begin{equation}
	I(v,\,w;\,\rho):=P(\rho)(1-v)+(1-P(\rho))(P(\rho)w-v),
	\label{eq:I}
\end{equation}
where
\begin{equation}
	P(\rho):=(1-\rho)^\mu, \quad \mu>0,
	\label{eq:P}
\end{equation}
is the probability of accelerating, see Figure~\ref{fig:P}. This interaction function says that in light traffic, i.e. for $\rho$ small, the $v$-vehicle basically relaxes its speed towards the maximum possible one, cf. the first term at the right-hand side of~\eqref{eq:I}. When $\rho$ increases the $v$-vehicle starts to adapt its speed also to a fraction $P(\rho)$ of the speed $w$ of the leading vehicle, cf. the second term at the right-hand side of~\eqref{eq:I}. The fraction is expressed by the function $P(\rho)$ itself, in such a way that the more congested the traffic the smaller the target speed $P(\rho)w$ towards which the $v$-vehicle relaxes. In heavy traffic, i.e. for $\rho$ large, this mechanism leads typically the $v$-vehicle to decelerate.

The choice of the function $D$, instead, has to be made taking into account the necessity to guarantee that the post-interaction speed $v'$ in~\eqref{eq:binary} complies with the bounds $0\leq v'\leq 1$. Concerning this, we establish first of all the following result:
\begin{proposition} \label{prop:eta.D}
In~\eqref{eq:binary}, let $\gamma\in [0,\,1]$ and $I(v,\,w;\,\rho)$ be given by~\eqref{eq:I}. If there exists $c>0$ such that
\begin{equation*}
	\begin{cases}
		\abs{\eta}\leq c(1-\gamma) \\[1mm]
		cD(v;\,\rho)\leq\min\{v,\,1-v\}, & \forall\,v,\,\rho\in [0,\,1]
	\end{cases}
\end{equation*}
then $v'\in [0,\,1]$ for every pair of pre-interaction speeds $v,\,w\in [0,\,1]$ and every $\rho\in [0,\,1]$.
\end{proposition}
\begin{proof}
To prove that $v'\geq 0$ we observe that, since $\gamma P(\rho),\,P(\rho)w\geq 0$, it is enough to show that $v-\gamma(P(\rho)v+(1-P(\rho))v)+D(v;\,\rho)\eta\geq 0$, i.e. that $(1-\gamma)v+D(v;\,\rho)\eta\geq 0$. Since, by assumption, $\eta\geq c(\gamma-1)$, with $\gamma-1\leq 0$, and $D(v;\,\rho)\leq\frac{v}{c}$ we discover $(1-\gamma)v+D(v;\,\rho)\eta\geq (1-\gamma)v+\frac{v}{c}c(\gamma-1)=0$ and we have the result.

Conversely, to prove that $v'\leq 1$ we observe that, since $P(\rho)w\leq 1$, it is sufficient to establish that $v+\gamma(P(\rho)(1-v)+(1-P(\rho))(1-v))+D(v;\,\rho)\eta\leq 1$, i.e. that $(\gamma-1)(1-v)+D(v;\,\rho)\eta\leq 0$. But $(\gamma-1)(1-v)+D(v;\,\rho)\eta\leq(\gamma-1)(1-v)+\frac{1-v}{c}c(1-\gamma)=0$, thus we are done.
\end{proof}

Conditions posed by Proposition~\ref{prop:eta.D} imply that $\eta$ is a compactly supported random variable in the interval $[-c(1-\gamma),\,c(1-\gamma)]$, which is in particular compatible with the fact that $\ave{\eta}=0$, and that $D(0;\,\rho)=D(1;\,\rho)=0$ for all $\rho\in [0,\,1]$. We defer to the next Section~\ref{sect:asympt} a more specific choice of $D$.

For the moment we observe that by choosing $\varphi(v)=1$ in~\eqref{eq:Boltzmann.spathom} we obtain
$$ \frac{d}{dt}\int_0^1f(t,\,v)\,dv=0. $$
This implies that if the initial speed distribution $f_0(v):=f(0,\,v)$ fulfils the normalisation condition $\int_0^1f_0(v)\,dv=1$ then $f$ is a probability density function for all $t>0$. Choosing instead $\varphi(v)=v$ in~\eqref{eq:Boltzmann.spathom} and using the interaction rules~\eqref{eq:binary},~\eqref{eq:I} we discover that the mean speed
$$ V(t):=\int_0^1vf(t,\,v)\,dv $$
satisfies the equation
\begin{equation}
	\frac{dV}{dt}=\frac{\gamma}{2}\left\{P(\rho)\left[1+(1-P(\rho))V\right]-V\right\},
	\label{eq:dV/dt}
\end{equation}
whose solution writes
\begin{equation}
	V(t)=V_0e^{-\frac{\gamma}{2}\left[P(\rho)+(1-P(\rho))^2\right]t}+
		\frac{P(\rho)}{P(\rho)+(1-P(\rho))^2}\left(1-e^{-\frac{\gamma}{2}\left[P(\rho)+(1-P(\rho))^2\right]t}\right)
	\label{eq:V}
\end{equation}
with $V_0:=\int_0^1vf_0(v)\,dv$ the mean speed at the initial time. For $t\to+\infty$, $V$ approaches exponentially fast the asymptotic value
\begin{equation}
	V_\infty(\rho):=\frac{P(\rho)}{P(\rho)+(1-P(\rho))^2},
	\label{eq:Vinf}
\end{equation}
which defines the speed diagram of traffic. The mapping $\rho\mapsto\rho V_\infty(\rho)$ defines instead the fundamental diagram. Figure~\ref{fig:funddiag} shows that these curves agree well with the qualitative empirical characteristics of the traffic diagrams discussed at the beginning of Section~\ref{sect:homog}, especially for $\mu>1$ in the expression~\eqref{eq:P} of $P(\rho)$.

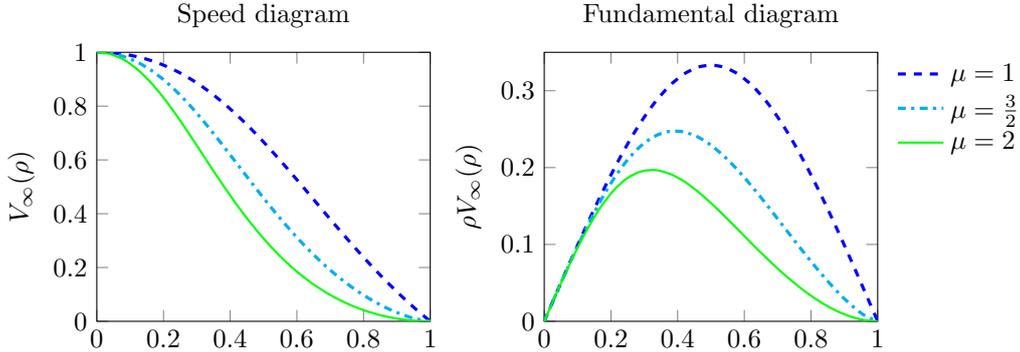
\begin{figure}[!t]
\centering
\begin{tikzpicture}[
declare function={
	P(\x,\mu)=(1-\x)^\mu;
	Vinf(\x,\mu)=P(\x,\mu)/(P(\x,\mu)+(1-P(\x,\mu))^2);
} 
]
\begin{groupplot}[
	group style={group size= 2 by 1,horizontal sep=15mm},
	width=0.4\textwidth,
	legend cell align=left,
	legend pos=outer north east,
	legend style={draw=none}
]
\nextgroupplot[ylabel=$V_\infty(\rho)$,xmin=0,xmax=1,ymin=0,ymax=1,title=Speed diagram] %grid=major,
	\addplot[blue,domain=0:1,smooth,very thick,style={dashed}]{Vinf(x,1)};
	\addplot[cyan,domain=0:1,smooth,very thick,style={dash dot}]{Vinf(x,1.5)};
	\addplot[green,domain=0:1,smooth,thick,style={solid}]{Vinf(x,2)};
\nextgroupplot[ylabel=$\rho V_\infty(\rho)$,xmin=0,xmax=1,ymin=0,ymax=0.35,title=Fundamental diagram] %grid=major,
	\addplot[blue,domain=0:1,smooth,very thick,style={dashed}]{x*Vinf(x,1)}; \addlegendentry{$\mu=1$}
	\addplot[cyan,domain=0:1,smooth,very thick,style={dash dot}]{x*Vinf(x,1.5)}; \addlegendentry{$\mu=\frac{3}{2}$}
	\addplot[green,domain=0:1,smooth,thick,style={solid}]{x*Vinf(x,2)}; \addlegendentry{$\mu=2$}
\end{groupplot}
\end{tikzpicture}
\caption{Traffic diagrams plotted from~\eqref{eq:Vinf} with different choices of the exponent $\mu$ in~\eqref{eq:P}.}
\label{fig:funddiag}
\end{figure}

\subsection{Asymptotic speed distribution}
\label{sect:asympt}
In order to compute the equilibrium speed distribution in homogeneous conditions one should find the asymptotic solutions of the Boltzmann-type equation~\eqref{eq:Boltzmann.spathom}, namely the probability density functions $f_\infty=f_\infty(v)$ independent of $t$ which make the right-hand side of~\eqref{eq:Boltzmann.spathom} vanish. They correspond to speed distributions which create an equilibrium of the binary interactions.

Unfortunately, although~\eqref{eq:Boltzmann.spathom} is suited to investigate the statistical moments of the distribution function $f$, it is in general not easy to obtain from it a pointwise description of $f$ itself and of its asymptotic trends. The reason is that~\eqref{eq:Boltzmann.spathom} is a high-resolution equation in time, i.e. one which catches the detail of every single binary interaction. The large-time trends of $f$ may however be successfully recovered by means of asymptotic procedures, which transform~\eqref{eq:Boltzmann.spathom} in simpler kinetic equations whose solutions approximate well the steady profiles of the asymptotic distributions of~\eqref{eq:Boltzmann.spathom}. One such procedure is the \textit{quasi-invariant interaction limit} introduced in~\cite{toscani2006CMS}, which is reminiscent of the \textit{grazing collision limit} of the classical kinetic theory of gases~\cite{villani1998ARMA} and which leads to \textit{Fokker-Planck-type} asymptotic equations. Here is how the procedure works.

Assume that we consider the binary interactions~\eqref{eq:binary} on a time scale $\tau$ much larger than their characteristic time scale $t$. Then on the $\tau$-scale the contribution of a single interaction~\eqref{eq:binary} is small (i.e., interactions are \textit{quasi-invariant}) but, at the same time, interactions are quite frequent. This can be formalised by defining $\tau:=\frac{\gamma}{2}t$ and then assuming that $\gamma$ and $\sigma^2$ (the variance of $\eta$) are small in~\eqref{eq:binary}. Notice that if the characteristic frequency of a binary interaction is $1/t=O(1)$ in the $t$-scale then it becomes $1/\tau=O(1/\gamma)\gg 1$ in the $\tau$-scale.

Let us introduce now the scaled distribution function $\tilde{f}(\tau,\,v):=f(2\tau/\gamma,\,v)$. We observe that, for every fixed $\tau>0$, $\gamma$ small implies $t=2\tau/\gamma$ large, hence the limit $\gamma\to 0^+$ describes the large-time trend of $f$. On the other hand, since for $\tau\to+\infty$ it results also $t\to+\infty$, the asymptotic trend of $\tilde{f}$ approximates well that of $f$. The idea is then to find from~\eqref{eq:Boltzmann.spathom} an equation satisfied by $\tilde{f}$ in the quasi-invariant interaction limit $\gamma,\,\sigma^2\to 0^+$, whence to study the asymptotic trend of $\tilde{f}$.

Noting that $\partial_\tau\tilde{f}=\frac{2}{\gamma}\partial_tf$, on the $\tau$-scale~\eqref{eq:Boltzmann.spathom} becomes:
\begin{equation}
	\frac{d}{d\tau}\int_0^1\varphi(v)\tilde{f}(\tau,\,v)\,dv
		=\frac{1}{\gamma}\ave{\int_0^1\int_0^1(\varphi(v')-\varphi(v))\tilde{f}(\tau,\,v)\tilde{f}(\tau,\,w)\,dv\,dw},
	\label{eq:Boltzmann.spathomog-scaled}
\end{equation}
where we have already taken into account that $\varphi(w')=\varphi(w)$ in view of the second rule in~\eqref{eq:binary}. Since for $\gamma,\,\sigma^2$ small the post-interaction speed $v'$ is close to the pre-interaction one $v$, if we assume that $\varphi$ is sufficiently smooth, namely $\varphi\in C^3([0,\,1])$, we can perform the following Taylor expansion:
$$ \varphi(v')-\varphi(v)=\varphi'(v)(v'-v)+\frac{1}{2}\varphi''(v)(v'-v)^2+\frac{1}{6}\varphi'''(\bar{v})(v'-v)^3, $$
where $\bar{v}\in (\min\{v,\,v'\},\,\max\{v,\,v'\})$ is used to express the Lagrange remainder. Writing $v'-v=\gamma I(v,\,w;\,\rho)+D(v;\,\rho)\eta$ from the first rule in~\eqref{eq:binary} and plugging into~\eqref{eq:Boltzmann.spathomog-scaled} we deduce:
\begin{align}
	\begin{aligned}[b]
		\frac{d}{d\tau}\int_0^1\varphi(v)\tilde{f}(\tau,\,v)\,dv &=
			\int_0^1\int_0^1\varphi'(v)I(v,\,w;\,\rho)\tilde{f}(\tau,\,v)\tilde{f}(\tau,\,w)\,dv\,dw \\
		&\phantom{=} +\frac{\sigma^2}{2\gamma}\int_0^1\varphi''(v)D^2(v;\,\rho)\tilde{f}(\tau,\,v)\,dv+R_\varphi(\tilde{f},\,\tilde{f}),
	\end{aligned}
	\label{eq:Boltzmann.spathomog-scaled.R}
\end{align}
where $R_\varphi(\tilde{f},\,\tilde{f})$ denotes the following remainder:
\begin{align*}
	R_\varphi(\tilde{f},\,\tilde{f}) &:= \frac{\gamma}{2}\int_0^1\int_0^1\varphi''(v)I^2(v,\,w;\,\rho)\tilde{f}(\tau,\,v)\tilde{f}(\tau,\,w)\,dv\,dw \\
	&\phantom{:=} +\frac{1}{6\gamma}\ave{\int_0^1\int_0^1\varphi'''(\bar{v})\left(\gamma I(v,\,w;\,\rho)+D(v;\,\rho)\eta\right)^3
		\tilde{f}(\tau,\,v)\tilde{f}(\tau,\,w)\,dv\,dw}.
\end{align*}
From~\eqref{eq:I} we check that $\abs{I(v,\,w;\,\rho)}\leq 1$ for all $v,\,w,\,\rho\in [0,\,1]$ and, in addition to this, from Proposition~\ref{prop:eta.D} we infer that also $D(v;\,\rho)$ has to be bounded. Moreover, $\varphi$ and its derivatives are in turn bounded in $[0,\,1]$ because of the assumed smoothness. Finally, if we assume that $\eta$ has the third order moment bounded, i.e. $\langle\abs{\eta}^3\rangle<+\infty$, then we can write $\eta=\sqrt{\sigma^2}\tilde{\eta}$ where $\tilde{\eta}$ is a random variable such that $\ave{\tilde{\eta}}=0$, $\langle\tilde{\eta}^2\rangle=1$ and $\langle\abs{\tilde{\eta}}^3\rangle<+\infty$, so that $\langle\abs{\eta}^3\rangle\sim(\sigma^2)^{3/2}$. As a result, we estimate\footnote{We use the notation $a\lesssim b$ to mean that there exists a constant $C>0$ such that $a\leq Cb$.}
\begin{equation}
	\abs{R_\varphi(\tilde{f},\,\tilde{f})}\lesssim\gamma+\gamma^2+\sigma^2+\frac{\sigma^2}{\gamma}\sqrt{\sigma^2}.
	\label{eq:R}
\end{equation}
At this point, in taking the quasi-invariant interaction limit $\gamma,\,\sigma^2\to 0^+$ we need to specify the behaviour of the ratio $\sigma^2/\gamma$. Assuming that $\sigma^2/\gamma\to\lambda>0$, so that the effects of the interactions and of the stochastic fluctuations balance asymptotically, we get $R_\varphi(\tilde{f},\,\tilde{f})\to 0$ from~\eqref{eq:R} and consequently from~\eqref{eq:Boltzmann.spathomog-scaled.R}
\begin{align}
	\begin{aligned}[b]
		\frac{d}{d\tau}\int_0^1\varphi(v)\tilde{f}(\tau,\,v)\,dv &=
			\int_0^1\varphi'(v)\left(\int_0^1I(v,\,w;\,\rho)\tilde{f}(\tau,\,w)\,dw\right)\tilde{f}(\tau,\,v)\,dv \\
		&\phantom{=} +\frac{\lambda}{2}\int_0^1\varphi''(v)D^2(v;\,\rho)\tilde{f}(\tau,\,v)\,dv.
	\end{aligned}
	\label{eq:FP.weak}
\end{align}

In view of the smoothness of $\varphi$, integrating back by parts the terms at the right-hand side this can be recognised as a weak form of the Fokker-Planck equation
\begin{equation}
	\partial_\tau\tilde{f}=\frac{\lambda}{2}\partial^2_v(D^2(v;\,\rho)\tilde{f})
		-\partial_v\left(\left(\int_0^1I(v,\,w;\,\rho)\tilde{f}(\tau,\,w)\,dw\right)\tilde{f}\right)
	\label{eq:FP}
\end{equation}
provided the following boundary conditions are satisfied:
\begin{align}
	\begin{aligned}[c]
		\left(\int_0^1I(v,\,w;\,\rho)\tilde{f}(\tau,\,w)\,dw\right)\tilde{f}(\tau,\,v)
			+\dfrac{\lambda}{2}\partial_v\left(D^2(v;\,\rho)\tilde{f}(\tau,\,v)\right) &= 0 \\
			D^2(v;\,\rho)\tilde{f}(\tau,\,v) &= 0
	\end{aligned}
	\label{eq:FP-bc}
\end{align}
for $v=0,\,1$ and all $\tau>0$. In particular, substituting in~\eqref{eq:FP} the interaction function $I$ given in~\eqref{eq:I} yields
$$ \partial_\tau\tilde{f}=\frac{\lambda}{2}\partial^2_v(D^2(v;\,\rho)\tilde{f})
	-\partial_v\biggl(\Bigl(P(\rho)\bigl(1+(1-P(\rho))\tilde{V}\bigr)-v\Bigr)\tilde{f}\biggr), $$
where $\tilde{V}(\tau):=V(2\tau/\gamma)$ and $V$ is the mean speed~\eqref{eq:V}. For $\tau\to+\infty$ the term $\tilde{V}$ converges to $V_\infty(\rho)$ given in~\eqref{eq:Vinf}, hence the asymptotic distribution $\tilde{f}_\infty$ satisfies the equation
$$ \frac{\lambda}{2}\partial^2_v(D^2(v;\,\rho)\tilde{f}_\infty)
		-\partial_v\biggl(\Bigl(P(\rho)\bigl(1+(1-P(\rho))V_\infty(\rho)\bigr)-v\Bigr)\tilde{f}_\infty\biggr)=0. $$
Since $P(\rho)(1+(1-P(\rho))V_\infty(\rho))=V_\infty(\rho)$, as it can be immediately checked from~\eqref{eq:dV/dt} or by a direct calculation using~\eqref{eq:Vinf}, this further simplifies into
$$ \frac{\lambda}{2}\partial^2_v(D^2(v;\,\rho)\tilde{f}_\infty)-\partial_v((V_\infty(\rho)-v)\tilde{f}_\infty)=0, $$
whose solution reads
\begin{equation}
	\tilde{f}_\infty(v)=\frac{C}{D^2(v;\,\rho)}\exp\left(\frac{2}{\lambda}\int\frac{V_\infty(\rho)-v}{D^2(v;\,\rho)}\,dv\right),
	\label{eq:finf.gen}
\end{equation}
where $C>0$ is a constant to ensure the normalisation $\int_0^1\tilde{f}_\infty(v)\,dv=1$.

In order to obtain from~\eqref{eq:finf.gen} a more explicit expression of $\tilde{f}_\infty$ we need to specify the diffusion coefficient $D(v;\,\rho)$. Choosing in particular
\begin{equation}
	D(v;\,\rho):=a(\rho)\sqrt{v(1-v)}, \qquad a(\rho)\geq 0,
	\label{eq:D}
\end{equation}
we get
\begin{equation}
	\tilde{f}_\infty(v)=\frac{v^{\frac{2V_\infty(\rho)}{\lambda a^2(\rho)}-1}
		(1-v)^{\frac{2(1-V_\infty(\rho))}{\lambda a^2(\rho)}-1}}
			{\Beta\left(\frac{2V_\infty(\rho)}{\lambda a^2(\rho)},\,\frac{2(1-V_\infty(\rho))}{\lambda a^2(\rho)}\right)},
	\label{eq:finf.beta}
\end{equation}
where $\Beta(x,\,y):=\int_0^1t^{x-1}(1-t)^{y-1}\,dt$ is the beta function. It can be checked that this function satisfies the boundary conditions~\eqref{eq:FP-bc} if e.g.
\begin{equation}
	a^2(\rho)\leq\frac{1}{\lambda}\min\{V_\infty(\rho),\,1-V_\infty(\rho)\},
	\label{eq:a}
\end{equation}
indeed in such a case both $\tilde{f}_\infty$ and $\partial_v\tilde{f}_\infty$ vanish at $v=0,\,1$.

\begin{figure}[!t]
\centering
\includegraphics[scale=0.8]{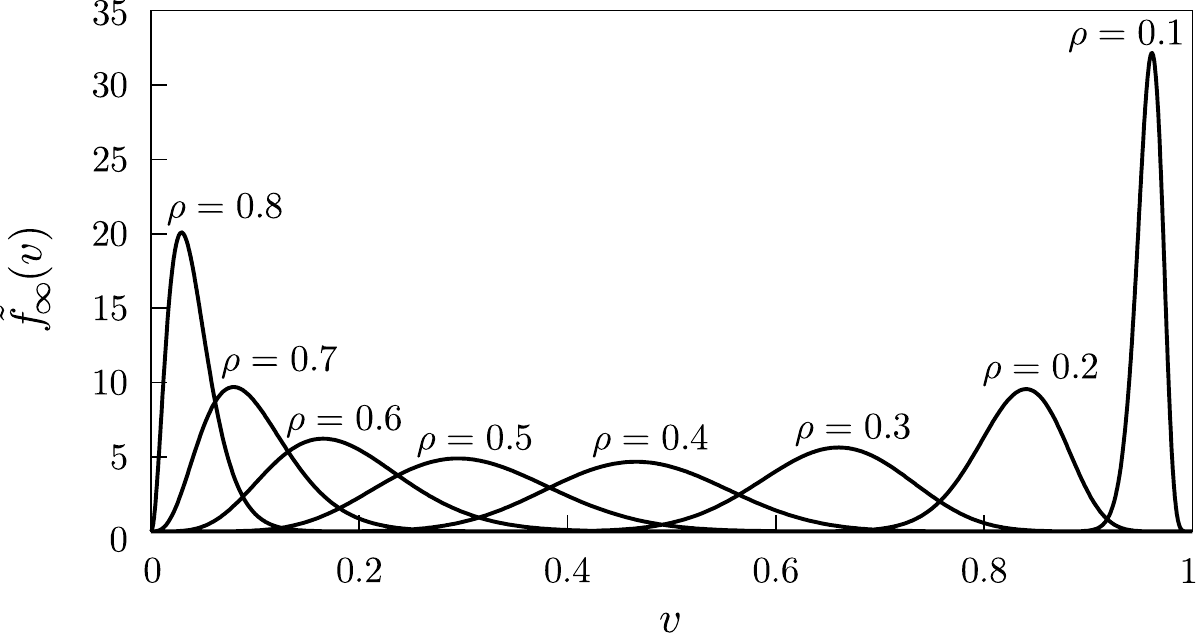}
\caption{The asymptotic speed distribution~\eqref{eq:finf.beta} for various traffic densities $\rho\in [0,\,1]$ with $V_\infty(\rho)$ given in~\eqref{eq:Vinf}, $P(\rho)$ like in~\eqref{eq:P} with $\mu=2$, $\lambda=1$ and $a(\rho)=\rho(1-\rho)$, which fulfils the constraint~\eqref{eq:a}.}
\label{fig:finf.beta}
\end{figure}

%\begin{figure}[!t]
%\centering
%\resizebox{0.8\textwidth}{10cm}{
%\begin{tikzpicture}[
%declare function={
%	P(\rho)=(1-\rho)^2;
%	Vinf(\rho)=P(\rho)/(P(\rho)+(1-P(\rho))^2);
%	a(\rho)=\rho*(1-\rho);
%	gamma(\z)=2.506628274631*sqrt(1/\z)+0.20888568*(1/\z)^(1.5)+0.00870357*(1/\z)^(2.5)
%		-(174.2106599*(1/\z)^(3.5))/25920-(715.6423511*(1/\z)^(4.5))/1244160)*exp((-ln(1/\z)-1)*\z;
%	B(\p,\q)=gamma(\p)*gamma(\q)/gamma(\p+\q);
%	finf(\x,\rho)=\x^(2*Vinf(\rho)/a(\rho)^2-1)*(1-\x)^(2*(1-Vinf(\rho))/a(\rho)^2-1)/
%		B(2*Vinf(\rho)/a(\rho)^2,2*(1-Vinf(\rho))/a(\rho)^2);
%} 
%]
%\begin{axis}[
%	grid=major,
%	xlabel={$\rho$},
%	xmin=0,
%	xmax=1,
%	ymin=0,
%	ymax=35,
%	legend style={draw=none},
%	legend cell align=left,
%	legend pos=outer north east,
%]
%	\foreach \rho in {0.2,0.4,0.8}{
%		\addplot[black,domain=0:1,smooth,very thick,style={solid},samples=200]{finf(x,\rho)};
%	}
%\end{axis}
%\end{tikzpicture}
%}
%\caption{The asymptotic speed distribution~\eqref{eq:finf.beta} for various traffic densities $\rho\in [0,\,1]$ with $V_\infty(\rho)$ given in~\eqref{eq:Vinf}, $P(\rho)$ like in~\eqref{eq:P} with $\mu=2$, $\lambda=1$ and $a(\rho)=\rho(1-\rho)$, which fulfils the constraint~\eqref{eq:a}.}
%\label{fig:finf.beta}
%\end{figure}

We notice that~\eqref{eq:finf.beta} is a \textit{beta probability density function} with parameters
$$ \alpha:=\frac{2V_\infty(\rho)}{\lambda a^2(\rho)}, \qquad
	\beta:=\frac{2(1-V_\infty(\rho))}{\lambda a^2(\rho)}, $$
whence the mean and variance of the random variable $X\sim\tilde{f}_\infty$ describing the vehicle speed at equilibrium are respectively
\begin{align}
	\begin{aligned}[b]
		& \E(X)=\frac{\alpha}{\alpha+\beta}=V_\infty(\rho), \\
		& \Var(X)=\frac{\alpha\beta}{{(\alpha+\beta)}^2(\alpha+\beta+1)}=
			\frac{\lambda a^2(\rho)}{2+\lambda a^2(\rho)}V_\infty(\rho)(1-V_\infty(\rho)).
	\end{aligned}
	\label{eq:beta.var}
\end{align}

Owing to the discussion set forth at the beginning of Section~\ref{sect:homog},~\eqref{eq:finf.beta} is a good model for the speed distribution at equilibrium. Nevertheless the choice~\eqref{eq:D} of the diffusion coefficient leading to it has to be justified more carefully, because that function actually does not satisfy the assumptions of Proposition~\ref{prop:eta.D}. Precisely, there does not exist any $c>0$ such that $cD(v;\,\rho)\leq\min\{v,\,1-v\}$ for all $v\in [0,\,1]$ due to the vertical tangents at $v=0$ and $v=1$ of the function~\eqref{eq:D}. To obviate this difficulty it is sufficient to consider preliminarily in~\eqref{eq:binary} the truncated diffusion coefficient
$$ D_\gamma(v;\,\rho):=a(\rho)\sqrt{\left((1+\gamma)v(1-v)-\frac{\gamma}{4}\right)_+} $$
where $(\cdot)_+:=\max\{0,\,\cdot\}$ denotes the positive part. This coefficient satisfies Proposition~\ref{prop:eta.D} with $c=\frac{1}{a(\rho)}\sqrt{\frac{\gamma}{1+\gamma}}$ and for $\gamma\to 0^+$ converges uniformly to~\eqref{eq:D}. Hence, in the quasi-invariant limit, from~\eqref{eq:Boltzmann.spathomog-scaled.R} with $D_\gamma(v;\,\rho)$ we obtain~\eqref{eq:FP} with $D(v;\,\rho)$.

In Figure~\ref{fig:finf.beta} we show the asymptotic distribution~\eqref{eq:finf.beta} for $\lambda=1$, $a(\rho)=\rho(1-\rho)$. Such a choice of $a(\rho)$ produces a vanishing diffusion when both $\rho=0$ and $\rho=1$. The asymptotic distributions to which~\eqref{eq:finf.beta} converges in these cases are $\tilde{f}_\infty(v)=\delta(v-1)$, $\tilde{f}_\infty(v)=\delta(v)$, respectively, i.e. the Dirac deltas centred at $v=1$ and $v=0$.

\section{Microscopic binary control for road risk mitigation}
\label{sect:micro.control}
The transportation literature acknowledges the speed variability within the stream of vehicles as one of the major sources of road risk~\cite{peden2004WHO,WHO2015report}. Hence a conceivable goal of driver-assist cars would be to mitigate \textit{collectively} the road risk through a reduction of the speed variance~\eqref{eq:beta.var}. In this section we aim at investigating to what extent this is possible by taking advantage of the ability of such cars to respond \textit{locally} to the actions of their drivers thanks to the automatic technologies they are equipped with.

To this purpose we reconsider the interaction rules~\eqref{eq:binary} and modify them as follows:
\begin{align}
	\begin{aligned}[c]
		v' &= v+\gamma(I(v,\,w;\,\rho)+\Theta\ucontr)+D(v;\,\rho)\eta \\
		w' &= w,
	\end{aligned}
	\label{eq:binary.u}
\end{align}
where $\ucontr$ is a \textit{control} representing the instantaneous correction of the ``natural'' interaction $I$ operated by the driver-assist vehicle and $\Theta\in\{0,\,1\}$ is a Bernoulli random variable expressing the fact that a randomly chosen vehicle may be equipped with driver-assist technology with a certain probability $p\in [0,\,1]$. Hence $\Theta\sim\operatorname{Bernoulli}(p)$ and $p$ gives the fraction of driver-assist vehicles present in the traffic stream, i.e. the so-called \textit{penetration rate}.

The optimal control $\ucontr^\ast$ is chosen so as to optimise the value of a certain \textit{binary cost functional} $J=J(v',\,\ucontr)$, whose minimisation is supposed to be linked locally to the mitigation of the road risk:
$$ \ucontr^\ast:=\argmin_{\ucontr\in\cU}J(v',\,\ucontr) $$
subject to~\eqref{eq:binary.u}, where $\cU$ is a set of admissible controls. Aiming at the reduction of the global speed variance of the flow of vehicles, a possible form of the cost functional in a single binary interaction is
\begin{equation}
	J(v',\,\ucontr)=\frac{1}{2}\ave{(w'-v')^2+\nu\ucontr^2}
	\label{eq:J.binvar}
\end{equation}
where $\frac{1}{2}(w'-v')^2$ is the \textit{binary variance} of the speeds of the two vehicles after the interaction, $\frac{\nu}{2}\ucontr^2$ is a penalisation of large controls with penalisation coefficient $\nu>0$ and finally $\ave{\cdot}$ denotes, as usual, the average with respect to the distribution of the stochastic fluctuation $\eta$. Another option is to minimise the gap between the speed of the vehicles and a certain desired speed $v_d\in [0,\,1]$, possibly $v_d=v_d(\rho)$, which may be thought of as a speed limit or as a recommended speed communicated to the equipped vehicles by some external monitoring devices. In this case we consider the binary cost functional
\begin{equation}
	J(v',\,\ucontr)=\frac{1}{2}\ave{(v_d(\rho)-v')^2+\nu\ucontr^2}.
	\label{eq:J.vd}
\end{equation}
Notice that~\eqref{eq:J.binvar},~\eqref{eq:J.vd} are special cases of the general cost functional
\begin{equation}
	J(v',\,\ucontr)=\frac{1}{2}\ave{(V_d(w',\,\rho)-v')^2+\nu\ucontr^2}
	\label{eq:J.gen}
\end{equation}
with either $V_d(w',\,\rho)=w'$ or $V_d(w',\,\rho)=v_d(\rho)$.

The minimisation of~\eqref{eq:J.gen} constrained to~\eqref{eq:binary.u} can be done by forming the Lagrangian
$$ \cL(v',\,\ucontr,\,\lambda):=J(v',\,\ucontr)+\lambda\ave{v'-v-\gamma(I(v,\,w;\,\rho)+\Theta\ucontr)-D(v;\,\rho)\eta}, $$
where $\lambda\in\R$ is the Lagrange multiplier associated to the constraint~\eqref{eq:binary.u}, and then by computing
\begin{equation*}
	\begin{cases}
		\partial_\ucontr\cL=\nu\ucontr-\gamma\Theta\lambda=0 \\
		\partial_{v'}\cL=\ave{v'-V_d(w',\,\rho)}+\lambda=0.
	\end{cases}
\end{equation*}
This yields the optimal control
$$ \ucontr^\ast=\frac{\gamma}{\nu}\Theta\ave{V_d(w',\,\rho)-v'}, $$
which, using the binary interactions~\eqref{eq:binary.u}, can be expressed in \textit{feedback form} as a function of the pre-interaction speeds $v$, $w$:
\begin{equation}
	\ucontr^\ast=\frac{\gamma\Theta}{\nu+\gamma^2\Theta^2}(V_d(w;\,\rho)-v)-\frac{\gamma^2\Theta}{\nu+\gamma^2\Theta^2}I(v,\,w;\,\rho).
	\label{eq:u.ast}
\end{equation}
Plugging~\eqref{eq:u.ast} into~\eqref{eq:binary.u} we finally obtain the feedback controlled microscopic rules in the form
\begin{align}
	\begin{aligned}[c]
		v' &= v+\dfrac{\nu\gamma}{\nu+\gamma^2\Theta^2}I(v,\,w;\,\rho)
			+\dfrac{\gamma^2\Theta^2}{\nu+\gamma^2\Theta^2}\left(V_d(w;\,\rho)-v\right)+D(v;\,\rho)\eta \\
		w' &= w.
	\end{aligned}
	\label{eq:bin.u-feedback}
\end{align}
Notice that these binary interactions are formally identical to~\eqref{eq:binary} up to introducing the new interaction function
$$ \cI(v,\,w;\,\rho):=\dfrac{\nu}{\nu+\gamma^2\Theta^2}I(v,\,w;\,\rho)
	+\dfrac{\gamma\Theta^2}{\nu+\gamma^2\Theta^2}\left(V_d(w;\,\rho)-v\right). $$
In particular, we can establish the following:
\begin{proposition} \label{prop:eta.D-u}
In~\eqref{eq:bin.u-feedback}, let $\gamma\in [0,\,1]$, $\nu>0$ and $I(v,\,w;\,\rho)$ be given by~\eqref{eq:I}. If there exists $c>0$ such that
\begin{equation*}
	\begin{cases}
		\abs{\eta}\leq c\left(1-\dfrac{\nu+\gamma}{\nu+\gamma^2}\gamma\right) \\[3mm]
		cD(v;\,\rho)\leq\min\{v,\,1-v\}, & \forall\,v,\,\rho\in [0,\,1]
	\end{cases}
\end{equation*}
then $v'\in [0,\,1]$ for every pair of pre-interaction speeds $v,\,w\in [0,\,1]$ and every $\rho\in [0,\,1]$.
\end{proposition}
\begin{proof}
The proof that $v'\geq 0$ is entirely analogous to the corresponding one of Proposition~\ref{prop:eta.D} upon observing that $P(\rho),\,P(\rho)w,\,V_d(w;\,\rho)\geq 0$ imply
\begin{align*}
	v+\gamma\cI(v,\,w;\,\rho)+D(v;\,\rho)\eta &\geq
		\left(1-\frac{\nu+\gamma\Theta^2}{\nu+\gamma^2\Theta^2}\gamma\right)v+D(v;\,\rho)\eta \\
	&\geq \left(1-\frac{\nu+\gamma}{\nu+\gamma^2}\gamma\right)v+D(v;\,\rho)\eta
\end{align*}
and then using the new assumptions on $\eta$, $D(v;\,\rho)$. Likewise, the proof that $v'\leq 1$ follows the very same line as the corresponding one in Proposition~\ref{prop:eta.D} considering that now it results
\begin{align*}
	v+\gamma\cI(v,\,w;\,\rho)+D(v;\,\rho)\eta &\leq
		v+\frac{\nu+\gamma\Theta^2}{\nu+\gamma^2\Theta^2}\gamma(1-v)+D(v;\,\rho)\eta \\
	&\leq v+\frac{\nu+\gamma}{\nu+\gamma^2}\gamma(1-v)+D(v;\,\rho)\eta
\end{align*}
because $P(\rho)w,\,V_d(w;\,\rho)\leq 1$.
\end{proof}

\begin{remark}
From the assumptions of Proposition~\ref{prop:eta.D-u} we see that if $\nu\to+\infty$, i.e. if the control is so penalised that the only possible one is $\ucontr^\ast=0$, cf.~\eqref{eq:u.ast}, then we recover the same condition on $\eta$ as the one of Proposition~\ref{prop:eta.D}. Conversely, if $\nu\to 0^+$, i.e. if the control is not penalised at all and the interactions of the driver-assist vehicles are fully dominated by it, then $\eta$ has to vanish. This is the only way in which the physical bounds on the post-interaction speed can be preserved in such a case. In fact purely controlled dynamics are not aware of the stochastic fluctuations, because $\ucontr^\ast$ has been deduced deterministically by averaging with respect to $\eta$.

Notice that when $\Theta=1$ and $\nu\to 0^+$ the first microscopic rule in~\eqref{eq:bin.u-feedback} can be seen as a model for a \textit{fully automated}, or \textit{autonomous}, vehicle.
\end{remark}

\subsection{Fundamental diagrams}
\label{sect:funddiag.pu}
With the new controlled binary interactions~\eqref{eq:bin.u-feedback} the Boltzmann-type equation for the distribution function $f$ writes
\begin{equation}
	\frac{d}{dt}\int_0^1\varphi(v)f^\ast(t,\,v)\,dv
		=\frac{1}{2}\E_\Theta\left[\ave{\int_0^1\int_0^1\left(\varphi(v')-\varphi(v)\right)f^\ast(t,\,v)f^\ast(t,\,w)\,dv\,dw}\right],
	\label{eq:Boltzmann_ETheta}
\end{equation}
where the superscript $\ast$ on the distribution function $f$ recalls that we are considering the statistical evolution of the system subject to the optimal binary control $\ucontr^\ast$~\eqref{eq:u.ast}, whose contribution is taken into account in $v'$. Furthermore, $\E_\Theta$ denotes the expectation with respect to the random variable $\Theta$ appearing in~\eqref{eq:bin.u-feedback}. In particular, considering that $\Theta^2\sim\operatorname{Bernoulli}(p)$, the evolution of the mean speed $V^\ast$, obtained from~\eqref{eq:Boltzmann_ETheta} with $\varphi(v)=v$, is now given by the equation
\begin{align}
	\begin{aligned}[b]
		\frac{dV^\ast}{dt} &= \frac{\gamma}{2}\left\{\frac{\nu+(1-p)\gamma^2}{\nu+\gamma^2}\Bigl(P(\rho)[1+(1-P(\rho))V^\ast]-V^\ast\Bigr)\right. \\
		&\phantom{=} \left.+\frac{\gamma p}{\nu+\gamma^2}\left(\int_0^1V_d(w;\,\rho)f(t,\,w)\,dw-V^\ast\right)\right\}.
	\end{aligned}
	\label{eq:dV/dt.pu}
\end{align}

We immediately notice that if $p=0$, i.e. if no car is actually equipped with driver-assist technologies, this equation reduces to~\eqref{eq:dV/dt} consistently with the fact that the whole model collapses onto the one considered in Section~\ref{sect:homog} (in fact in such a case we have $\Theta=0$ almost surely in~\eqref{eq:binary.u}). The same conclusion holds also if $\nu\to+\infty$, for then the cost for applying a driver-assist control is so high that the optimal strategy turns out to be not to apply any control. If conversely $\nu\to 0^+$, i.e. the cost of the driver-assist control is negligible, and $p=1$, i.e. all cars in the traffic stream are equipped with driver-assist technologies, then the evolution of $V$ is fully dominated by the second term at the right-hand side of~\eqref{eq:dV/dt.pu}, which results from the action of the control; whereas if $0<p<1$ then the spontaneous (viz. uncontrolled) dynamics, cf. the first term at the right-hand side of~\eqref{eq:dV/dt.pu}, still play a role.

We now discuss in more detail the consequences of~\eqref{eq:dV/dt.pu} in the cases $0<p<1$, $0<\nu<+\infty$ for a generic target speed $V_d(w;\,\rho)$. In order to reduce the analytical complexity due to the number of microscopic parameters in~\eqref{eq:dV/dt.pu} and to preserve simultaneously the qualitative large time behaviour of the system we refer to the quasi-invariant interaction regime. In particular, similarly to what we have done in Section~\ref{sect:asympt}, we consider the limit $\gamma,\,\nu\to 0^+$ and we assume that $\nu/\gamma\to\kappa>0$, so that we can observe asymptotically a balanced contribution of the interactions and of the control. Under the scaling $\tau:=\frac{\gamma}{2}t$, $\tilde{f}^\ast(\tau,\,w):=f^\ast(2\tau/\gamma,\,w)$ and $\tilde{V}^\ast(\tau):=V^\ast(2\tau/\gamma)=\int_0^1 v\tilde{f}^\ast(\tau,\,v)\,dv$ we obtain from~\eqref{eq:dV/dt.pu}:
$$ \frac{d\tilde{V}^\ast}{d\tau}=P(\rho)[1+(1-P(\rho))\tilde{V}^\ast]+p^\ast\int_0^1V_d(w;\,\rho)\tilde{f}^\ast(\tau,\,w)\,dw
	-(1+p^\ast)\tilde{V}^\ast, $$
where
$$ p^\ast:=\frac{p}{\kappa} $$
can be understood as an \textit{effective penetration rate} taking into account not only the actual percentage $p$ of vehicles equipped with driver-assist technologies but also the relative penalisation $\kappa$ of the in-vehicle control. Thus the asymptotic value $V_\infty^\ast$ that the mean speed approaches as $\tau\to+\infty$ satisfies
\begin{equation}
	P(\rho)\left[1+(1-P(\rho))V_\infty^\ast\right]+p^\ast\int_0^1V_d(w;\,\rho)\tilde{f}_\infty^\ast(w)\,dw
		=\left(1+p^\ast\right)V_\infty^\ast.
	\label{eq:Vtildeinf}
\end{equation}

With $V_d(w;\,\rho)=w$, i.e. when the driver-assist control seeks to minimise the binary variance of the speeds of the two interacting vehicles,~\eqref{eq:Vtildeinf} gives for $V_\infty^\ast$ the same as~\eqref{eq:V}, hence there are apparently no differences with respect to the uncontrolled case. However we anticipate that a more accurate investigation of the asymptotic statistical properties of the flow of vehicles, cf. the next Section~\ref{sect:asympt.controlled}, will reveal that the driver-assist control actually succeeds in reducing the asymptotic variance of the microscopic speeds, which is at the heart of the risk mitigation issues.

\begin{figure}[!t]
\centering
\begin{tikzpicture}[
declare function={
	P(\x)=(1-\x)^2;
	vd(\x)=(1-\x);
	Vinf(\x,\p)=(P(\x)+\p*vd(\x))/(P(\x)+(1-P(\x))^2+\p);
} 
]
\begin{groupplot}[
	group style={group size= 2 by 1,horizontal sep=15mm},
	width=0.4\textwidth,
	legend cell align=left,
	legend pos=outer north east,
	legend style={draw=none}
]
\nextgroupplot[ylabel=$\tilde{V}_\infty^\ast(\rho)$,xmin=0,xmax=1,ymin=0,ymax=1,title=Speed diagram] %grid=major,
	\addplot[gray,domain=0:1,smooth,very thick,style={solid}]{Vinf(x,0)};
	\addplot[blue,domain=0:1,smooth,very thick,style={dashed}]{Vinf(x,0.5)};
	\addplot[cyan,domain=0:1,smooth,very thick,style={dash dot}]{Vinf(x,1)};
	\addplot[green,domain=0:1,smooth,very thick,style={dotted}]{Vinf(x,5)};
\nextgroupplot[ylabel=$\rho\tilde{V}_\infty^\ast(\rho)$,xmin=0,xmax=1,ymin=0,ymax=0.25,title=Fundamental diagram] %grid=major,
	\addplot[gray,domain=0:1,smooth,very thick,style={solid}]{x*Vinf(x,0)}; \addlegendentry{$p^\ast=0$};
	\addplot[blue,domain=0:1,smooth,very thick,style={dashed}]{x*Vinf(x,0.5)}; \addlegendentry{$p^\ast=\frac{1}{2}$}
	\addplot[cyan,domain=0:1,smooth,very thick,style={dash dot}]{x*Vinf(x,1)}; \addlegendentry{$p^\ast=1$}
	\addplot[green,domain=0:1,smooth,very thick,style={dotted}]{x*Vinf(x,5)}; \addlegendentry{$p^\ast=5$}
\end{groupplot}
\end{tikzpicture}
\caption{Speed and fundamental diagrams from~\eqref{eq:Vinf.pu} with $v_d(\rho)=1-\rho$ and $P(\rho)$ like in~\eqref{eq:P} with $\mu=2$ for various effective penetration rates $p^\ast$. The case $p^\ast=0$ corresponds to the uncontrolled scenario.}
\label{fig:funddiag.pu}
\end{figure}
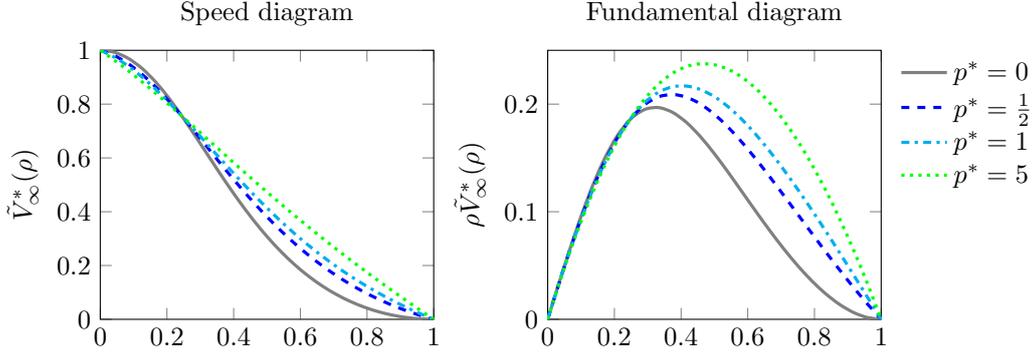

Conversely, with $V_d(w;\,\rho)=v_d(\rho)$, i.e. when the driver-assist control tries to align the car speed to a possibly traffic-dependent desired speed, from~\eqref{eq:Vtildeinf} we deduce
\begin{equation}
	V_\infty^\ast(\rho)=\frac{P(\rho)+p^\ast v_d(\rho)}
		{P(\rho)+{\left(1-P(\rho)\right)}^2+p^\ast}.
	\label{eq:Vinf.pu}
\end{equation}
As a general fact, we notice that now for $p^\ast$ small the speed and fundamental diagrams of traffic are close to those found in the uncontrolled case, cf.~\eqref{eq:Vinf} and Figure~\ref{fig:funddiag}. On the contrary, for $p^\ast$ large they get closer and closer to $v_d(\rho)$, $\rho v_d(\rho)$, respectively, cf. Figure~\ref{fig:funddiag.pu}. Nevertheless, also in this case a more accurate characterisation of the global effect of the driver-assist control on road risk mitigation can be obtained by studying in more detail the statistical properties of the flow at equilibrium, cf. the next Section~\ref{sect:asympt.controlled}.

\subsection{Asymptotic speed variance and risk mitigation}
\label{sect:asympt.controlled}
As claimed at the beginning of Section~\ref{sect:micro.control}, an indicator of the road risk and of the effectiveness of in-vehicle driver-assist control strategies for its mitigation is the variance of the speed distribution. It is therefore interesting to investigate such a statistical property of the flow of vehicles at equilibrium in the case of the controlled binary interactions~\eqref{eq:bin.u-feedback}, taking advantage of the analytical procedure illustrated in Section~\ref{sect:asympt}.

As already set forth in Section~\ref{sect:funddiag.pu}, in order to study the quasi-invariant interaction regime of the Boltzmann-type equation~\eqref{eq:Boltzmann_ETheta} we consider the limit $\gamma,\,\nu,\,\sigma^2\to 0^+$ and assume $\sigma^2/\gamma\to\lambda>0$, $\nu/\gamma\to\kappa>0$, implying that both diffusive and control contributions balance asymptotically with interactions. Under the time scaling $\tau:=\frac{\gamma}{2}t$ we obtain the Fokker-Planck equation
\begin{equation}
	\partial_\tau\tilde{f}^\ast=\frac{\lambda}{2}\partial^2_v(D^2(v;\,\rho)\tilde{f}^\ast)-
		\partial_v\left[\left(\int_0^1\left(I(v,\,w;\,\rho)+p^\ast V_d(w;\,\rho)\right)\tilde{f}^\ast(\tau,\,w)\,dw-p^\ast v\right)
			\tilde{f}^\ast\right],
	\label{eq:FP.pu}
\end{equation}
where $\tilde{f}^\ast$ is required to satisfy boundary conditions similar to~\eqref{eq:FP-bc}. In particular, such conditions are met if $\tilde{f}^\ast(\tau,\,v)=\partial_v\tilde{f}^\ast(\tau,\,v)=0$ for $v=0,\,1$ and all $\tau>0$.

Using the expression~\eqref{eq:I} of $I$ and taking~\eqref{eq:Vtildeinf} into account we have that the asymptotic speed distribution $\tilde{f}_\infty^\ast(v)$ that the system approaches for $\tau\to+\infty$ solves
\begin{equation}
	\frac{\lambda}{2}\partial^2_v(D^2(v;\,\rho)\tilde{f}_\infty^\ast)-
		(1+p^\ast)\partial_v\left(\left(V_\infty^\ast(\rho)-v\right)\tilde{f}_\infty^\ast\right)=0,
	\label{eq:steady_FP.u}
\end{equation}
where $V_\infty^\ast(\rho)$ is either~\eqref{eq:Vinf} or~\eqref{eq:Vinf.pu} depending on the chosen target speed $V_d$. For $D(v;\,\rho)$ like in~\eqref{eq:D} the solution to~\eqref{eq:steady_FP.u} reads
\begin{equation}
	\tilde{f}_\infty^\ast(v)=\frac{v^{\frac{2(1+p^\ast)}{\lambda a^2(\rho)}V_\infty^\ast(\rho)-1}
		(1-v)^{\frac{2(1+p^\ast)}{\lambda a^2(\rho)}(1-V_\infty^\ast(\rho))-1}}
			{\Beta\left(\frac{2(1+p^\ast)}{\lambda a^2(\rho)}V_\infty^\ast(\rho),\,
				\frac{2(1+p^\ast)}{\lambda a^2(\rho)}(1-V_\infty^\ast(\rho))\right)}
	\label{eq:finf.ast}
\end{equation}
and $\tilde{f}_\infty^\ast$, $\partial_v\tilde{f}_\infty^\ast$ vanish at $v=0,\,1$ if
\begin{equation}
	a^2(\rho)\leq\frac{1+p^\ast}{\lambda}\min\left\{V_\infty^\ast(\rho),\,1-V_\infty^\ast(\rho)\right\}.
	\label{eq:a.pu}
\end{equation}
On the whole, the random variable $X^\ast\sim\tilde{f}_\infty^\ast$ describing the controlled vehicle speed at equilibrium is again distributed according to a beta probability density function but now its variance is
\begin{equation}
	\Var(X^\ast)=\frac{\lambda a^2(\rho)}{2+\lambda a^2(\rho)+2p^\ast}
		V_\infty^\ast(\rho)\left(1-V_\infty^\ast(\rho)\right).
	\label{eq:VarXast}
\end{equation}

Let us assume $V_d(w;\,\rho)=w$, so that $V_\infty^\ast(\rho)$ is the same as $V_\infty(\rho)$ in~\eqref{eq:Vinf}. Then a direct comparison between~\eqref{eq:beta.var} and~\eqref{eq:VarXast} shows that $\Var(X^\ast)$ is invariably smaller than $\Var(X)$ for all $\rho\in [0,\,1]$ provided $p^\ast>0$, meaning that an in-vehicle driver-assist system designed to reduce the binary speed variance can effectively mitigate the collective driving risk. We stress that instead a purely macroscopic analysis of the traffic flow based on more standard tools, such as the traffic diagrams, is unable to catch any difference with respect to the uncontrolled case.

\begin{figure}[!t]
\centering
\begin{tikzpicture}[
	declare function={
		f(\x)=0.5*(1+1/16)*\x/(1-\x);
		a(\x)=\x*(1-\x);
		P(\x)=(1-\x)^2;
		vd(\x)=(1-\x);
		Var(\x,\p)=a(\x)^2/(2+a(\x)^2+2*\p)*(P(\x)+\p*vd(\x))/(P(\x)+(1-P(\x))^2+\p)*(1-(P(\x)+\p*vd(\x))/(P(\x)+(1-P(\x))^2+\p));
	}
]
\begin{groupplot}[
	group style={group size= 2 by 1,horizontal sep=15mm},
	width=0.45\textwidth,
	legend cell align=left,
	legend pos=outer north east,
	legend style={draw=none}
]
\nextgroupplot[
	axis lines=middle,
	xtick={0,0.6530,1}, xticklabels={0,$q_\mathrm{max}$,1}, xtick pos={left},
	ytick={0,1}, ytick pos={left},
	xlabel={$q$}, x label style={at={(axis description cs:1,0)},anchor=north},
	ylabel={$p$}, y label style={at={(axis description cs:0,1)},anchor=east},
	xmin=0,
	xmax=1.1,
	ymin=0,
	ymax=1.2
]
	\addplot[orange,domain=0:0.6530,smooth,very thick,style={solid},name path=A]{f(x)};
	\addplot[orange,domain=0.6530:0.99,smooth,very thick,style={dashed}]{f(x)};
	\addplot[gray,style={dashed}] coordinates {(0.6530,0) (0.6530,1)};
	\addplot[gray,style={dashed},name path=B] coordinates {(0,1) (0.6530,1)};
	\addplot[gray,opacity=0.3] fill between[of=A and B,soft clip={domain=0:0.6530}];
	\node at (axis cs: 0.27,0.65){\begin{minipage}{3cm}\centering \footnotesize Risk mitigation \\ achievable\end{minipage}};
	\node at (axis cs: 0.9,0.9){\begin{minipage}{3cm}\centering \footnotesize Risk \\ mitigation \\ not achievable\end{minipage}};

\nextgroupplot[
	xlabel={$\rho$},
	ylabel={$\Var(X^\ast)$},
	xmin=0,
	xmax=1,
	ymin=0
]
	\addplot[gray,domain=0:1,smooth,very thick,style={solid}]{Var(x,0)}; \addlegendentry{$p^\ast=0$}
	\addplot[blue,domain=0:1,smooth,very thick,style={dashed}]{Var(x,0.5)}; \addlegendentry{$p^\ast=\frac{1}{2}$}
	\addplot[cyan,domain=0:1,smooth,very thick,style={dash dot}]{Var(x,1)}; \addlegendentry{$p^\ast=1$}
	\addplot[green,domain=0:1,smooth,very thick,style={dotted}]{Var(x,5)}; \addlegendentry{$p^\ast=5$}
\end{groupplot}
\end{tikzpicture}
\caption{Left: Relationship between the target risk mitigation $q$ and the penetration rate $p$ with the binary speed variance control strategy, cf.~\eqref{eq:p},~\eqref{eq:qbar}. Right: $\Var(X^\ast)$ as a function of the traffic density $\rho$, cf.~\eqref{eq:VarXast}, in the case of the desired speed control strategy. Here $V_\infty(\rho)$ is given by~\eqref{eq:Vinf.pu} with $v_d(\rho)=1-\rho$; moreover we have set $\lambda=1$, $a(\rho)=\rho(1-\rho)$, $P(\rho)$ like in~\eqref{eq:P} with $\mu=2$. It can be checked that these choices comply with condition~\eqref{eq:a.pu} for all $p^\ast=\frac{p}{\kappa}\geq 0$.}
\label{fig:pq.binvar}
\end{figure}
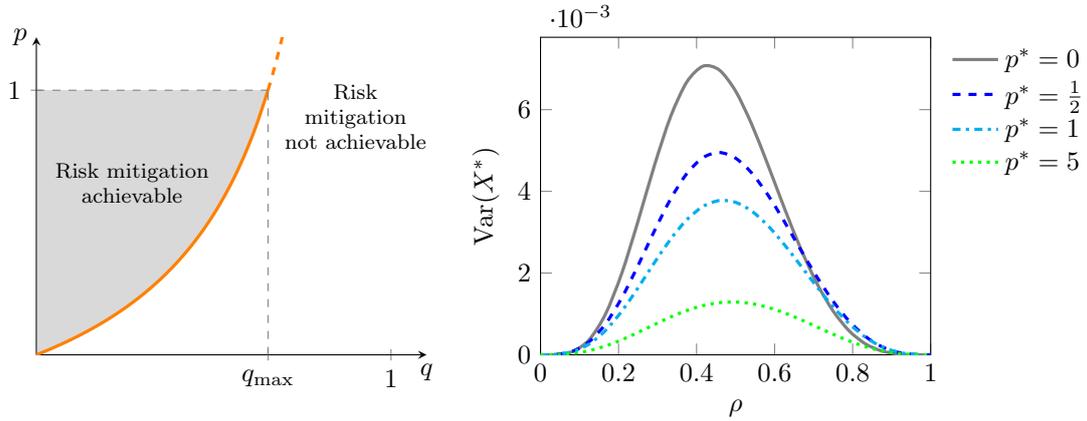

The relative reduction of the speed variance at equilibrium representing the risk mitigation factor, say $q$, with respect to the uncontrolled scenario is
$$ q:=\frac{\Var(X)-\Var(X^\ast)}{\Var(X)}=\frac{p^\ast}{1+\frac{\lambda}{2}a^2(\rho)+p^\ast}, $$
whence the minimum penetration rate $p$ necessary to achieve a given risk mitigation factor can be computed as
\begin{equation}
	p\geq\kappa\left(1+\frac{\lambda}{2}a^2(\rho)\right)\frac{q}{1-q}.
	\label{eq:p}
\end{equation}
Considering that the penetration rate can be at most $p=1$ when all vehicles in the traffic stream are equipped with driver-assist technologies, from~\eqref{eq:p} we also infer that the maximum achievable risk mitigation, say $q_\mathrm{max}$, is
\begin{equation}
	q_\mathrm{max}=\frac{1}{1+\kappa\left(1+\frac{\lambda}{2}a^2(\rho)\right)},
	\label{eq:qbar}
\end{equation}
see Figure~\ref{fig:pq.binvar} (left).

Conversely, if we assume $V_d(w;\,\rho)=v_d(\rho)$ then $V_\infty^\ast(\rho)$ is given by~\eqref{eq:Vinf.pu} and a comparison of~\eqref{eq:VarXast} with the uncontrolled case~\eqref{eq:beta.var} is now less straightforward. From~\eqref{eq:Vinf.pu} we have that $V_\infty^\ast(\rho)\to v_d(\rho)$ when $p^\ast\to+\infty$, i.e. $\kappa\to 0^+$; in the same limit, from~\eqref{eq:VarXast} we also find $\Var(X^\ast)\to 0$. Thus the rationale behind this control strategy is to mitigate the driving risk by inducing the synchronisation of the traffic flow around a traffic-dependent recommended speed. However, since $v_d(\rho)$ may be chosen independently of the ``spontaneous'' mean speed~\eqref{eq:Vinf}, we observe that in general this control strategy does not guarantee that the speed variance $\Var(X^\ast)$ for $p^\ast>0$ be always strictly lower than $\Var(X)$ for $p^\ast=0$, see Figure~\ref{fig:pq.binvar} (right).

\bigskip

We conclude this section by generalising the results discussed so far to sufficiently arbitrary interaction functions $I$ and diffusion coefficients $D$. For this it is convenient to introduce the concept of \textit{energy} of the system, which is defined as
$$ \tilde{E}(\tau):=\int_0^1v^2\tilde{f}(\tau,\,v)\,dv. $$
Notice that the speed variance at every time $\tau\geq 0$ can then be computed as $\tilde{E}(\tau)-\tilde{V}^2(\tau)$.

\begin{theorem}[Binary variance control] \label{theo:bin.var}
In~\eqref{eq:bin.u-feedback}, let $V_d(w;\,\rho)=w$ and furthermore let $I$ be a linear-affine function in $v,\,w$:
$$ I(v,\,w;\,\rho)=A(\rho)v+B(\rho)w+C(\rho), $$
and $D$ be given by~\eqref{eq:D}. Assume that initially the mean speed and the energy of the controlled and uncontrolled models are the same, i.e. $\tilde{V}^\ast(0)=\tilde{V}(0)$ and $\tilde{E}^\ast(0)=\tilde{E}(0)$. Then:
\begin{align*}
	& \tilde{V}^\ast(\tau)=\tilde{V}(\tau) \\
	& \tilde{E}^\ast(\tau)-(\tilde{V}^\ast(\tau))^2\leq \tilde{E}(\tau)-\tilde{V}^2(\tau)
\end{align*}
for all $\tau>0$.
\end{theorem}
\begin{proof}
Letting $\varphi(v)=v$ in~\eqref{eq:FP.weak} gives, in the uncontrolled case,
$$ \frac{d\tilde{V}}{d\tau}=\left(A(\rho)+B(\rho)\right)\tilde{V}+C(\rho). $$
Likewise, multiplying~\eqref{eq:FP.pu} by $v$ and integrating by parts on $[0,\,1]$ with the proper boundary conditions on $\tilde{f}^\ast$ produces, in the controlled case with $V_d(w;\,\rho)=w$,
$$ \frac{d\tilde{V}^\ast}{d\tau}=\left(A(\rho)+B(\rho)\right)\tilde{V}^\ast+C(\rho), $$
thus $\frac{d}{d\tau}(\tilde{V}^\ast-\tilde{V})=\left(A(\rho)+B(\rho)\right)(\tilde{V}^\ast-\tilde{V})$ and finally $\tilde{V}^\ast(\tau)=\tilde{V}(\tau)$ for all $\tau>0$ as $\tilde{V}^\ast(0)=\tilde{V}(0)$.

Letting now $\varphi(v)=v^2$ with $D(v;\,\rho)=a(\rho)\sqrt{v(1-v)}$ in~\eqref{eq:FP.weak} yields, in the uncontrolled case,
$$ \frac{d\tilde{E}}{d\tau}=2\left(A(\rho)\tilde{E}+B(\rho)\tilde{V}^2+C(\rho)\tilde{V}\right)
	+\lambda a^2(\rho)\left(\tilde{V}-\tilde{E}\right). $$
On the other hand, multiplying~\eqref{eq:FP.pu} by $v^2$ and integrating on $[0,\,1]$ we discover, in the controlled case,
\begin{align*}
	\frac{d\tilde{E}^\ast}{d\tau} &= 2\left(A(\rho)\tilde{E}^\ast+B(\rho)(\tilde{V}^\ast)^2+C(\rho)\tilde{V}^\ast\right)
		+\lambda a^2(\rho)\left(\tilde{V}^\ast-\tilde{E}^\ast\right) \\
	&\phantom{=} +2p^\ast\left((\tilde{V}^\ast)^2-\tilde{E}^\ast\right).
\end{align*}
Since $\tilde{V}^\ast=\tilde{V}$ while $(\tilde{V}^\ast)^2-\tilde{E}^\ast\leq 0$, because it is the opposite of the variance of $\tilde{f}^\ast$, this further implies
$$ \frac{d\tilde{E}^\ast}{d\tau}\leq 2\left(A(\rho)\tilde{E}^\ast+B(\rho)\tilde{V}^2+C(\rho)\tilde{V}\right)
	+\lambda a^2(\rho)\left(\tilde{V}-\tilde{E}^\ast\right), $$
whence by difference $\frac{d}{d\tau}(\tilde{E}^\ast-\tilde{E})\leq (2A(\rho)-\lambda a^2(\rho))(\tilde{E}^\ast-\tilde{E})$ and finally $\tilde{E}^\ast(\tau)-\tilde{E}(\tau)\leq 0$ for all $\tau>0$ because $\tilde{E}^\ast(0)=\tilde{E}(0)$. Then
$$ \tilde{E}^\ast(\tau)-(\tilde{V}^\ast(\tau))^2=\tilde{E}^\ast(\tau)-\tilde{V}^2(\tau)\leq
	\tilde{E}(\tau)-\tilde{V}^2(\tau) $$
and the thesis follows.
\end{proof}

\begin{theorem}[Desired speed control]
In~\eqref{eq:bin.u-feedback}, let $V_d(w;\,\rho)=v_d(\rho)$ and let $I,\,D$ be bounded for all $v,\,w,\,\rho\in [0,\,1]$. Then
$$ \abs{V_\infty^\ast(\rho)-v_d(\rho)}\lesssim\frac{1}{p^\ast},
	\qquad \abs{E_\infty^\ast(\rho)-v_d^2(\rho)}\lesssim\frac{1}{p^\ast} $$
and in particular
$$ E_\infty^\ast(\rho)-{(V_\infty^\ast(\rho))}^2\lesssim\frac{1}{p^\ast}. $$
\end{theorem}
\begin{proof}
Multiplying the Fokker-Planck equation~\eqref{eq:FP.pu} by $v$ and integrating on $[0,\,1]$ gives
$$ \frac{d\tilde{V}^\ast}{d\tau}=\int_0^1\int_0^1 I(v,\,w;\,\rho)\tilde{f}^\ast(\tau,\,v)\tilde{f}^\ast(\tau,\,w)\,dv\,dw
	+p^\ast(v_d(\rho)-\tilde{V}^\ast). $$
Writing $\frac{d\tilde{V}^\ast}{d\tau}=\frac{d}{d\tau}(\tilde{V}^\ast-v_d(\rho))$ because $v_d(\rho)$ is constant we obtain
$$ \tilde{V}^\ast(\tau)-v_d(\rho)=e^{-p^\ast\tau}(\tilde{V}^\ast(0)-v_d(\rho))
	+\int_0^\tau e^{-p^\ast(\tau-s)}\int_0^1\int_0^1 I(v,\,w;\,\rho)\tilde{f}^\ast(\tau,\,v)\tilde{f}^\ast(\tau,\,w)\,dv\,dw $$
whence, since $I$ is bounded, say $\abs{I(v,\,w;\,\rho)}\leq I_\text{max}$ for all $v,\,w,\,\rho\in [0,\,1]$,
$$ \abs{\tilde{V}^\ast(\tau)-v_d(\rho)}\leq e^{-p^\ast\tau}\abs{\tilde{V}^\ast(0)-v_d(\rho)}
	+\frac{I_\text{max}}{p^\ast}\left(1-e^{-p^\ast\tau}\right), $$
which asymptotically ($\tau\to+\infty$) yields $\abs{V_\infty^\ast-v_d(\rho)}\leq\frac{I_\text{max}}{p^\ast}$.

Multiplying now~\eqref{eq:FP.pu} by $v^2$ and integrating on $[0,\,1]$ produces
\begin{align*}
	\frac{d\tilde{E}^\ast}{d\tau} &= 2\int_0^1\int_0^1 vI(v,\,w;\,\rho)\tilde{f}^\ast(\tau,\,v)\tilde{f}^\ast(\tau,\,w)\,dv\,dw \\
	&\phantom{=} +2p^\ast\left(v_d(\rho)\tilde{V}^\ast(\tau)-\tilde{E}^\ast(\tau)\right)
		+\lambda\int_0^1 D^2(v;\,\rho)\tilde{f}^\ast(\tau,\,v)\,dv.
\end{align*}
The asymptotic behaviour of the energy ($\tau\to+\infty$) is found by setting the left-hand side to zero, whence
$$ E_\infty^\ast=v_d(\rho)V_\infty^\ast+
	\frac{1}{p^\ast}\int_0^1\int_0^1 vI(v,\,w;\,\rho)\tilde{f}_\infty^\ast(v)\tilde{f}_\infty^\ast(w)\,dv\,dw
		+\frac{\lambda}{2p^\ast}\int_0^1 D^2(v;\,\rho)\tilde{f}_\infty^\ast(v)\,dv. $$
From here, subtracting $v_d^2(\rho)$ to both sides and using the boundedness of $I$, $D$ (with, say, $D(v;\,\rho)\leq D_\text{max}$ for all $v,\,\rho\in [0,\,1]$) and the previous result on $V_\infty^\ast$ we obtain
\begin{align*}
	\abs{E_\infty^\ast-v_d^2(\rho)} &\leq v_d(\rho)\abs{V_\infty^\ast-v_d(\rho)}
		+\left(I_\text{max}+\frac{\lambda D^2_\text{max}}{2}\right)\frac{1}{p^\ast} \\
	&\leq \left(v_d(\rho)I_\text{max}+I_\text{max}+\frac{\lambda D^2_\text{max}}{2}\right)\frac{1}{p^\ast} \\
	&\leq \left(2I_\text{max}+\frac{\lambda D^2_\text{max}}{2}\right)\frac{1}{p^\ast}.
\end{align*}

Finally, the speed variance at equilibrium is
\begin{align*}
	E_\infty^\ast-{(V_\infty^\ast)}^2 &= E_\infty^\ast-v_d^2(\rho)+v_d^2(\rho)-{(V_\infty^\ast)}^2 \\
	&\leq \abs{E_\infty^\ast-v_d^2(\rho)}+\abs{(v_d(\rho)+V_\infty^\ast)(v_d(\rho)-V_\infty^\ast)} \\
	&\leq \abs{E_\infty^\ast-v_d^2(\rho)}+2\abs{v_d(\rho)-V_\infty^\ast} \\
	&\leq \left(4I_\text{max}+\frac{\lambda D^2_\text{max}}{2}\right)\frac{1}{p^\ast}. \qedhere
\end{align*}
\end{proof}

\section{Hydrodynamic models}
\label{sect:hydro}
The homogeneous kinetic equations studied in the previous sections are the basis to derive hydrodynamic traffic models incorporating the microscopic control strategies of driver-assist vehicles.

The kinetic framework to obtain hydrodynamic models is provided by the \textit{inhomogeneous Boltzmann equation}, which with the controlled binary interaction rules~\eqref{eq:bin.u-feedback} reads
\begin{align}
	\begin{aligned}[b]
		\partial_t\int_0^1\varphi(v)f^\ast(t,\,x,\,v)\,dv &+ \partial_x\int_0^1 v\varphi(v)f^\ast(t,\,x,\,v)\,dv \\
		&= \frac{1}{2}\E_\Theta\left[\ave{\int_0^1\int_0^1\left(\varphi(v')-\varphi(v)\right)f^\ast(t,\,x,\,v)f^\ast(t,\,x,\,w)\,dv\,dw}\right].
	\end{aligned}
	\label{eq:Boltzmann.inhomog}
\end{align}
Notice that this is the spatially inhomogeneous counterpart of~\eqref{eq:Boltzmann_ETheta}, with $x\in\R$ denoting the space position of the vehicles. On the whole, the microscopic state of the vehicles is now defined by the position-speed pair $(x,\,v)\in\R\times [0,\,1]$ and $f^\ast=f^\ast(t,\,x,\,v)$ is its probability density at time $t\geq 0$. In particular,
$$ \int_\R\int_0^1f^\ast(t,\,x,\,v)\,dv\,dx=1, \qquad \forall\,t\geq 0 $$
while
$$ \rho(t,\,x):=\int_0^1f^\ast(t,\,x,\,v)\,dv $$
is the vehicle density at time $t$ in the point $x$, which, unlike the homogeneous model, is in general no longer constant in time due to the transport dynamics in space (cf. the second term at the left-hand side in~\eqref{eq:Boltzmann.inhomog}).

\subsection{Local equilibrium closure}
\label{sect:local.equil}
Transport models at the macroscopic scale can be recovered from~\eqref{eq:Boltzmann.inhomog} by means of a hyperbolic scaling of time and space:
$$ \tau:=\epsilon t, \qquad \xi:=\epsilon x \qquad (0<\epsilon\ll 1), $$
which defines the hydrodynamic temporal and spatial scales, respectively. After introducing the scaled distribution function $\tilde{f}^\ast(\tau,\,\xi,\,v):=f^\ast(\tau/\epsilon,\,\xi/\epsilon,\,v)$ and noticing that $\partial_\tau\tilde{f}^\ast=\frac{1}{\epsilon}\partial_tf^\ast$, $\partial_\xi\tilde{f}^\ast=\frac{1}{\epsilon}\partial_xf^\ast$ this yields
\begin{align}
	\begin{aligned}[b]
		\partial_\tau\int_0^1\varphi(v)\tilde{f}^\ast(\tau,\,\xi,\,v)\,dv
			&+\partial_\xi\int_0^1 v\varphi(v)\tilde{f}^\ast(\tau,\,\xi,\,v)\,dv \\
		&= \frac{1}{2\epsilon}\E_\Theta\left[\ave{\int_0^1\int_0^1\left(\varphi(v')-\varphi(v)\right)\tilde{f}^\ast(\tau,\,\xi,\,v)
			\tilde{f}^\ast(\tau,\,\xi,\,w)\,dv\,dw}\right].
	\end{aligned}
	\label{eq:complete_BOL_weak}
\end{align}
If $\epsilon$ is small then vehicle interactions (right-hand side) dominate over the advection of the distribution function $\tilde{f}^\ast$ (left-hand side). Inspired by~\cite{during2007PHYSA}, we can then split the dynamics on two well-separated time scales as follows: we consider a ``slow'' pure transport
\begin{equation}
	\partial_\tau\int_0^1\varphi(v)\tilde{f}^\ast(\tau,\,\xi,\,v)\,dv
		+\partial_\xi\int_0^1 v\varphi(v)\tilde{f}^\ast(\tau,\,\xi,\,v)\,dv=0
	\label{eq:splitting.transp}
\end{equation}
and parallelly ``quick'' interactions
\begin{equation}
	\partial_\tau\int_0^1\varphi(v)\tilde{f}^\ast(\tau,\,\xi,\,v)\,dv
		=\frac{1}{2\epsilon}\E_\Theta\left[\ave{\int_0^1\int_0^1\left(\varphi(v')-\varphi(v)\right)
			\tilde{f}^\ast(\tau,\,\xi,\,v)\tilde{f}^\ast(\tau,\,\xi,\,w)\,dv\,dw}\right]
	\label{eq:splitting.int.ftilde}
\end{equation}
which produce a local redistribution of the speeds. Notice that~\eqref{eq:splitting.int.ftilde} is actually an equation parametrised by $\xi$ on the original $t$-scale, in fact setting $g(t,\,\xi,\,v):=\tilde{f}^\ast(\epsilon t,\,\xi,\,v)$ we have
\begin{equation}
	\partial_t\int_0^1\varphi(v)g(t,\,\xi,\,v)\,dv
		=\frac{1}{2}\E_\Theta\left[\ave{\int_0^1\int_0^1\left(\varphi(v')-\varphi(v)\right)
			g(t,\,\xi,\,v)g(t,\,\xi,\,w)\,dv\,dw}\right].
	\label{eq:splitting.int.g}
\end{equation}

For fixed $\tau>0$, if $\epsilon\to 0^+$ then $t\to +\infty$, whence we deduce that the solution $\tilde{f}^\ast$ to~\eqref{eq:splitting.int.ftilde} at time $\tau$ with $\epsilon$ small is close to the local stationary solution $g_\infty(\xi,\,v)$ to~\eqref{eq:splitting.int.g}. On the other hand,~\eqref{eq:splitting.int.g} is virtually~\eqref{eq:Boltzmann_ETheta} for every fixed $\xi\in\R$ but with $\int_0^1g(t,\,\xi,\,v)\,dv=\rho(\tau,\,\xi)$ for all $t>0$. In fact, setting $\varphi(v)=1$ in~\eqref{eq:splitting.int.g} we obtain that the $v$-integral of $g$ is constant in time; moreover, by definition of $g$, it has to be equal to that of $\tilde{f}^\ast$ at time $\tau$. In conclusion, $g_\infty$ can be consistently written as $g_\infty(\xi,\,v)\approx\rho(\tau,\,\xi)\tilde{f}^\ast_\infty(v)$, where $\tilde{f}^\ast_\infty$ is the Fokker-Planck approximation~\eqref{eq:finf.ast} of the equilibrium solution to~\eqref{eq:Boltzmann_ETheta}, and finally:
$$ \tilde{f}^\ast(\tau,\,\xi,\,v)\approx\rho(\tau,\,\xi)\tilde{f}^\ast_\infty(v), $$
which provides the local equilibrium closure (in classical terms, the \textit{local ``Maxwellian'' function}) to be plugged into~\eqref{eq:splitting.transp} to obtain macroscopic conservation laws for the hydrodynamic parameters:
$$ \partial_\tau\left(\rho\int_0^1\varphi(v)\tilde{f}^\ast_\infty(v)\,dv\right)
	+\partial_\xi\left(\rho\int_0^1v\varphi(v)\tilde{f}^\ast_\infty(v)\,dv\right)=0. $$

Since the microscopic interactions~\eqref{eq:bin.u-feedback} conserve only the zeroth moment of the kinetic distribution function, a closed hydrodynamic equation is consistently obtained in terms of $\rho$ alone by choosing $\varphi(v)=1$:
\begin{equation}
	\partial_\tau\rho+\partial_\xi(\rho V^\ast_\infty(\rho))=0.
	\label{eq:first.order}
\end{equation}
This is a first order hydrodynamic traffic model with flux $\rho V^\ast_\infty(\rho)$, which is not necessarily concave for all $\rho\in [0,\,1]$, cf. Figure~\ref{fig:funddiag.pu} (right), as it happens instead more commonly in classical macroscopic models of vehicular traffic. Moreover, the flux depends ultimately on the microscopic control strategy implemented on the driver-assist vehicles. If one chooses the binary variance control then $V^\ast_\infty(\rho)$ is actually given by~\eqref{eq:Vinf} and no macroscopic impact on the vehicle density is observed with respect to the uncontrolled case:
\begin{equation}
	\partial_{\tau}\rho+\partial_{\xi}\left(\frac{\rho P(\rho)}{P(\rho)+(1-P(\rho))^2}\right)=0,
	\label{eq:macro_binvar}
\end{equation}
where $P(\rho)$ is given by~\eqref{eq:P}. Conversely, if one chooses the desired speed control then $V^\ast_\infty(\rho)$ is given by~\eqref{eq:Vinf.pu} and in the hydrodynamic limit we have
\begin{equation}
	\partial_\tau\rho+\partial_\xi\left(\frac{\rho\left(P(\rho)+p^\ast v_d(\rho)\right)}{P(\rho)+(1-P(\rho))^2+p^\ast}\right)=0.
	\label{eq:macro_vdes}
\end{equation}
Notice that now the concavity of the flux may depend strongly on the effective penetration rate $p^\ast=\frac{p}{\kappa}$ of the driver-assist vehicles and on the choice of the recommended speed $v_d(\rho)$.

\subsection{Monokinetic closure}
\label{sect:monokinetic}
A quite different procedure to obtain hydrodynamic models from~\eqref{eq:Boltzmann.inhomog}, which does not use the idea of local equilibrium of the interactions, consists in making the following ansatz on the solution $f^\ast$ to~\eqref{eq:Boltzmann.inhomog}:
\begin{equation}
	f^\ast(t,\,x,\,v)=\rho(t,\,x)\delta(v-u(t,\,x)),
	\label{eq:f.monokin}
\end{equation}
where $\rho$, $u$ are the hydrodynamic parameters denoting the vehicle density and the mean speed, respectively, and $\delta$ is the Dirac delta distribution. Such an ansatz is called a \textit{monokinetic} closure, because it corresponds to assuming that locally all vehicles travel at the same speed or, in other words, that the kinetic distribution function has locally zero speed variance.

Plugging~\eqref{eq:f.monokin} into~\eqref{eq:Boltzmann.inhomog} with $\varphi(v)=1$ yields
$$ \partial_t\rho+\partial_x(\rho u)=0, $$
namely the continuity equation stating the conservation of the mass of vehicles. Since in the monokinetic closure the hydrodynamic parameters $\rho$, $u$ are assumed to be independent, another macroscopic equation is needed in order to get a self-consistent hydrodynamic model. This is obtained from~\eqref{eq:Boltzmann.inhomog} with~\eqref{eq:f.monokin} and $\varphi(v)=v$, which gives the momentum balance equation
$$ \partial_t(\rho u)+\partial_x(\rho u^2)=
	\frac{\gamma}{2}\rho^2\left(\frac{\nu+(1-p)\gamma^2}{\nu+\gamma^2}I(u,\,u;\,\rho)
		+\frac{\gamma p}{\nu+\gamma^2}(V_d(u;\,\rho)-u)\right). $$
From these equations, passing to the hydrodynamic temporal and spatial scales $\tau:=\frac{\gamma}{2}t$, $\xi:=\frac{\gamma}{2}x$ and taking the quasi-invariant interaction limit $\gamma,\,\nu\to 0^+$ with $\nu/\gamma\to\kappa>0$, we finally obtain the pressureless second order hydrodynamic traffic model
\begin{equation}
	\begin{cases}
		\partial_\tau\rho+\partial_\xi(\rho u)=0 \\[1mm]
		\partial_\tau(\rho u)+\partial_\xi(\rho u^2)=	\rho^2\left[I(u,\,u;\,\rho)+p^\ast(V_d(u;\,\rho)-u)\right].
	\end{cases}
	\label{eq:second.order}
\end{equation}

Using the function $I$ in~\eqref{eq:I}, we can derive from~\eqref{eq:second.order} two specific models depending on the choice of the control strategy. In the case of the binary variance control, i.e. for $V_d(w;\,\rho)=w$, we have
\begin{equation}
	\begin{cases}
		\partial_\tau\rho+\partial_\xi(\rho u)=0 \\[1mm]
		\partial_\tau(\rho u)+\partial_\xi(\rho u^2)=	\rho^2\left[P(\rho)-\left(P(\rho)+(1-P(\rho))^2\right)u\right].
	\end{cases}
	\label{eq:macro_second_binvar}
\end{equation}
Notice that here the control has actually no effect at all consistently with the monokinetic ansatz, which indeed postulates a locally null speed variance. The right-hand side of the momentum equation turns out to be a relaxation of the mean speed $u$ towards the local equilibrium $V_\infty(\rho(t,\,x))$, where $V_\infty(\rho)$ is the uncontrolled asymptotic speed~\eqref{eq:Vinf}, which actually coincides with the controlled asymptotic speed, cf. Section~\ref{sect:funddiag.pu} and Theorem~\ref{theo:bin.var}. Conversely, in the case of the desired speed control, i.e. for $V_d(w;\,\rho)=v_d(\rho)$, we have
\begin{equation}
	\begin{cases}
		\partial_\tau\rho+\partial_\xi(\rho u)=0 \\[1mm]
		\partial_\tau(\rho u)+\partial_\xi(\rho u^2)=	\rho^2\left[P(\rho)-\left(P(\rho)+(1-P(\rho))^2\right)u
			+p^\ast(v_d(\rho)-u)\right].
	\end{cases}
	\label{eq:macro_second_vdes}
\end{equation}
Now the right-hand side of the momentum equation expresses a relaxation of $u$ towards the local equilibrium $V_\infty^\ast(\rho(t,\,x))$, where $V_\infty^\ast(\rho)$ is the controlled asymptotic speed~\eqref{eq:Vinf.pu}.

\begin{remark}[Hyperbolicity of system~\eqref{eq:second.order}]
It can be easily checked that the two eigenvalues of the Jacobian matrix of the flux of system~\eqref{eq:second.order} are both equal to $u$. In particular, they satisfy the property that no characteristic speed be higher than the flow speed, a fact related to the anisotropy of the interactions between any two vehicles which has become a consistency requirement for all second order hydrodynamic traffic models since the celebrated papers~\cite{aw2000SIAP,daganzo1995TR}.

It is worth pointing out that such a requirement is instead inevitably violated if one attempts to obtain second order models from a local equilibrium closure. In this case typically one defines $(\rho u)(t,\,x):=\int_0^1 vf(t,\,x,\,v)\,dv$ and lets $\varphi(v)=1,\,v$ in~\eqref{eq:Boltzmann.inhomog} to get
\begin{equation*}
	\begin{cases}
		\partial_t\rho+\partial_x(\rho u)=0 \\[1mm]
		\partial_t(\rho u)+\partial_x(\rho u^2+\cP)=
			\displaystyle\frac{1}{2}\E_\Theta\left[\ave{\int_0^1\int_0^1(v'-v)f^\ast(t,\,x,\,v)f^\ast(t,\,x,\,w)\,dv\,dw}\right],
	\end{cases}
\end{equation*}
where $\cP:=\int_0^1(v-u)^2f^\ast(t,\,x,\,v)\,dv$ is the \textit{traffic pressure}. Then in order to close the momentum equation one forces the ansatz $f^\ast(t,\,x,\,v)=\rho(t,\,x)\tilde{f}_\infty^\ast(v)$ and simultaneously replaces $u$ with $V_\infty^\ast(\rho)$ in $\cP$. As a result, since the asymptotic distribution function $\tilde{f}_\infty^\ast$ is parametrised only by $\rho$, the traffic pressure becomes a function of $\rho$ alone, i.e. $\cP=\cP(\rho)$, and the hydrodynamic model finally reads
\begin{equation*}
	\begin{cases}
		\partial_t\rho+\partial_x(\rho u)=0 \\[1mm]
		\partial_t(\rho u)+\partial_x(\rho u^2+\cP(\rho))=
			\displaystyle\frac{\rho^2}{2}\E_\Theta\left[\ave{\int_0^1\int_0^1(v'-v)\tilde{f}_\infty^\ast(v)\tilde{f}_\infty^\ast(w)\,dv\,dw}\right].
	\end{cases}
\end{equation*}
Now the eigenvalues of the Jacobian matrix of the flux are $u\pm\sqrt{\cP'(\rho)}$, thus the system is hyperbolic if and only if $\cP'(\rho)\geq 0$. Notice that if $\cP'(\rho)=0$, i.e. if $\cP$ is independent of $\rho$, then $f^\ast$ ought to be the monokinetic distribution function, which however is incompatible with the closure with $\tilde{f}_\infty^\ast$. Then it has to be $\cP'(\rho)>0$, which nevertheless implies $u+\sqrt{\cP'(\rho)}>u$, thereby violating the consistency requirement previously recalled.

The splitting procedure of Section~\ref{sect:local.equil} suggests why the local equilibrium closure may not be suited to the derivation of second order macroscopic traffic models. In fact, the procedure shows that such a closure is justified if there is a clear separation between the time scale of the microscopic interactions~\eqref{eq:splitting.int.ftilde} and the space-time scale of a \textit{conservative} transport of the moments of the local ``Maxwellian'' distribution~\eqref{eq:splitting.transp}, which provide the hydrodynamic quantities of interest for the macroscopic model. Consistently, such moments need to be \textit{conserved} by the microscopic interactions (in classical terms, they need to be \textit{``collision'' invariants}). In traffic models, however, the mean speed is \textit{not} conserved by the microscopic interactions, so neither is the momentum. The only conserved hydrodynamic quantity is the traffic density, which makes the continuity equation~\eqref{eq:first.order} straightforwardly closed in terms of $\rho$ alone. Consequently, it is neither necessary nor possible to join to it a second equation for the momentum.
\end{remark}

\section{Numerical tests}
\label{sect:num}
In this section we present several numerical tests, which highlight the main features of the proposed control strategies for a speed-dependent risk mitigation at both the kinetic and the hydrodynamic levels. In particular, we give some insights into the Boltzmann-type controlled kinetic model and the corresponding hydrodynamic approximations for various choices of the penetration rate $p$ of driver-assist vehicles in the traffic stream.

We adopt a Monte Carlo approach for the numerical solution of the Boltzmann-type equation in the quasi-invariant interaction limit, see~\cite{pareschi2001ESAIMP,pareschi2013BOOK} for an introduction. We use instead finite volume WENO schemes to tackle the macroscopic conservation laws with non-convex fluxes, see~\cite{shu2009SIREV} and references therein. In all tests we consider the function $P(\rho)$ given in~\eqref{eq:P} with $\mu=2$ and the recommended speed $v_d(\rho)=1-\rho$. Other relevant parameters will be specified from case to case.

\subsection{Inhomogeneous kinetic model}
We begin by rewriting the inhomogeneous Boltzmann-type model~\eqref{eq:complete_BOL_weak} in strong form, which is more suited to numerical purposes. This reads
\begin{equation}
	\partial_\tau\tilde{f}(\tau,\,\xi,\,v)+v\partial_\xi\tilde{f}(\tau,\,\xi,\,v)
		=\frac{1}{\epsilon}Q(\tilde{f},\,\tilde{f})(\tau,\,\xi,\,v),
	\label{eq:Boltzmann_strong}
\end{equation}
where $Q$ at the right-hand side is the binary interaction operator (in classical terms, the \textit{``collision'' operator}) defined as
$$	Q(\tilde{f},\,\tilde{f})(\tau,\,\xi,\,v)=
		\frac{1}{2}\E_\Theta\left[\ave{\int_0^1\left(\frac{1}{\pr{J}}\tilde{f}(\tau,\,\xi,\,\pr{v})
			\tilde{f}(\tau,\,\xi,\,\pr{w})-\tilde{f}(\tau,\,\xi,\,v)\tilde{f}(\tau,\,\xi,\,w)\right)dw}\right], $$
where $(\pr{v},\,\pr{w})$ are the pre-interaction speeds which generate the post-interaction speeds $(v,\,w)$ according to the interaction rules~\eqref{eq:bin.u-feedback} and $\pr{J}$ is the Jacobian of the transformation from $(\pr{v},\,\pr{w})$ to $(v,\,w)$.

We now briefly account for the numerical scheme by which we solve~\eqref{eq:Boltzmann_strong}. After fixing $\epsilon=10^{-3}$ and introducing a time discretisation $\tau^n:=n\Delta{\tau}$, with $\Delta{\tau}>0$ and $n\in\mathbb{N}$, we adopt a splitting approach.
\begin{description}
\item[Interaction step.] Starting from $\tau^n$, we first integrate the interactions in a single time step and in all space positions $\xi$:
\begin{equation}
	\begin{cases}
		\partial_\tau F(\tau,\,\xi,\,v)=\dfrac{1}{\epsilon}Q(F,\,F)(\tau,\,\xi,\,v), \quad \tau\in (\tau^n,\,\tau^{n+1/2}] \\[2mm]
		F(\tau^n,\,\xi,\,v)=\tilde{f}(\tau^n,\,\xi,\,v)
	\end{cases}
	\label{eq:collision}
\end{equation}
using the Nanbu algorithm for Maxwellian molecules, see e.g.~\cite{bobylev2000PRE}. In particular, fixing $\xi=\xi_i$ and setting $F_i(\tau,\,v):=F(\tau,\,\xi_i,\,v)$, $\rho_i(\tau):=\int_0^1F_i(\tau,\,v)\,dv$, we observe that~\eqref{eq:collision} can be rewritten as
\begin{equation}
	\partial_\tau F_i(\tau,\,v)=\frac{1}{\epsilon}Q^+(F_i,\,F_i)(\tau,\,v)-\frac{\rho_i(\tau)}{2\epsilon}F_i(\tau,\,v),
	\label{eq:Boltzmann.strong.gain}
\end{equation}
where $Q^+$ denotes the gain part of the binary interaction operator $Q$, i.e.
$$ Q^+(F_i,\,F_i)(\tau,\,v):=\frac{1}{2}\E_\Theta\left[\ave{\int_0^1\frac{1}{\pr{J}}F_i(\tau,\,\pr{v})F_i(\tau,\,\pr{w})\,dw}\right]. $$
Discretising~\eqref{eq:Boltzmann.strong.gain} in time through the forward Euler scheme yields, whenever $\rho_i\neq 0$,
\begin{equation}
	F^{n+1/2}_i=\left(1-\frac{\rho_i\Delta{\tau}}{2\epsilon}\right)F_i^n
		+\frac{\rho_i\Delta{\tau}}{2\epsilon}\cdot\frac{2Q^+(F_i^n,\,F_i^n)}{\rho_i}
	\label{eq:F.Euler}
\end{equation}
whence, from $\int_0^1 F_i^n(v)\,dv=\frac{2}{\rho_i}\int_0^1Q^+(F_i^n,\,F_i^n)(v)\,dv=\rho_i$, we see that also $F_i^{n+1/2}$ has mass $\rho_i$, i.e. the numerical scheme preserves the vehicle density in a single interaction step, provided $\frac{\rho_i\Delta{\tau}}{2\epsilon}\leq 1$. Under such a restriction,~\eqref{eq:F.Euler} has the following interpretation: in a single binary interaction a vehicle with speed $v$ in the position $\xi_i$ either does not change speed with probability $1-\frac{\rho_i\Delta{\tau}}{2\epsilon}$ or changes it according to the rules encoded in $Q^+$ with probability $\frac{\rho_i\Delta{\tau}}{2\epsilon}$.
%In the following we will consider $\Delta \tau = \epsilon$, $\epsilon=10^{-3}$ and in the quasi-invariant scaling for the collision step, for this regime we proved the formal convergence to the Fokker-Planck operator in the homogeneous setting, see also \cite{toscani2006CMS,villani1998ARMA}. 

\item[Transport step.] Subsequently we take the output of the interactions as the input of a pure transport step towards the next time $\tau^{n+1}$:
\begin{equation*}
	\begin{cases}
		\partial_\tau\tilde{f}(\tau,\,\xi,\,v)+v\partial_\xi\tilde{f}(\tau,\,\xi,\,v)=0, \quad \tau\in (\tau^{n+1/2},\,\tau^{n+1}] \\[1mm]
		\tilde{f}(\tau^{n+1/2},\,\xi,\,v)=F(\tau^{n+1/2},\,\xi,\,v).
	\end{cases}
\end{equation*}
\end{description}

\begin{figure}[!t]
\centering
\subfigure[$\tau=0$]{\includegraphics[scale=0.32]{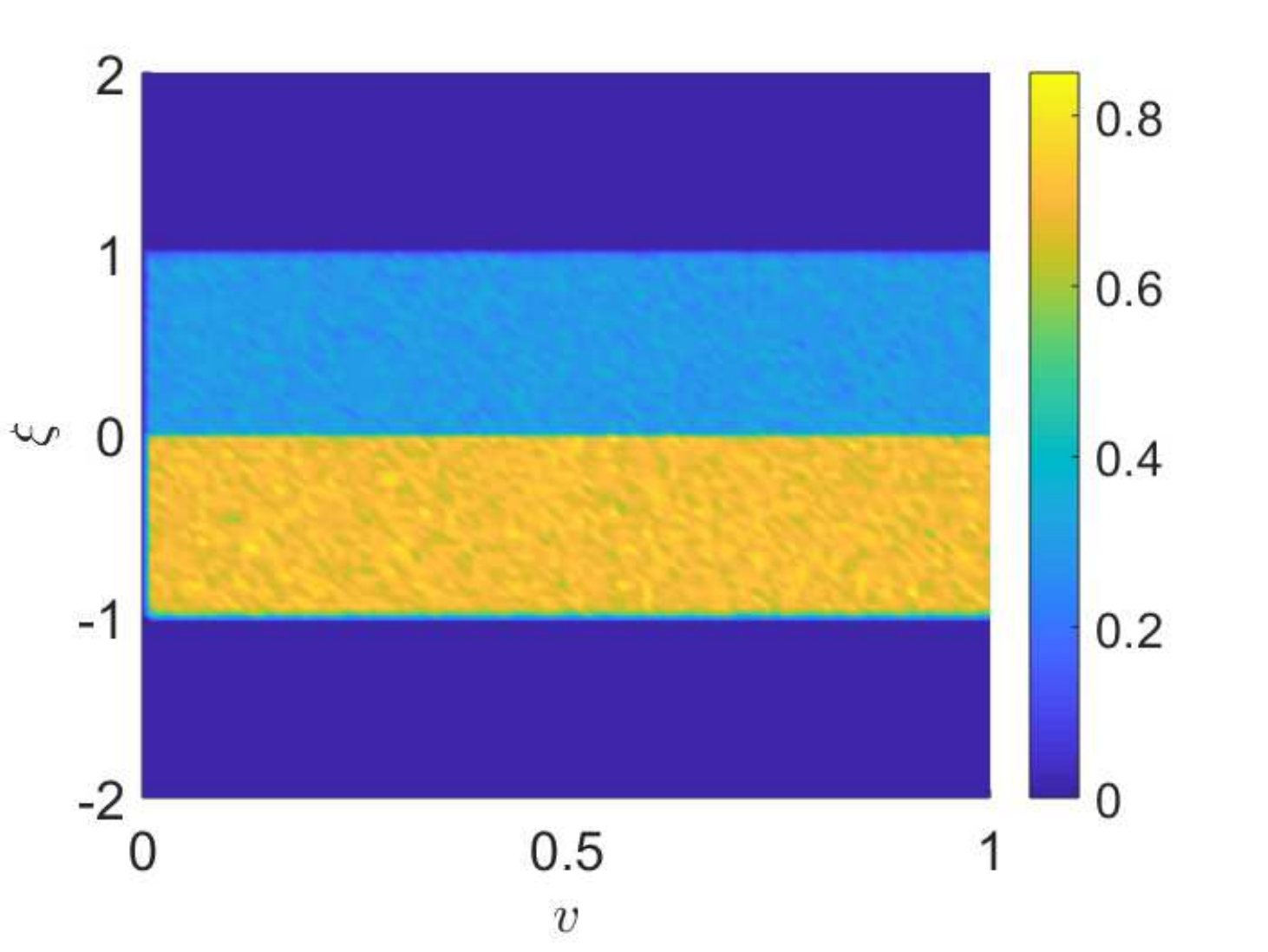}}
\subfigure[$\tau=0.25$]{\includegraphics[scale=0.32]{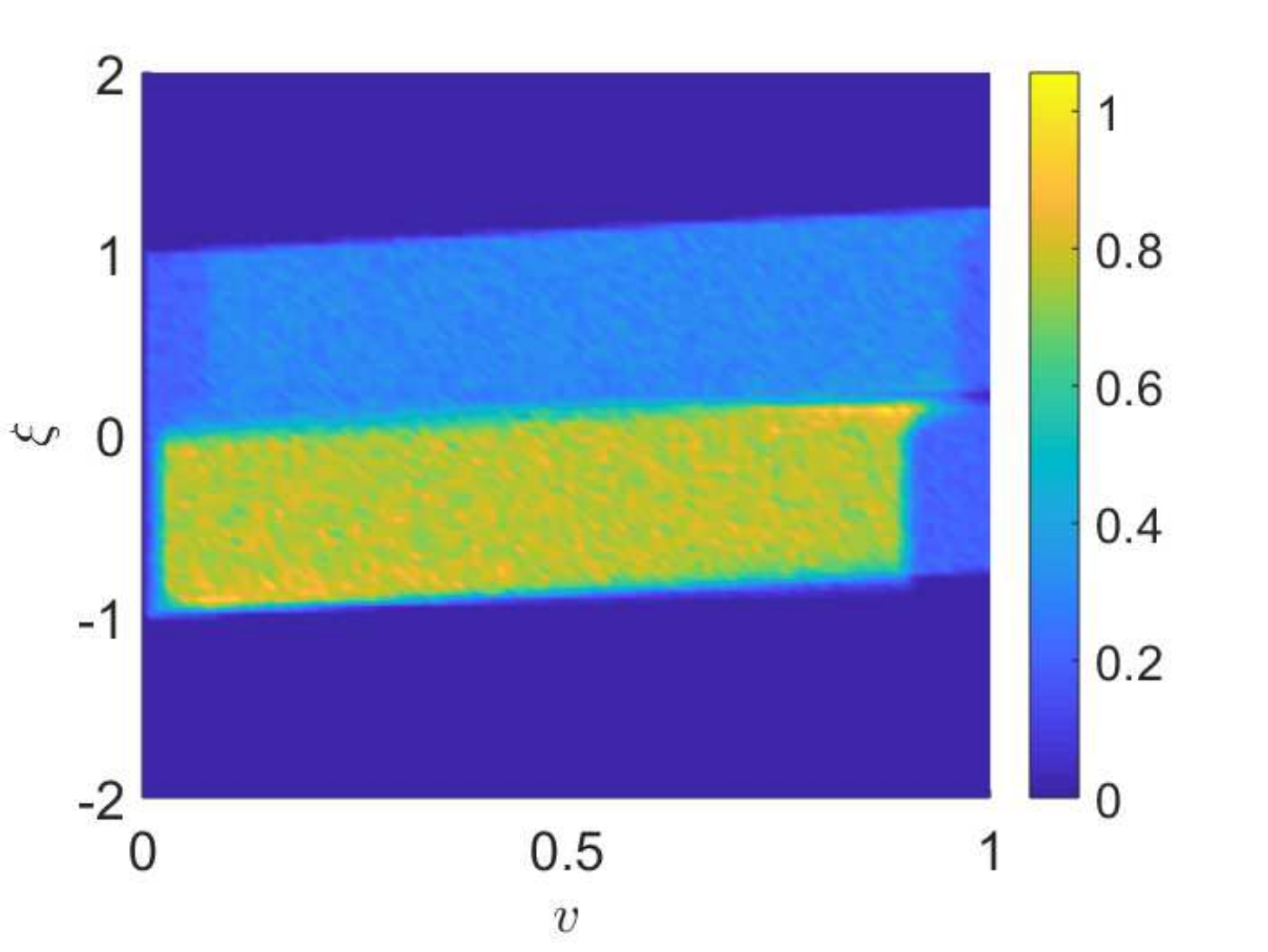}}
\subfigure[$\tau=0.5$]{\includegraphics[scale=0.32]{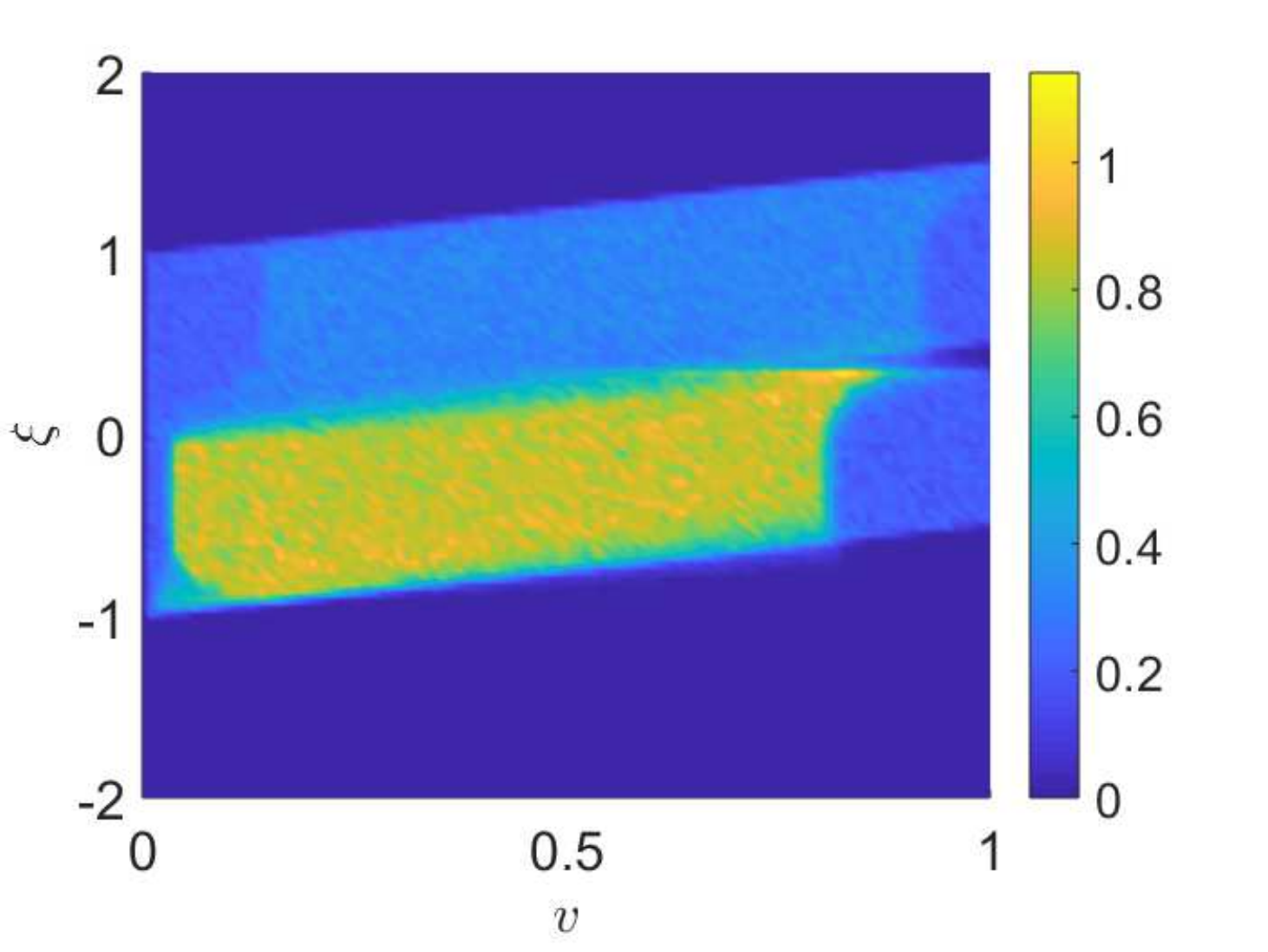}}
\caption{(a) Initial distribution~\eqref{eq:initial_distribution} with $\rho_L=0.8$, $\rho_R=0.2$. (b)-(c) Evolution of the kinetic distribution $\tilde{f}(\tau,\,\xi,\,v)$ for $0<\tau\leq 0.5$ in the case of zero penetration rate $p$ (namely, no vehicle of the system is influenced by the action of the control).}
\label{fig:initial_uncon}
\end{figure}

\begin{figure}[!t]
\centering
\subfigure[$p=\frac{1}{2}$, $\tau=0.25$]{\includegraphics[scale=0.4]{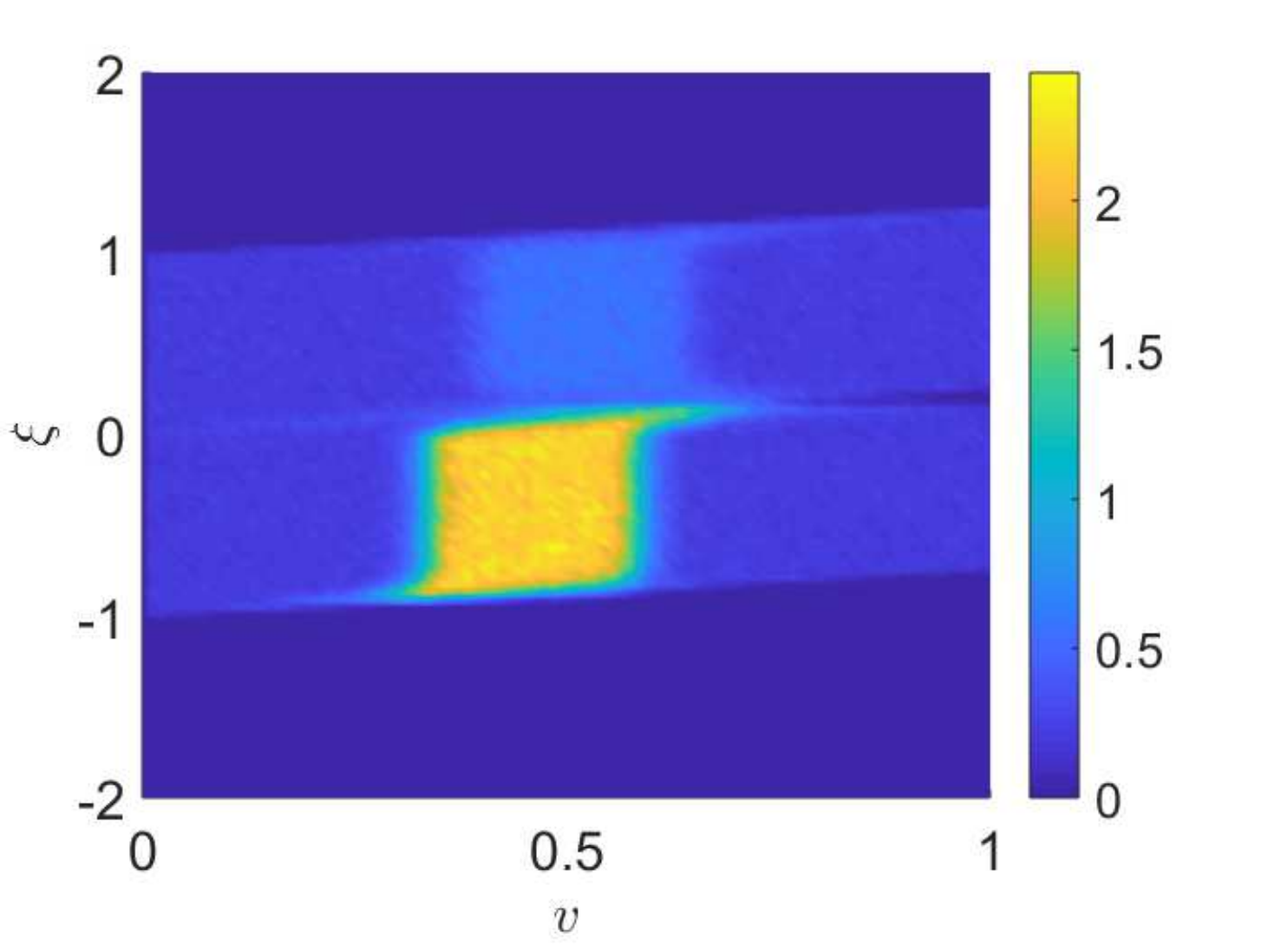}}
\subfigure[$p=\frac{1}{2}$, $\tau=0.5$]{\includegraphics[scale=0.4]{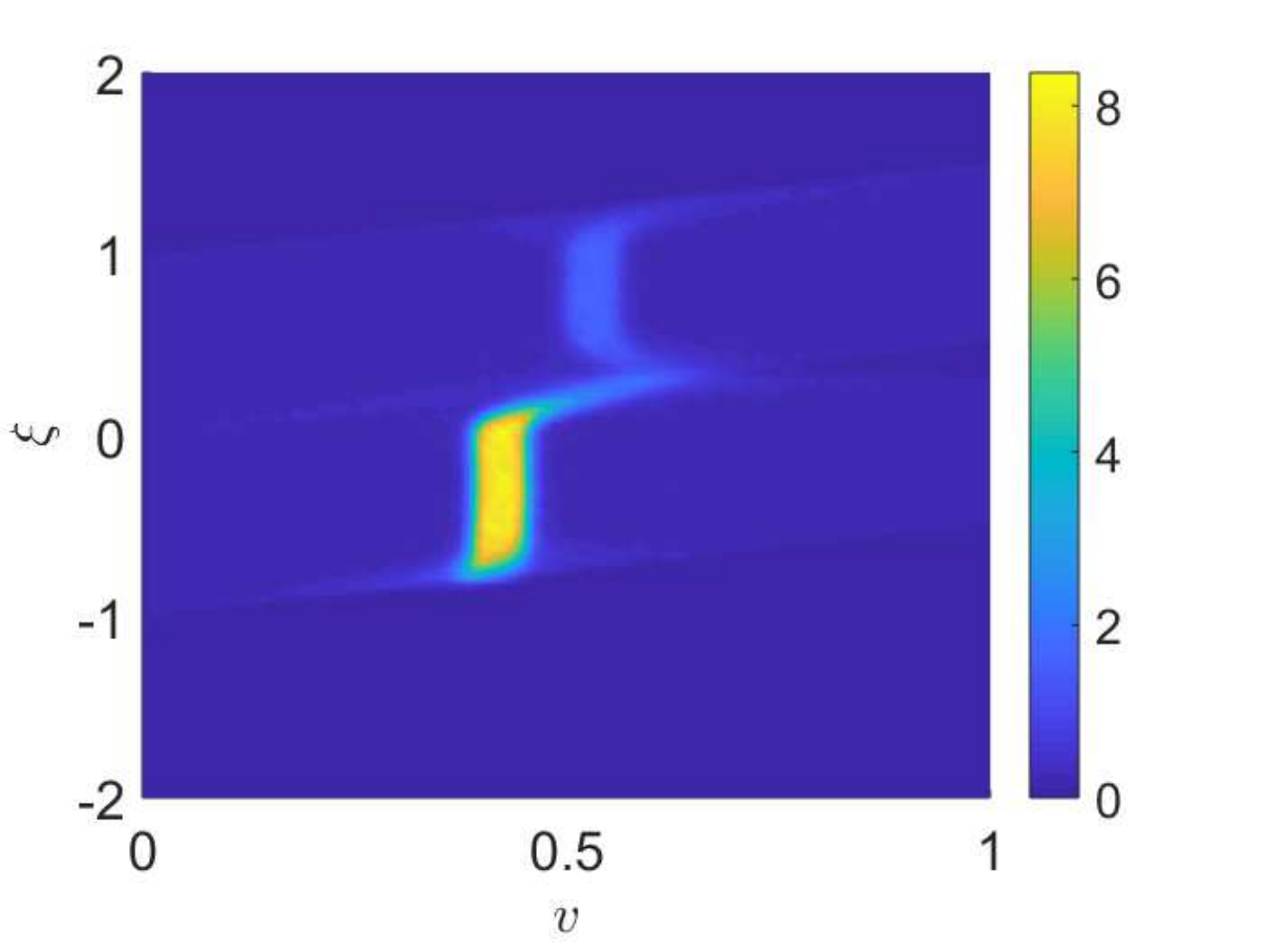}} \\
\subfigure[$p=1$, $\tau=0.25$]{\includegraphics[scale=0.4]{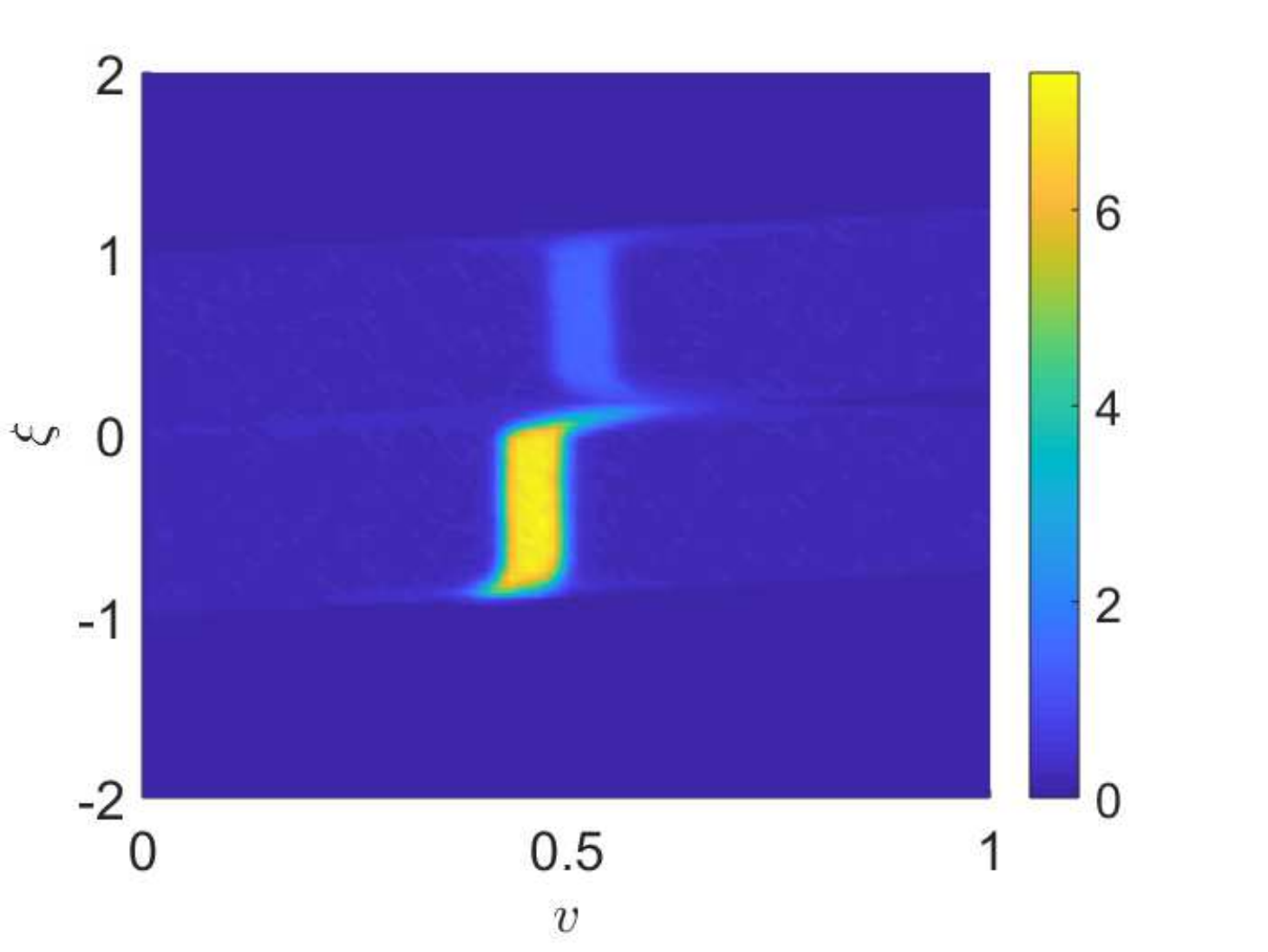}}
\subfigure[$p=1$, $\tau=0.5$]{\includegraphics[scale=0.4]{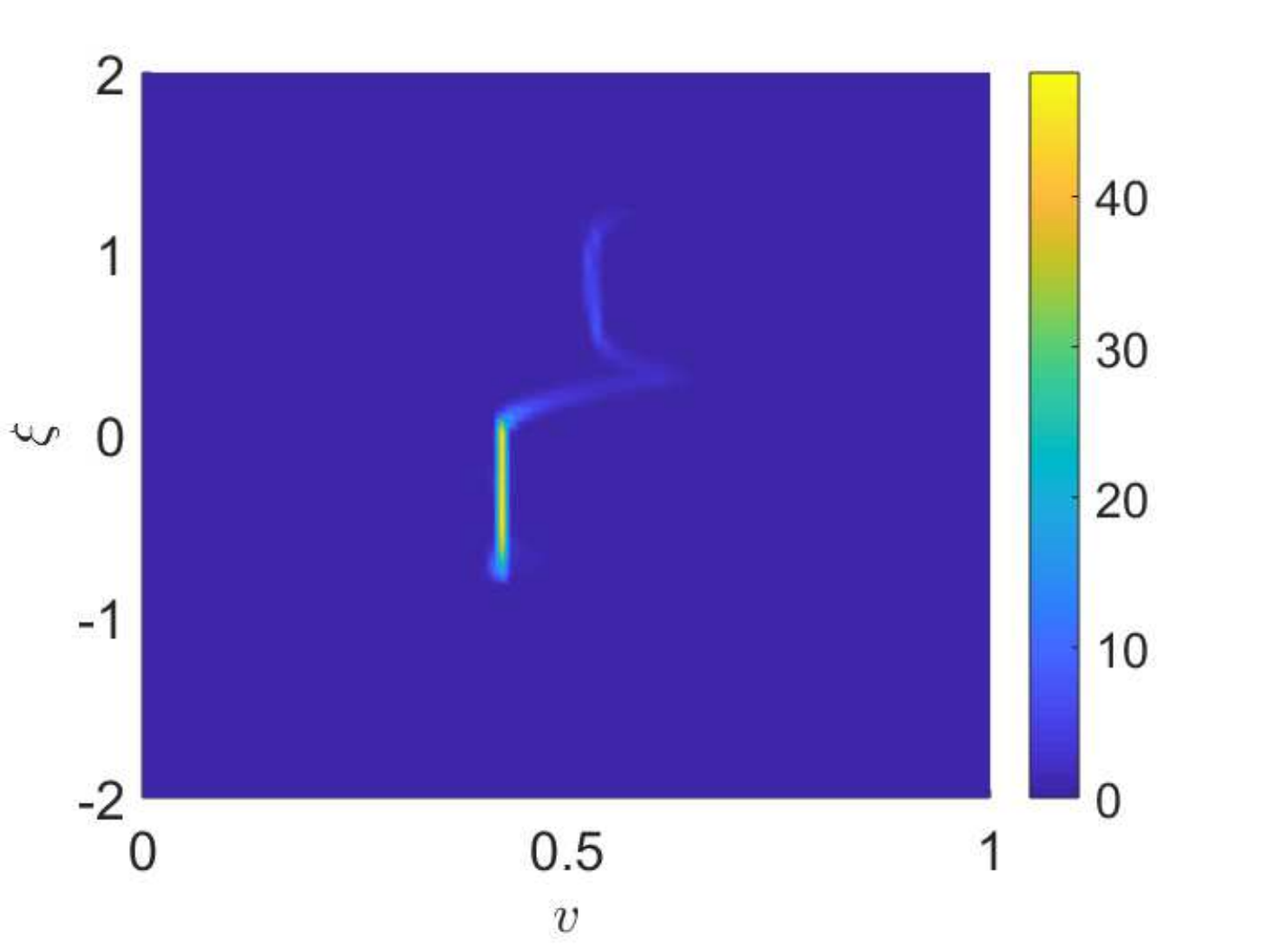}}
\caption{\textbf{Binary variance control}. Evolution of the kinetic distribution $\tilde{f}(\tau,\,\xi,\,v)$ starting from the initial condition depicted in Figure~\ref{fig:initial_uncon}a and with $\kappa=5\cdot 10^{-2}$, $\lambda=10^{-3}$. First row (a, b): penetration rate $p=\frac{1}{2}$. Second row (c, d): penetration rate $p=1$.}
\label{fig:kinetic_binvar}
\end{figure}

On the whole, we consider~\eqref{eq:Boltzmann_strong} in the bounded domain $(\xi,\,v)\in [-2,\,2]\times [0,\,1]$ with periodic boundary conditions on the space variable $\xi$. As initial condition we prescribe the following distribution:
\begin{equation}
	\tilde{f}_0(\xi,\,v)= 
		\begin{cases}
			\rho_L & \text{for } (\xi,\,v)\in [-1,\,0)\times [0,\,1] \\
			\rho_R & \text{for } (\xi,\,v)\in [0,\,1]\times [0,\,1] \\
			0 & \text{otherwise},
		\end{cases}
	\label{eq:initial_distribution}
\end{equation}
which is piecewise constant in $\xi$ and uniform in $v$ for all fixed $\xi\in [-1,\,1]$. The constants $\rho_L,\,\rho_R>0$ represent the vehicle density to the left and to the right, respectively, of the position $\xi=0$ and are chosen in such a way that the total mass of vehicles is normalised to $1$, i.e.:
$$ \rho_L+\rho_R=\int_0^1\int_{-2}^2\tilde{f}_0(\xi,\,v)\,d\xi\,dv=1, $$
see Figure~\ref{fig:initial_uncon}a for a particular case.

Figures~\ref{fig:initial_uncon}b, c show the evolution of the distribution function $\tilde{f}$ at two successive times in the case $p=0$, i.e. for a null penetration rate meaning that no vehicle in the traffic stream is equipped with driver-assist technologies. We clearly observe that the vehicle speeds are highly dispersed at the final time, hence that the associated driving risk does not tend to decrease spontaneously.

\begin{figure}[!t]
\centering
\subfigure[$p=\frac{1}{2}$, $\tau=0.25$]{\includegraphics[scale=0.4]{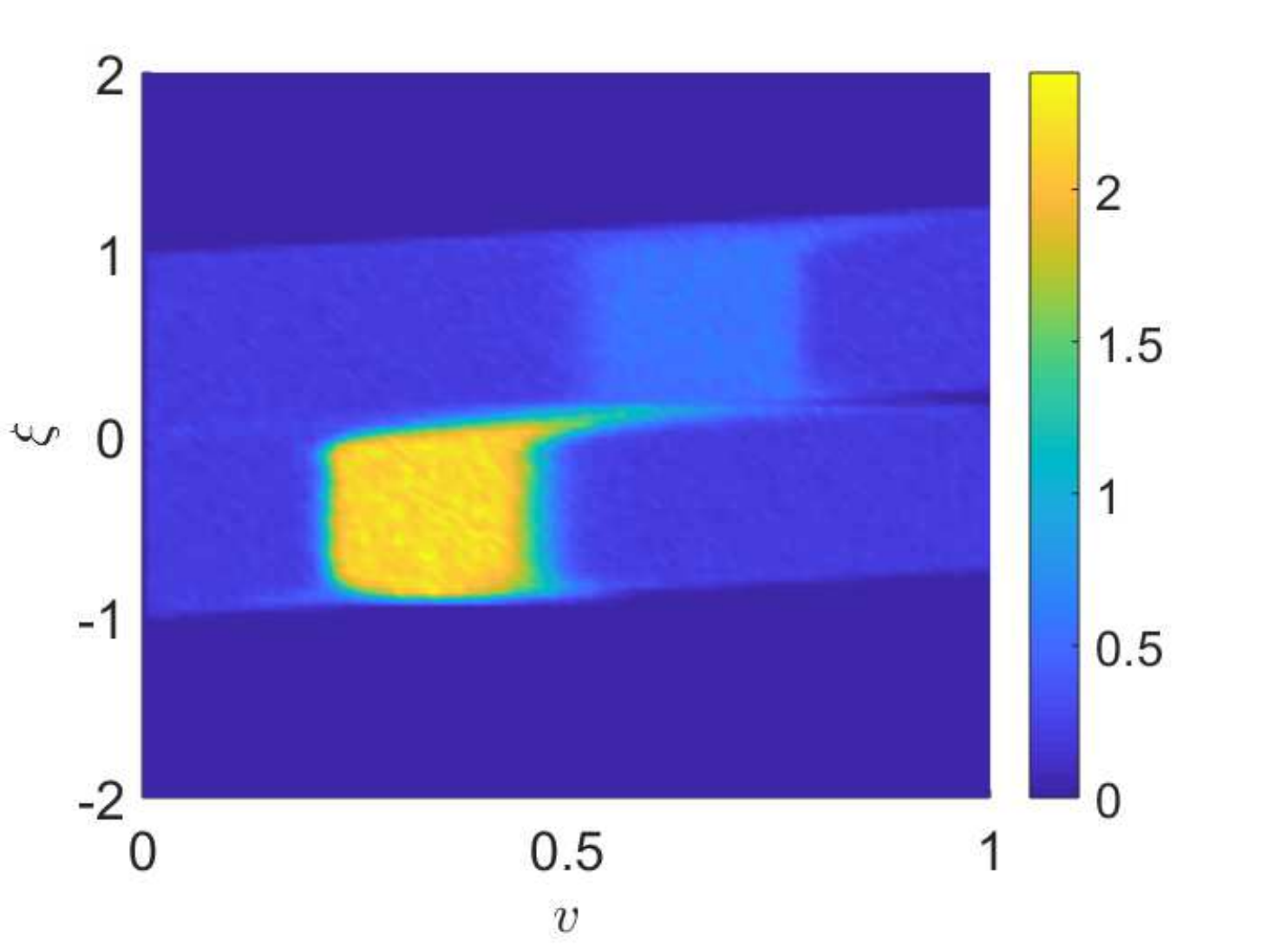}}
\subfigure[$p=\frac{1}{2}$, $\tau=0.5$]{\includegraphics[scale=0.4]{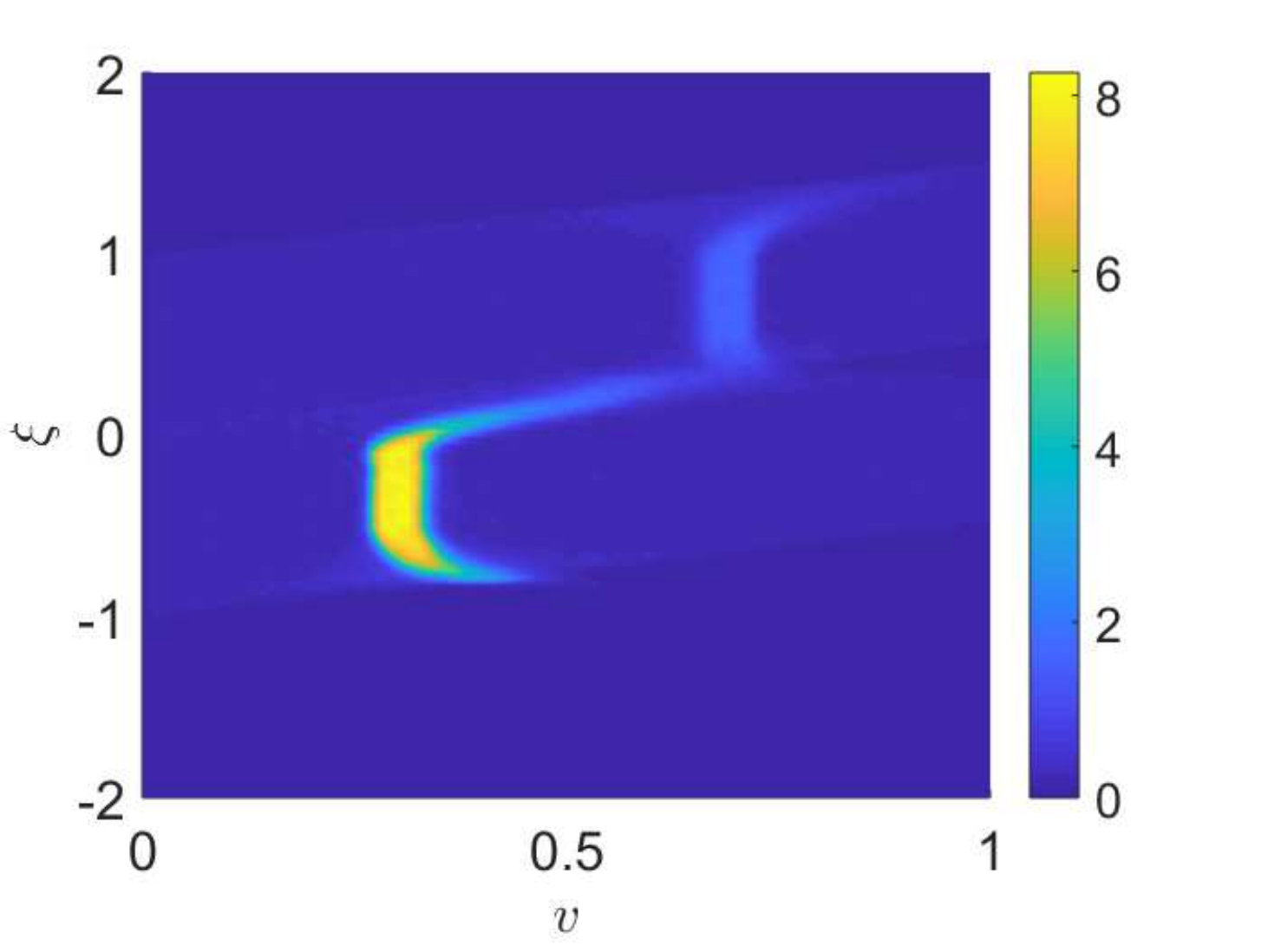}} \\
\subfigure[$p=1$, $\tau=0.25$]{\includegraphics[scale=0.4]{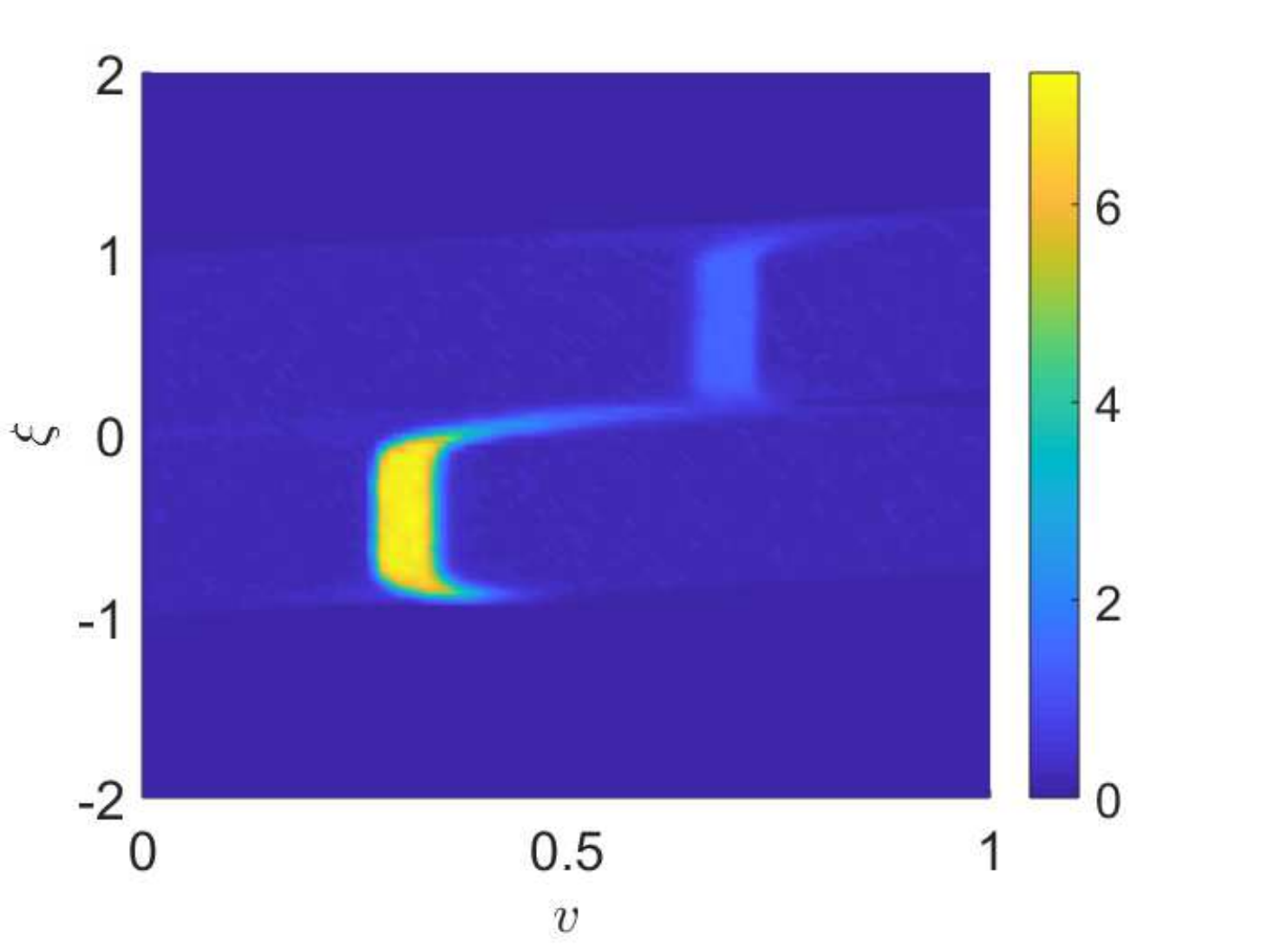}}
\subfigure[$p=1$, $\tau=0.5$]{\includegraphics[scale=0.4]{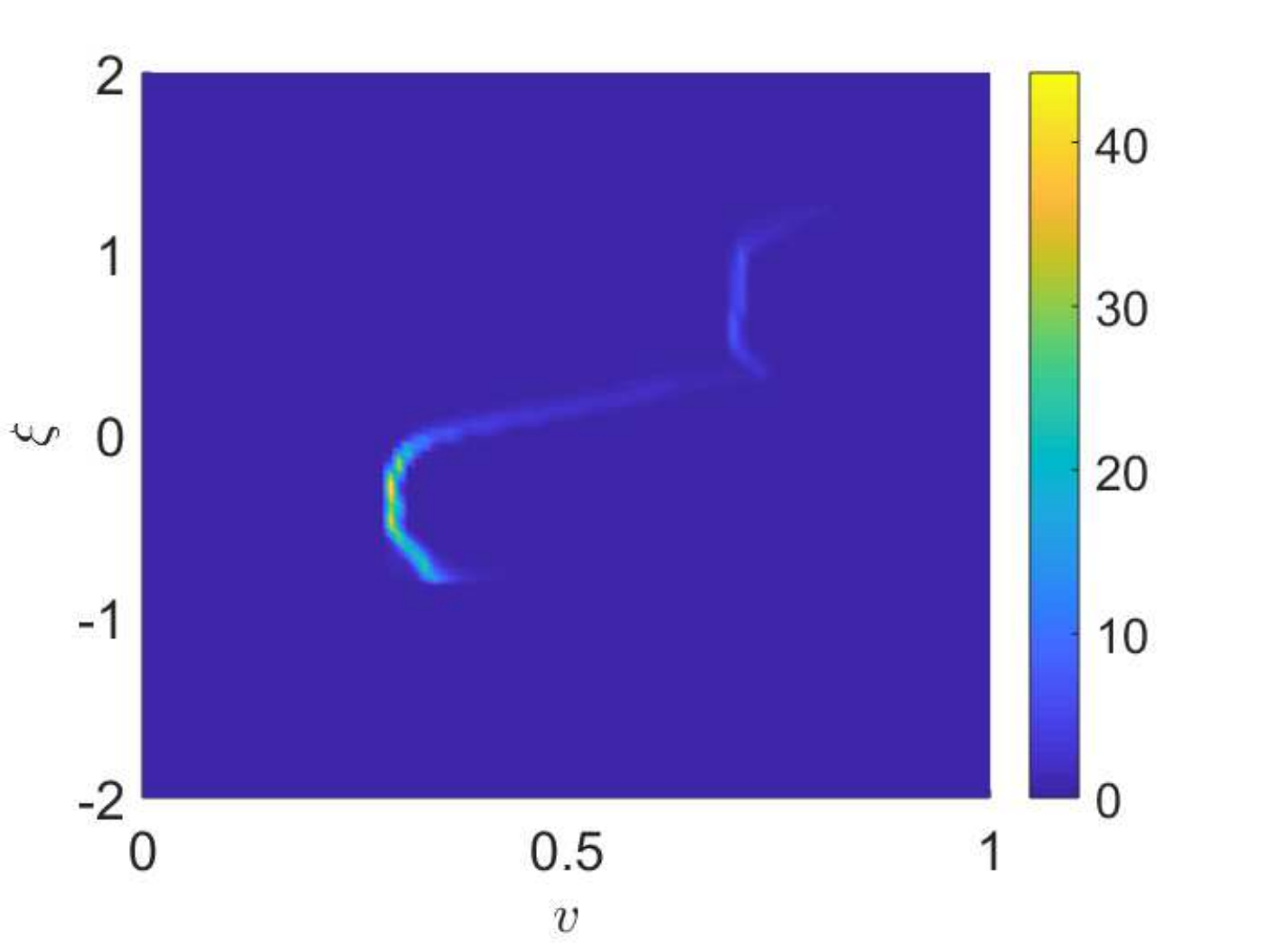}}
\caption{\textbf{Desired speed control}. Evolution of the kinetic distribution $\tilde{f}(\tau,\,\xi,\,v)$ starting from the initial condition depicted in Figure~\ref{fig:initial_uncon}a, with $\kappa=5\cdot 10^{-2}$, $\lambda=10^{-3}$ and $v_d(\rho)=1-\rho$. First row (a, b): penetration rate $p=\frac{1}{2}$. Second row (c, d): penetration rate $p=1$.}
\label{fig:kinetic_vd}
\end{figure}

Figures~\ref{fig:kinetic_binvar},~\ref{fig:kinetic_vd} show instead the evolution of $\tilde{f}$ under the action of a driver-assist control which seeks to minimise either the binary variance of the speeds in each pairwise interaction or the difference with the congestion-dependent recommended speed $v_d(\rho)=1-\rho$, respectively. In both cases we consider the scenarios with either $50\%$ or $100\%$ of vehicles equipped with driver-assist technologies in the traffic stream, corresponding to penetration rates $p=\frac{1}{2}$ and $p=1$, respectively, for a fixed control penalisation $\kappa=5\cdot 10^{-2}$. The numerical results show that both control strategies manage to reduce the global speed variance. Indeed at the final time the distribution function $\tilde{f}$ clearly approaches Dirac delta-like distributions in the $v$-variable, particularly for $p=1$.

\begin{figure}[!t]
\centering
\includegraphics[scale=0.4]{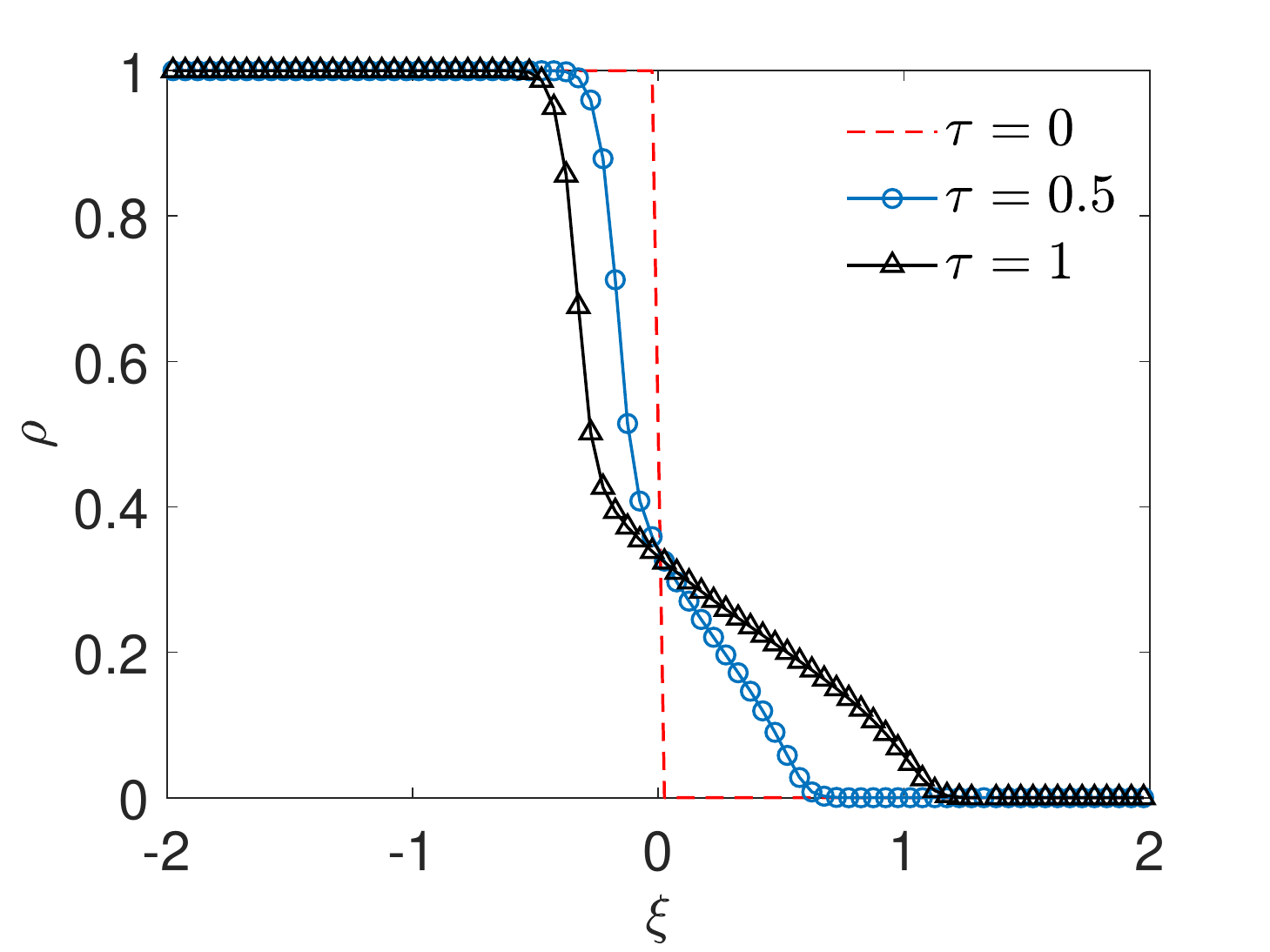}
\caption{Solution to~\eqref{eq:macro_binvar} at times $\tau=0$ (initial condition) and $\tau=0.5,\,1$.}
\label{fig:hydro_binvar}
\end{figure}

\subsection{First order hydrodynamic model}
We already observed that the flux function $\F(\rho):=\rho V_\infty^\ast(\rho)$ in the conservation law~\eqref{eq:first.order} is in general neither strictly convex nor strictly concave, cf. in particular~\eqref{eq:macro_binvar},~\eqref{eq:macro_vdes}, because $\F''(\rho)$ may change sign for $\rho\in (0,\,1)$. Hence the solution to a Riemann problem is expected to be a combination of shock and rarefaction waves, sometimes called a \textit{compound wave}. The unique entropy solution can be built taking advantage of a convex-hull reconstruction, see e.g.~\cite{leveque2004BOOK} for an introduction.

At the numerical level, it is quite challenging to prove the convergence of high-order schemes to the entropy solution despite the good numerical performances of such schemes in many regimes. In the non-convex case counterexamples exist for Godunov and Lax-Friedrichs numerical fluxes under the usual CFL condition. In the following we use a WENO finite volume scheme~\cite{shu2009SIREV} with Lax-Friedrichs numerical flux. In order to enforce the convergence towards the correct entropy solution we adopt the first order monotone modification proposed in~\cite{qiu2008SISC}. The resulting method is high-order accurate in smooth regions, whereas near a non-convex discontinuity region it uses a discontinuity indicator which is $O(1)$.

To discretise equations~\eqref{eq:macro_binvar},~\eqref{eq:macro_vdes} we introduce a uniform mesh in the space domain $[-2,\,2]$ made of $N_{\xi}=80$ grid points, which implies a mesh parameter $\Delta{\xi}=5\cdot 10^{-2}$. Furthermore, we choose the time step $\Delta{\tau}>0$ according to the CFL condition:
$$ \max_{\xi\in [-2,\,2]}\abs{\F'(\rho)}\cdot\frac{\Delta{\tau}}{\Delta{\xi}}=\text{CFL}\leq 1 $$
fixing $\text{CFL}=0.5$.

In all the tests of this section we prescribe the following initial condition for the vehicle density:
$$	\rho_0(\xi)= 
		\begin{cases}
			1 & \text{for } \xi\leq 0 \\
			0 & \text{for } \xi>0,
		\end{cases}
$$
which reproduces the classic example of a queue upstream a traffic light placed in $\xi=0$ which turns to green at time $\tau=0$. Figure~\ref{fig:hydro_binvar} shows such initial condition (dashed red line) and the evolution of the vehicle density $\rho$ at two successive times ruled by~\eqref{eq:macro_binvar}. Notice that this problem is representative at once of three different scenarios:
\begin{enumerate*}[label=(\textit{\roman*})]
\item the case of completely uncontrolled dynamics;
\item the case of binary variance control with any penetration rate, which, as already observed, has no visible impact on the purely macroscopic stream of vehicles;
\item the case of desired speed control with zero penetration rate, for then equation~\eqref{eq:macro_vdes} reduces to~\eqref{eq:macro_binvar}.
\end{enumerate*}
Since we fixed $\mu=2$ in~\eqref{eq:P}, the flux of~\eqref{eq:macro_binvar} is non-concave, hence the solution is a combination of a backward propagating shock and a rarefaction wave modelling vehicles which progressively depart at the green light.

\begin{figure}[!t]
\centering
\subfigure[$\kappa=1$, $p=\frac{1}{4}$]{\includegraphics[scale=0.32]{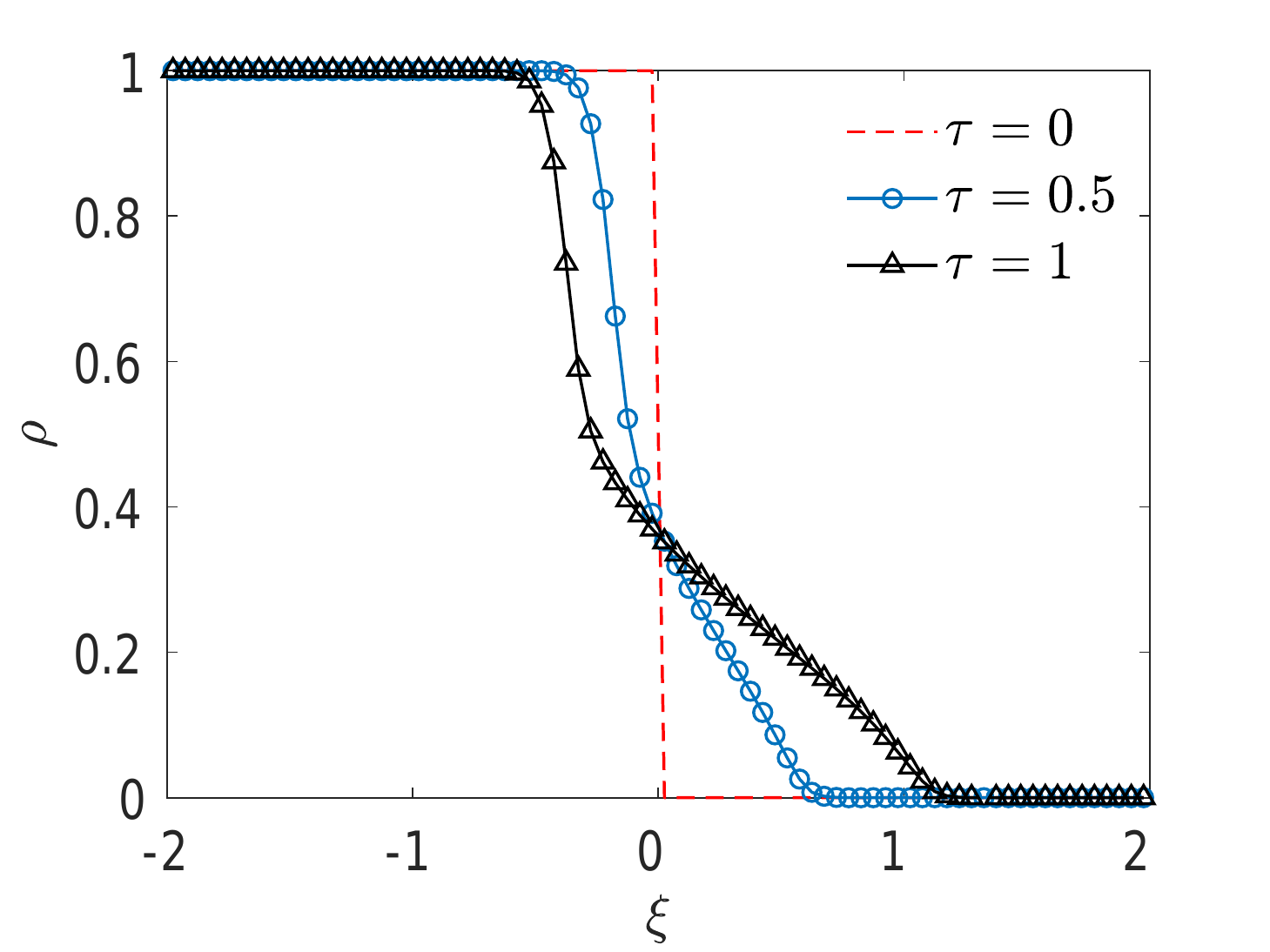}}
\subfigure[$\kappa=1$, $p=\frac{1}{2}$]{\includegraphics[scale=0.32]{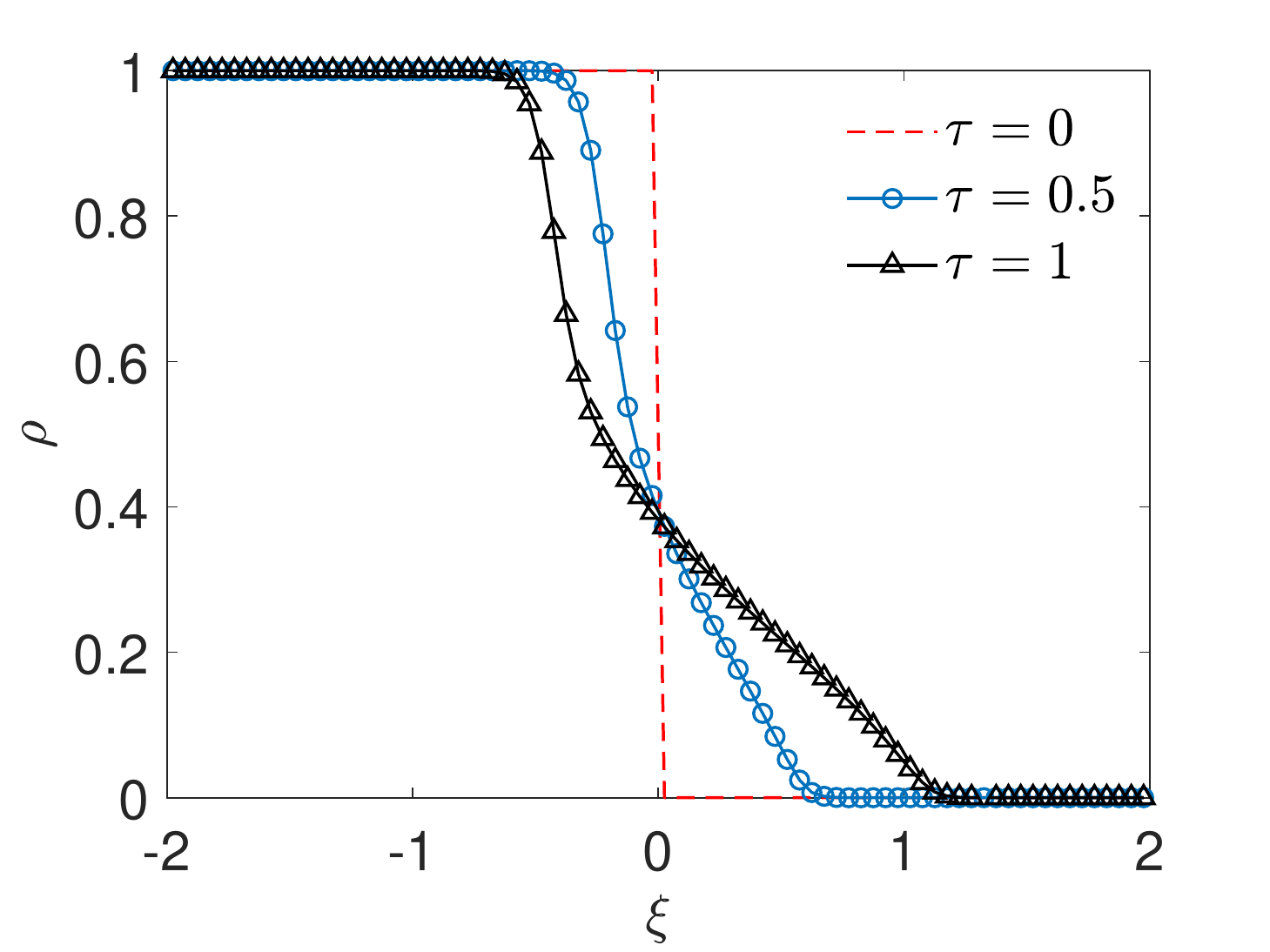}}
\subfigure[$\kappa=1$, $p=1$]{\includegraphics[scale=0.32]{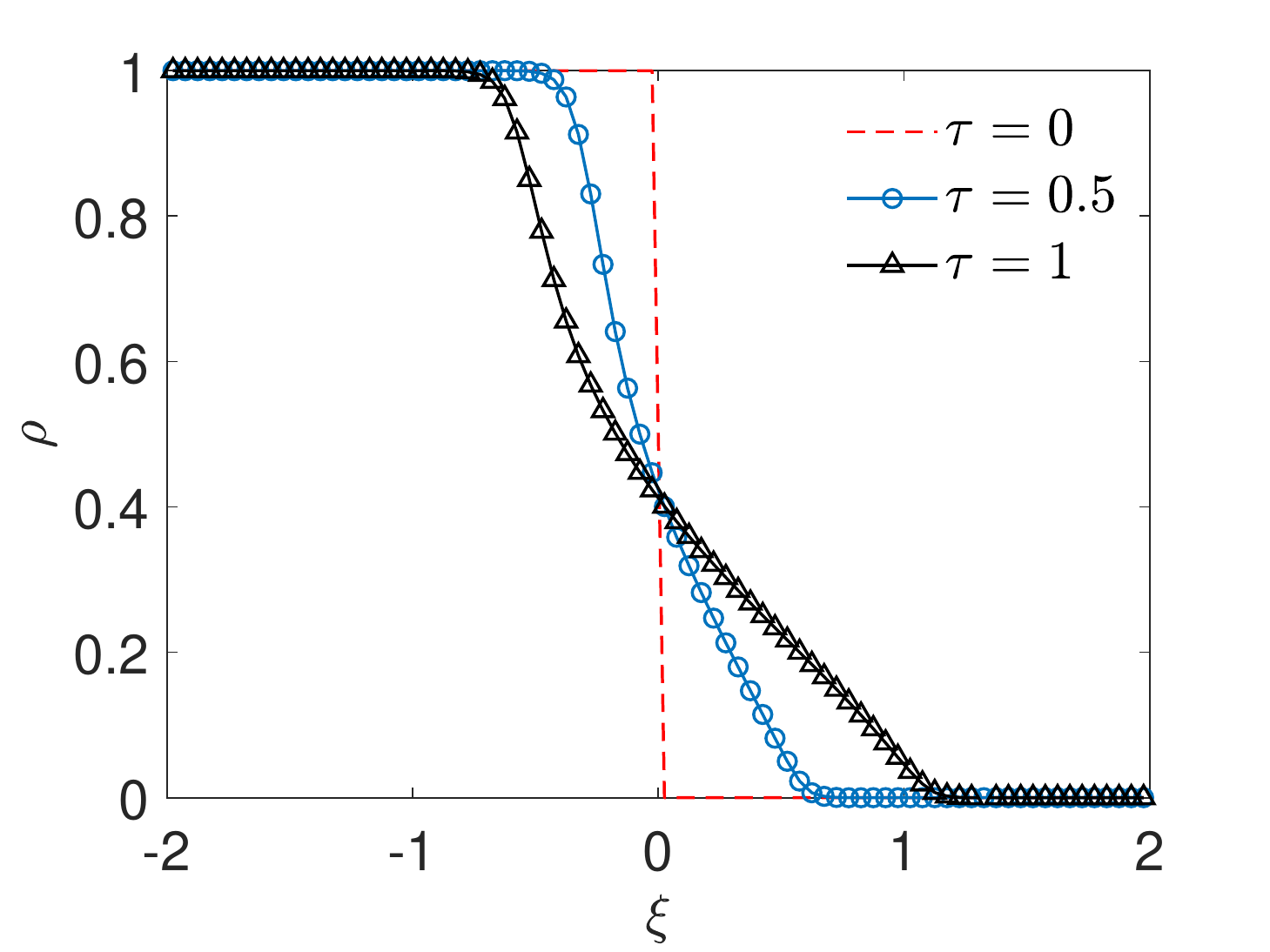}} \\
\subfigure[$\kappa=10^{-1}$, $p=\frac{1}{4}$]{\includegraphics[scale=0.32]{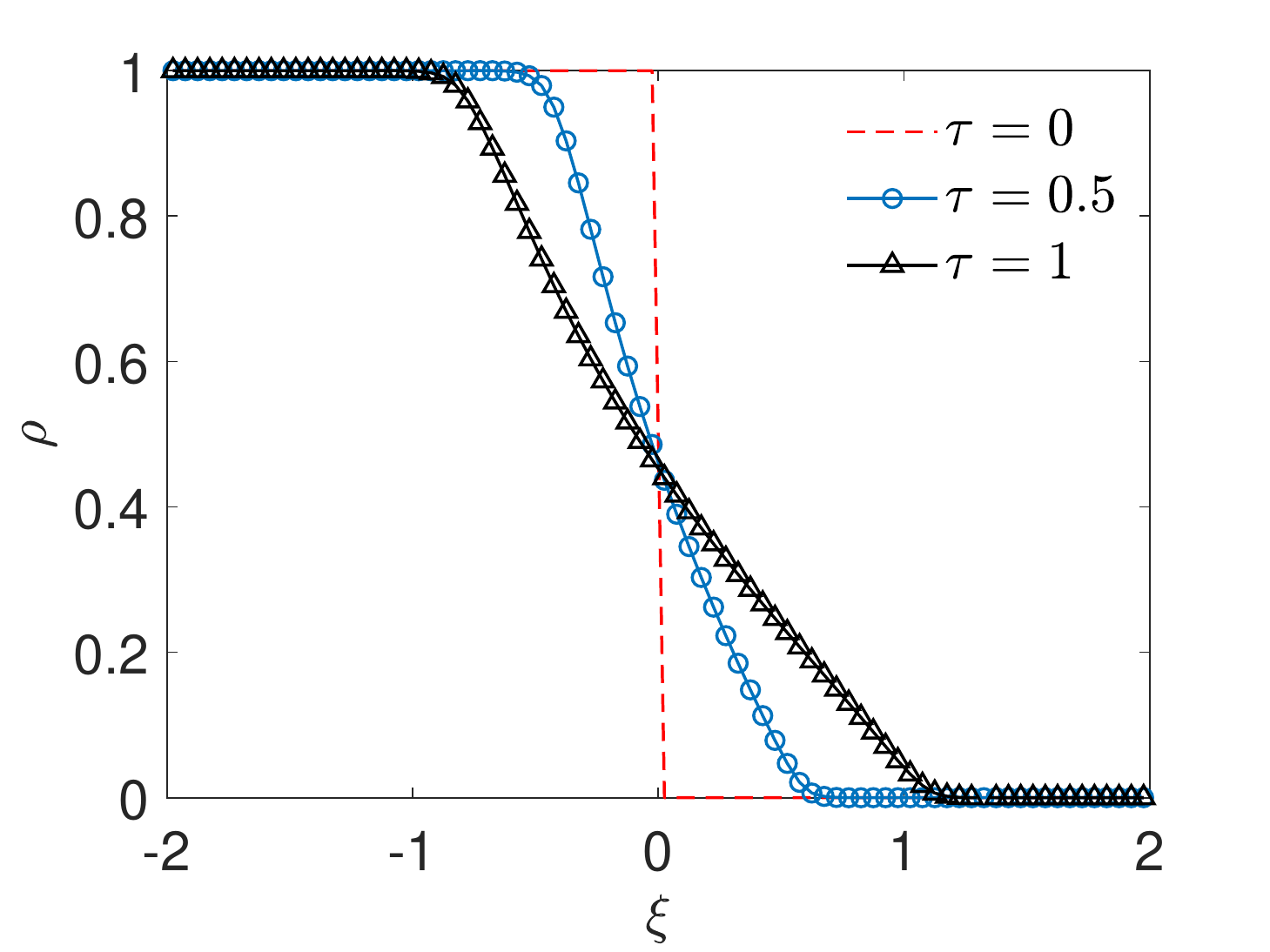}}
\subfigure[$\kappa=10^{-1}$, $p=\frac{1}{2}$]{\includegraphics[scale=0.32]{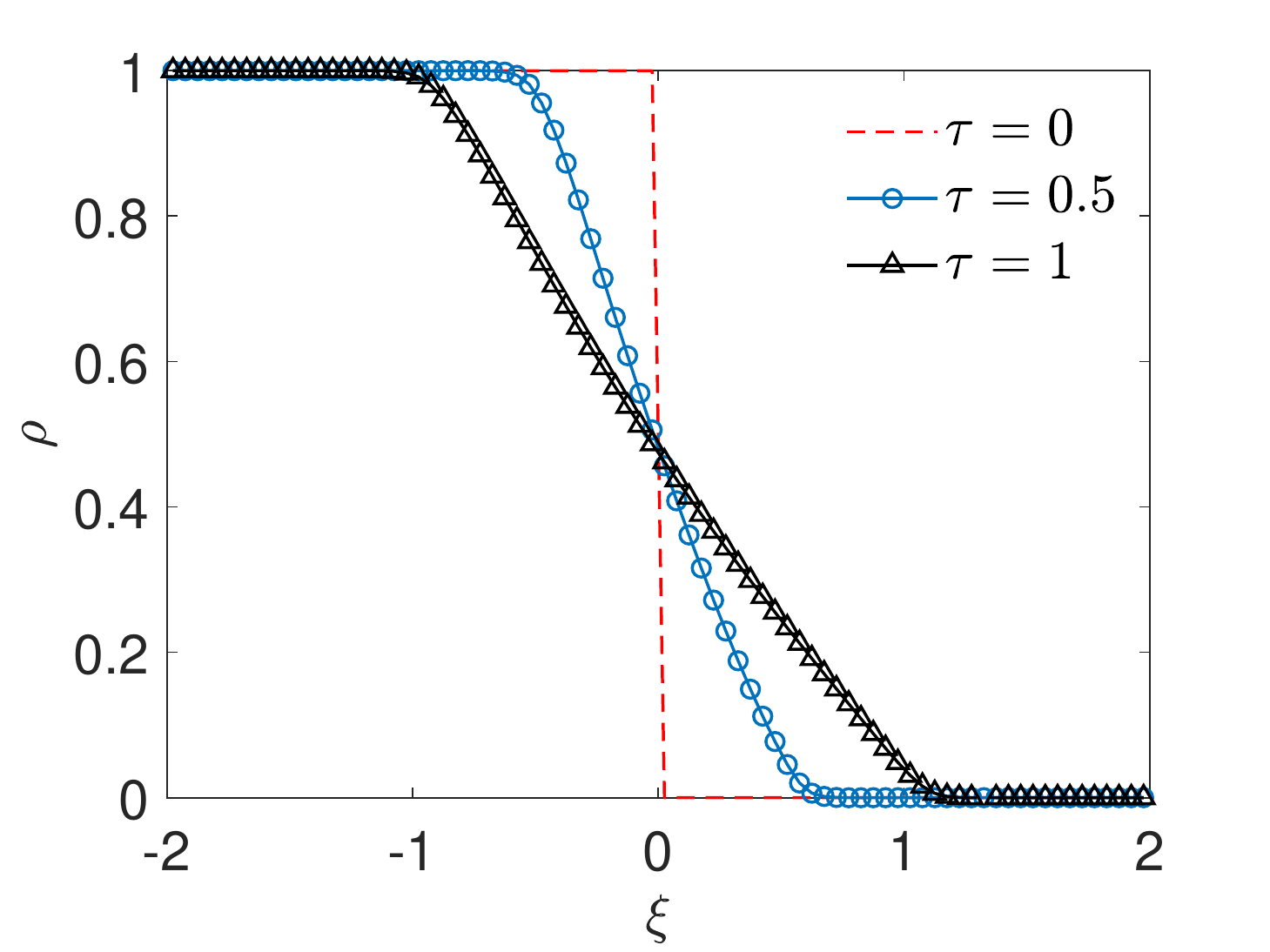}}
\subfigure[$\kappa=10^{-1}$, $p=1$]{\includegraphics[scale=0.32]{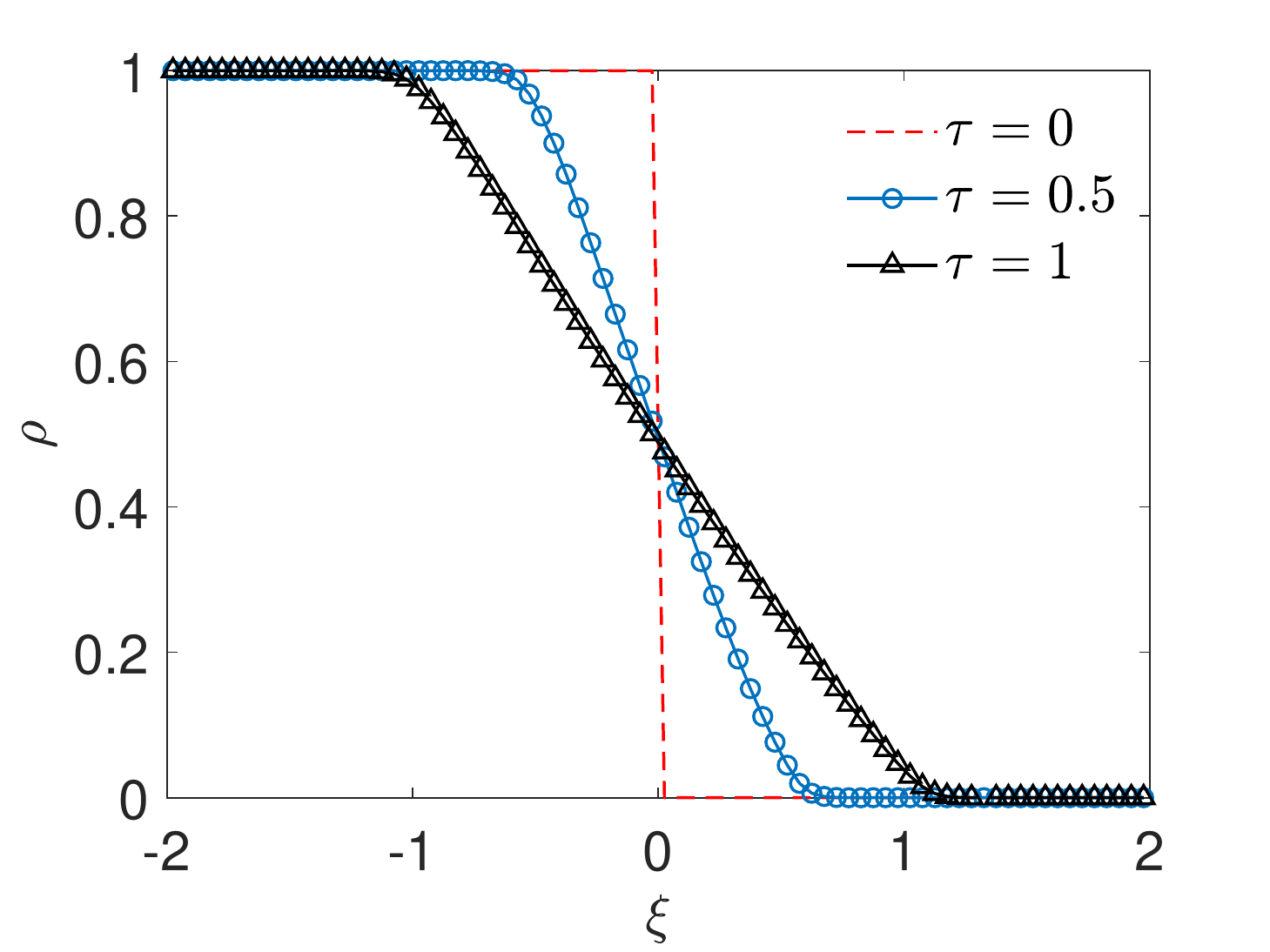}}
\caption{\textbf{Desired speed control}. First row: solution to~\eqref{eq:macro_vdes} at times $\tau=0$ (initial condition) and $\tau=0.5,\,1$ for the three penetration rates $p=\frac{1}{4},\,\frac{1}{2},\,1$ in the case $\kappa=1$ (strongly penalised control). Second row: same solution in the case $\kappa=10^{-1}$ (weakly penalised control).}
\label{fig:vdes_kappa}
\end{figure}

Figure~\ref{fig:vdes_kappa} shows instead the solution to~\eqref{eq:macro_vdes}, namely the first order hydrodynamic model with desired speed control, for three different penetration rates ($p=\frac{1}{4},\,\frac{1}{2},\,1$) and two choices of the penalisation coefficient ($\kappa=10^{-1},\,1$). As already stated, the recommended speed is set to $v_d(\rho)=1-\rho$, therefore in the limit $\kappa\to 0^+$ (non-penalised control) the flux of~\eqref{eq:macro_vdes} tends to the classical Greenshield's parabolic one $\rho(1-\rho)$, which gives a pure rarefaction wave as solution to the traffic light problem. From the second row of Figure~\ref{fig:vdes_kappa} we clearly observe that for $\kappa=10^{-1}$ (weakly penalised control) the density profile approaches indeed the expected one: the higher the penetration rate $p$ the more the shock visible in Figure~\ref{fig:hydro_binvar} is absorbed by the action of the control, so that the whole evolution is consistent with pure rarefaction dynamics. We stress that the shock is instead still present in case of a more strongly penalised control, see the first row of Figure~\ref{fig:vdes_kappa} where $\kappa=1$. Nevertheless, it is slightly smoothed with respect to Figure~\ref{fig:hydro_binvar} for a sufficiently high penetration rate $p$. Obviously, the aforesaid choice of $v_d(\rho)$ is just a possible example. It can be replaced by other more elaborated forms of the recommended speed, such as e.g. those proposed in~\cite{frejo2016IEEE,han2017TRC} for a similar problem.

\subsection{Second order hydrodynamic model}
The second order hydrodynamic models derived in Section~\ref{sect:monokinetic} consist in pressureless and isothermal Euler-type equations with a reaction term describing a relaxation towards the local equilibrium speed predicted by the kinetic model, cf.~\eqref{eq:macro_second_binvar}, plus possibly a further relaxation towards the speed induced by the microscopic control, cf.~\eqref{eq:macro_second_vdes}. Pressureless systems of balance laws have been studied at both the theoretical and the numerical levels by several authors in recent years, we mention among others~\cite{bouchut1994BOOKCH,bouchut2003SINUM} and the references therein. One of the typical difficulties is that pressureless systems are weakly hyperbolic, which, in the absence of source terms, causes the emergence of vacuum states in a finite time. As a consequence, in order to ensure stability numerical methods would require a time step tending to zero.

In the following we solve numerically systems~\eqref{eq:macro_second_binvar} and~\eqref{eq:macro_second_vdes} by means of an operator splitting approach in the space domain $[-2,\,2]$ discretised by means of $N_\xi=80$ grid points, which implies a mesh parameter $\Delta{\xi}=5\cdot 10^{-2}$. We impose $\text{CFL}=0.5$ for the choice of the time step $\Delta{\tau}$. Moreover, we observe that for a vanishing control penalisation $\kappa$ the source term in the second equation of~\eqref{eq:macro_second_vdes} becomes stiff, thereby leading to additional constraints on the choice of the time step. %of the form $\abs{\alpha\Delta{\tau}}=\kappa\text{CFL}$.

\begin{figure}[!t]
\centering
\includegraphics[scale=0.4]{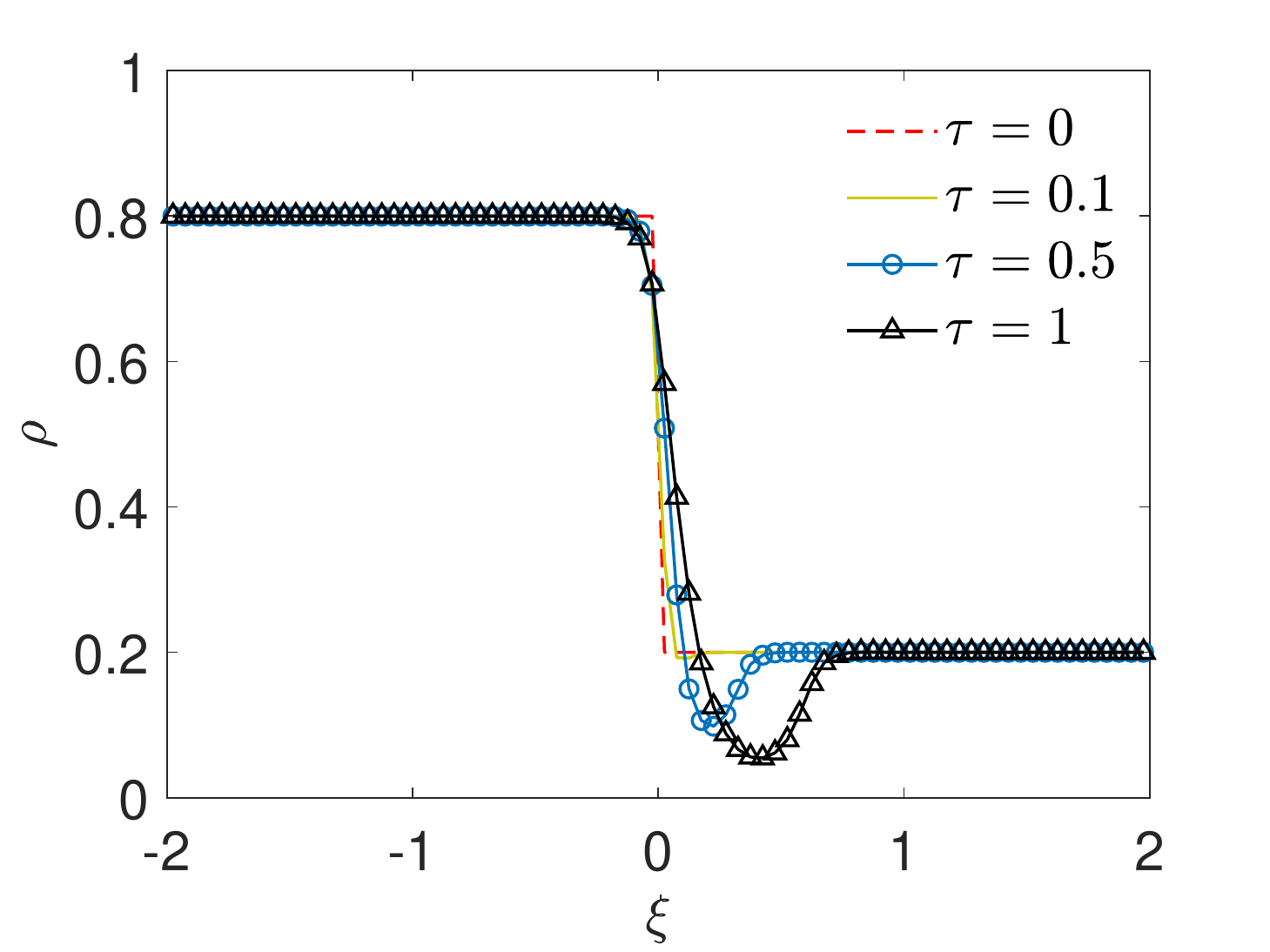}
\includegraphics[scale=0.4]{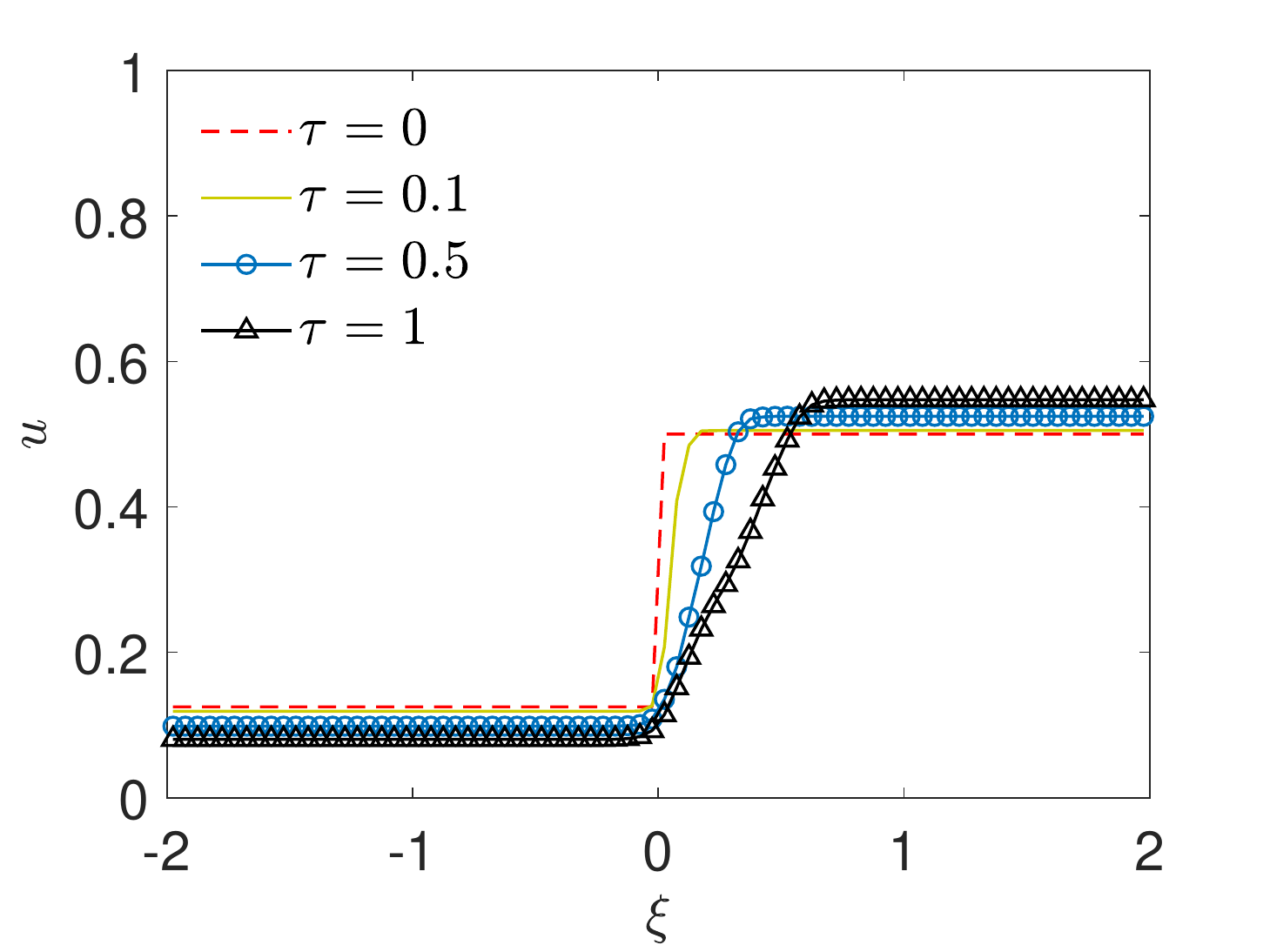}
\caption{Solution to~\eqref{eq:macro_second_binvar} at times $\tau=0$ (initial condition) and $\tau=0.1,\,0.5,\,1$.}
\label{fig:mono_1}
\end{figure}
\begin{figure}[!t]
\centering
\subfigure[$\rho(\tau,\,\xi)$ with $\kappa=1$]{\includegraphics[scale=0.4]{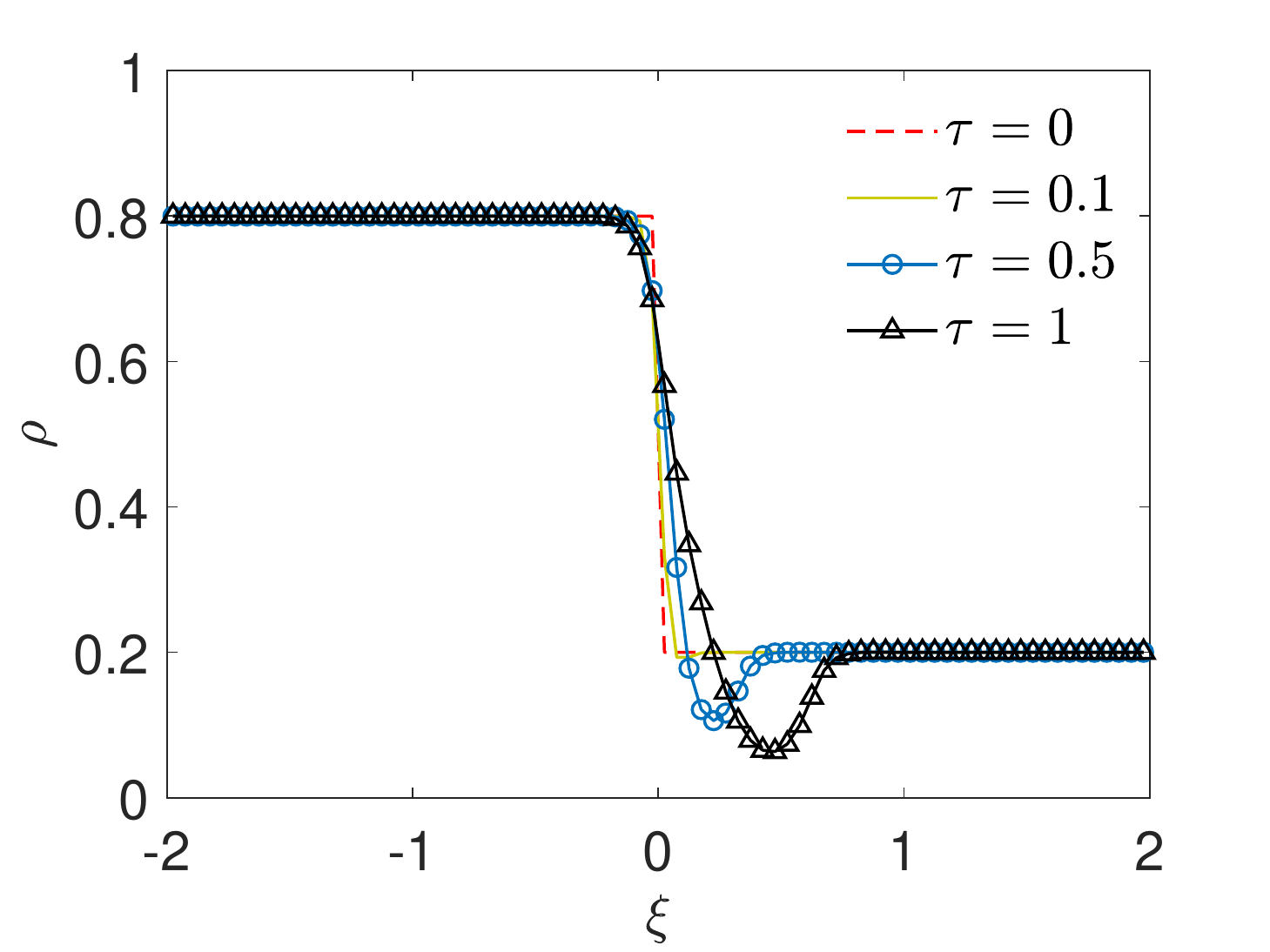}}
\subfigure[$u(\tau,\,\xi)$ with $\kappa=1$]{\includegraphics[scale=0.4]{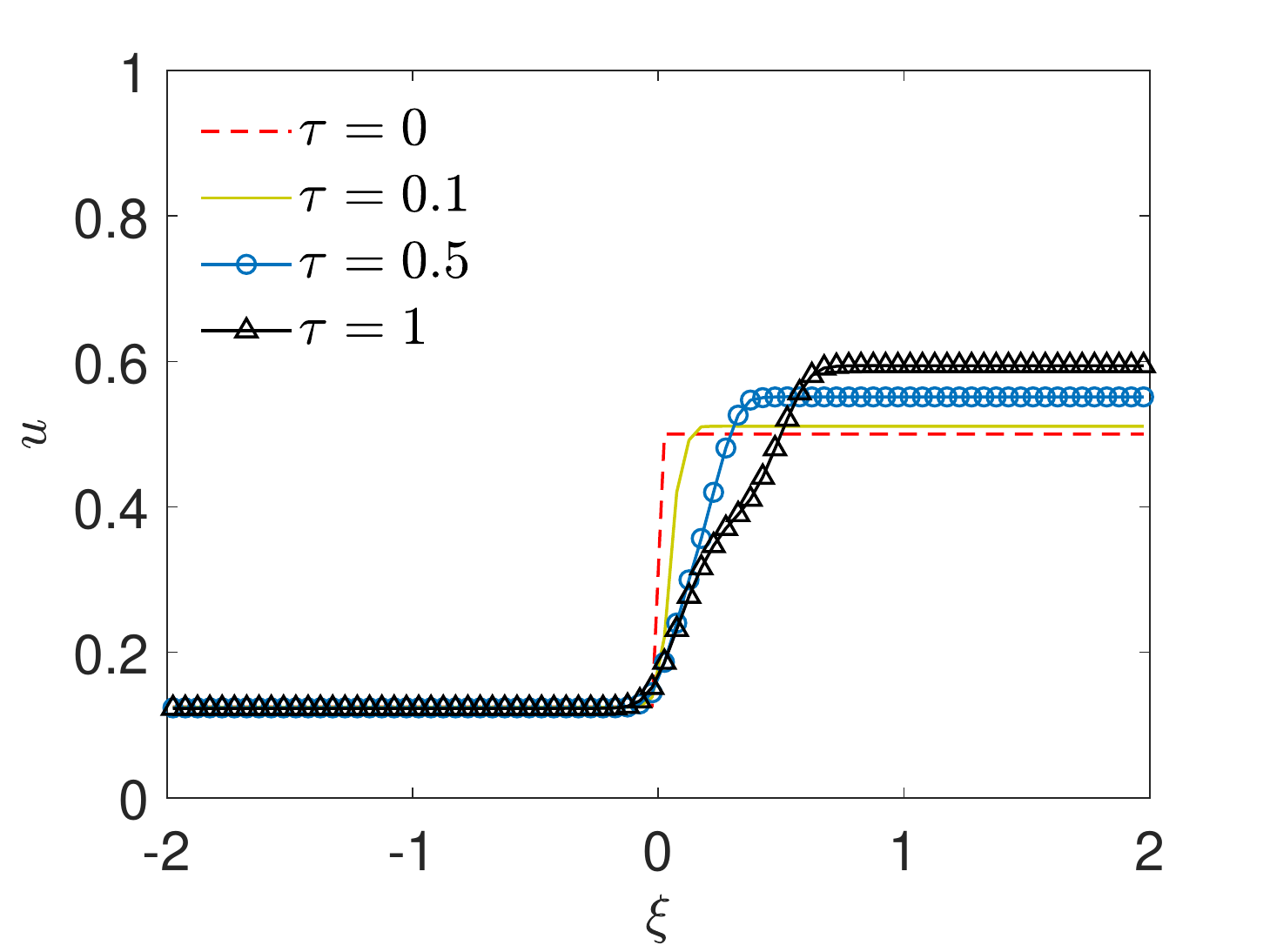}} \\
\subfigure[$\rho(\tau,\,\xi)$ with $\kappa=10^{-3}$]{\includegraphics[scale=0.4]{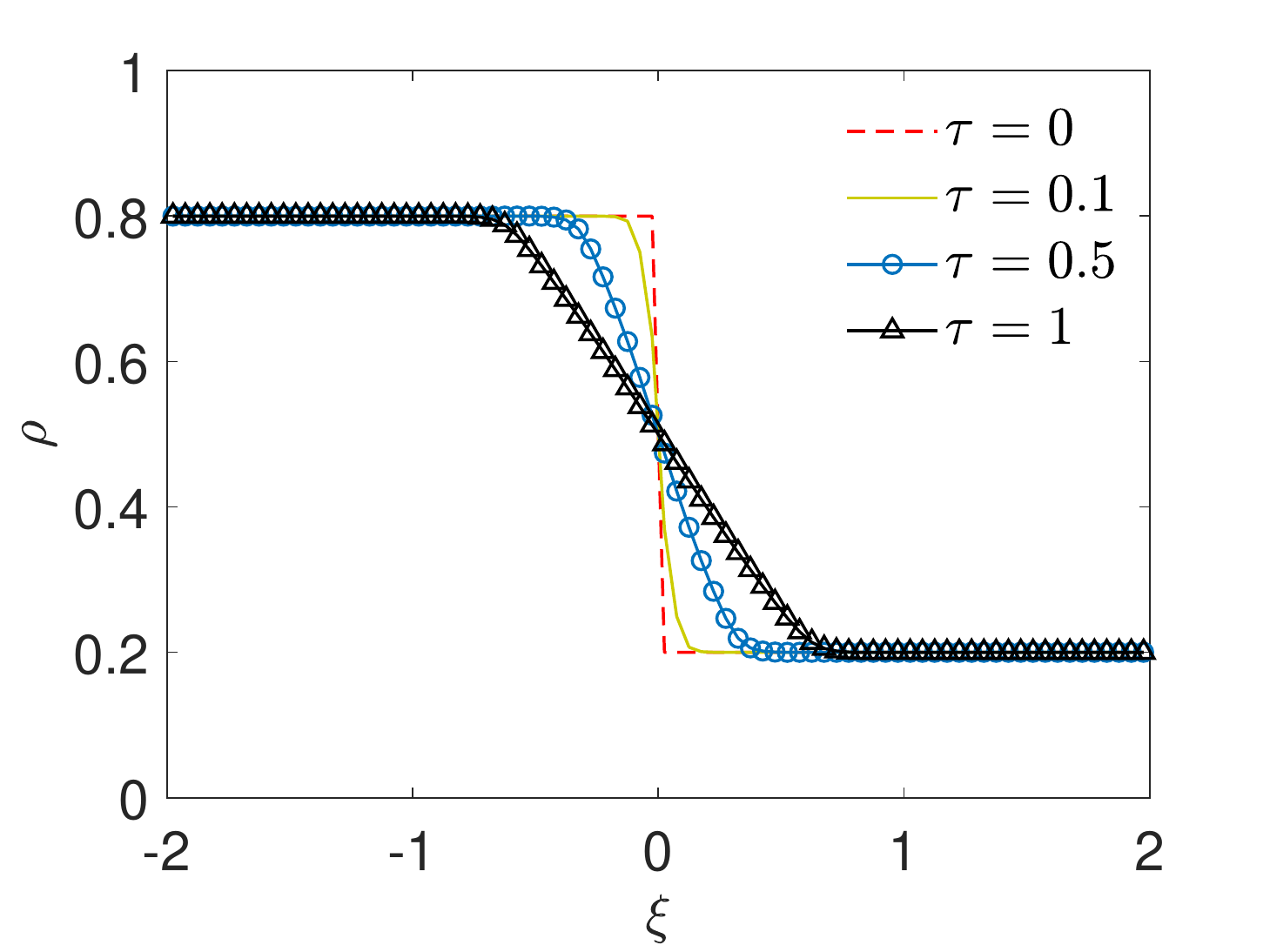}}
\subfigure[$u(\tau,\,\xi)$ with $\kappa=10^{-3}$]{\includegraphics[scale=0.4]{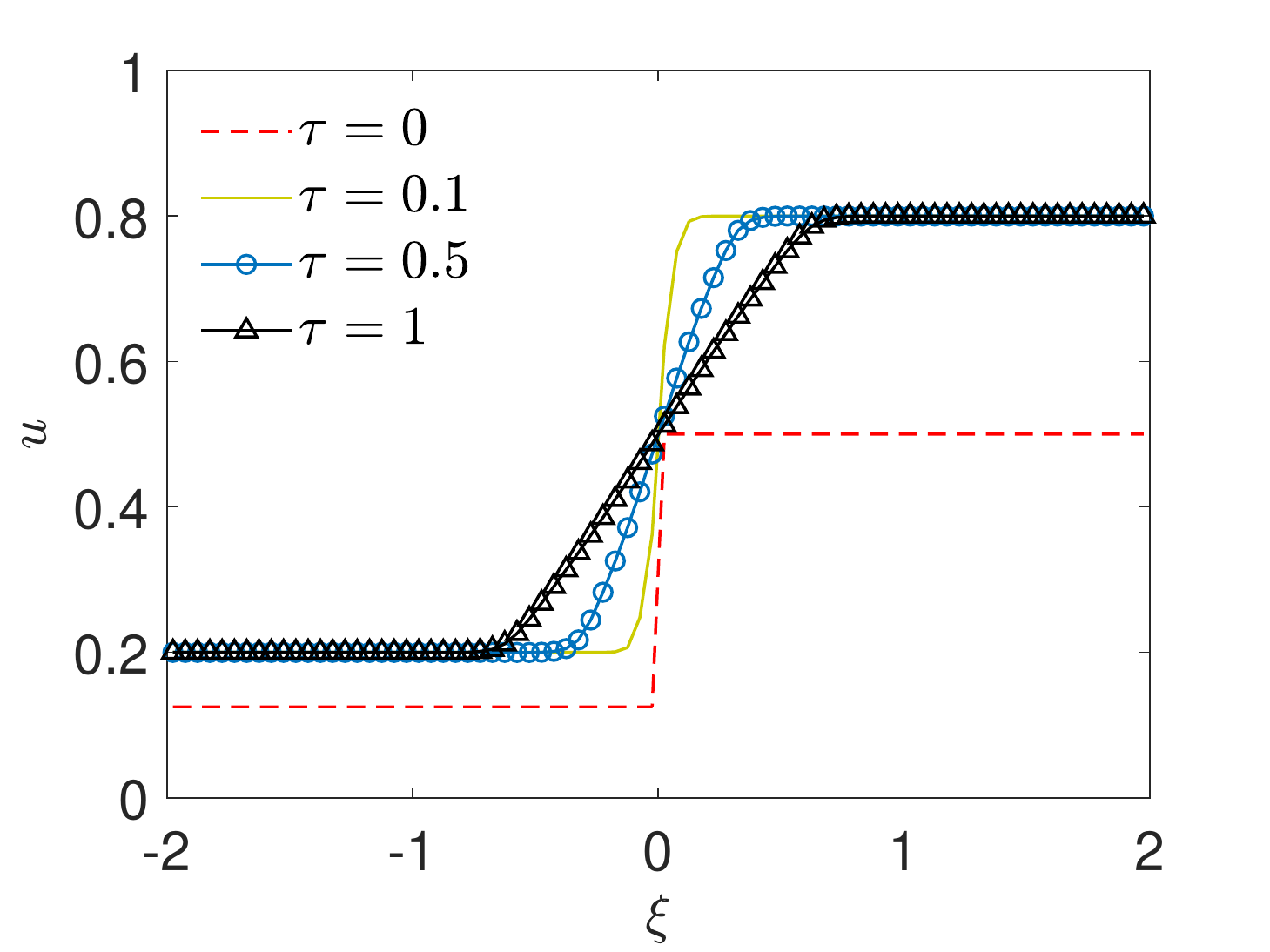}}
\caption{\textbf{Desired speed control}. Solution to~\eqref{eq:macro_second_vdes} at times $\tau=0$ (initial condition) and $\tau=0.1,\,0.5,\,1$ with penetration rate $p=1$ and control penalisation $\kappa=1$ (weak control, top row) and $\kappa=10^{-3}$ (strong control, bottom row).}
\label{fig:mono_2}
\end{figure}

We consider as initial condition the following density-speed pair:
\begin{equation*}
	\rho_0(\xi)=
	\begin{cases}
		0.8 & \text{for } \xi\leq 0 \\
		0.2 & \text{for } \xi>0,
	\end{cases}
	\qquad
	u_0(\xi)=
	\begin{cases}
		0.125 & \text{for } \xi\leq 0 \\
		0.5 & \text{for } \xi>0,
	\end{cases}
\end{equation*}
which mimics the fact that more densely packed vehicles are slower on average than less densely packed ones. Figure~\ref{fig:mono_1} shows the evolution of $\rho$ and $u$ predicted by model~\eqref{eq:macro_second_binvar}, which, as already observed in Section~\ref{sect:monokinetic}, represents simultaneously the hydrodynamic limit in the case of no control of any vehicle in the traffic stream and of binary variance control with arbitrary effective penetration rate $p^\ast>0$. We observe that vacuum tends naturally to form (see the vehicle density in the left panel) as expected from pressureless dynamics with fast vehicles preceding slow ones.

The macroscopic action of the control is instead clearly visible in Figure~\ref{fig:mono_2}. There we display the evolution of $\rho$ and $u$ predicted by model~\eqref{eq:macro_second_vdes}, which implements the microscopic desired speed control towards the recommended speed $v_d(\rho)=1-\rho$. We fix in particular $p=1$, corresponding to $100\%$ penetration of the driver-assist technology in the traffic stream, and we vary the control penalisation from $\kappa=1$ (strongly penalised control) to $\kappa=10^{-3}$ (weakly penalised control). The numerical results show that in the first case vacuum still tends to form (Figure~\ref{fig:mono_2}a) as a consequence of the pressuless dynamics, although in the long run ($\tau=1$) the control slightly perturbs the speed profile (Figure~\ref{fig:mono_2}b) with respect to the case illustrated in Figure~\ref{fig:mono_1}. In the second case, instead, the stronger action of the control dominates the speed dynamics, which, according to the second equation in~\eqref{eq:macro_second_vdes}, become essentially a quick local relaxation of $u$ towards $v_d(\rho)$ (cf. also Figure~\ref{fig:mono_2}d). The evolution of the corresponding density profile follows very closely a pure rarefaction wave between the left state $\rho_L=0.8$ and the right state $\rho_R=0.2$ (Figure~\ref{fig:mono_2}c), which is actually the expected solution to the \textit{first order} hydrodynamic model with flux $\F(\rho)=\rho v_d(\rho)$. On the whole, then, we find that if the control is sufficiently strong, namely if $p^\ast=\frac{p}{\kappa}\gg 1$, the solution to the second order hydrodynamic model~\eqref{eq:macro_second_vdes} collapses onto that of the first order model~\eqref{eq:macro_vdes}, which remarkably implies no more vacuum formation (compare Figures~\ref{fig:mono_2}a, c).

\section{Summary and conclusions}
\label{sect:conclusions}
In this paper we have developed a hierarchical approach to the control of traffic flow based on the nowadays increasingly popular idea that automated vehicles can be profitably used as inner controllers in a bottom-up control perspective. The general goal is to regularise the stream of vehicles from the inside; in particular, in this work we have considered control actions aimed at the mitigation of the road risk. First we have proposed a model of stochastic microscopic binary interactions among the vehicles, which include probabilistically the presence of driver-assist vehicles in the traffic flow. Such interactions produce speed variations through accelerations and decelerations but when they involve a driver-assist vehicle they are further controlled in such a way that the speed variance of the interacting vehicles is reduced. In fact reports of the World Health Organisation~\cite{peden2004WHO,WHO2015report} have stressed that speed differences from vehicle to vehicle are among the major causes of increased levels of crash risk. It is worth noticing that our probabilistic approach easily allows us to address both sparse and non-sparse control problems depending on the percentage of driver-assist vehicles in the traffic stream, i.e. the so-called penetration rate. Then we have upscaled these interactions to the level of the global flow by means of a space homogeneous kinetic Boltzmann-type equation. The analysis of such an equation, in particular of its asymptotic solutions, has provided us with detailed insights into the impact of the microscopic control strategies on the observable aggregate behaviour of the system. Interestingly, the results have revealed that some control strategies successfully reduce the speed-dependent road risk factors although they do not modify the macroscopic flow. From a purely macroscopic point of view they may therefore erroneously seem to be uninfluential. At last, we have reformulated the kinetic equation in a space inhomogeneous setting and we have used it to derive first and second order hydrodynamic traffic models consistent with the original microscopic controlled interactions among the vehicles. To this purpose we have taken advantage of closure methods which rely strongly on the ability of our kinetic model to provide explicit information on the speed distribution at equilibrium (the equivalent of the Maxwellian distribution in classical gas dynamics). The resulting equations for the density and the mean speed of the vehicles constitute original macroscopic traffic models, in which the action of the control is directly embedded from the microscopic scale (bottom-up) rather than being imposed through a control problem of the macroscopic equations themselves (top-down).

\section*{Acknowledgements}
This work has been written within the activities of the Excellence Project \textrm{CUP: E11G18000350001} of the Department of Mathematical Sciences "G. L. Lagrange" of Politecnico di Torino funded by MIUR (Italian Ministry for Education, University and Research), and within the activities of GNFM (Gruppo Nazionale per la Fisica Matematica) and GNCS (Gruppo Nazionale per il Calcolo Scientifico) of INdAM (Istituto Nazionale di Alta Matematica), Italy. M.Z. acknowledges support from ``Compagnia di San Paolo'' (Torino, Italy).

%\bibliographystyle{plain}
%\bibliography{TaZm-hydro_traffic}

\end{document}